\documentclass[onecolumn, 10 pt, journal]{IEEEtran}
\usepackage[utf8]{inputenc}
\usepackage{tikz} 
\usepackage{hyperref}
\usetikzlibrary{automata,arrows,shapes,snakes,automata,backgrounds,petri}
\usepackage{times}
\usepackage{caption}
\usepackage{setspace}
\usepackage{cite}
\usepackage{graphicx}
\usepackage{subcaption}
\usepackage{amsmath}
\usepackage{amssymb}
\usepackage{amsthm}
\usepackage{algorithm}
\usepackage{algorithmic}
\usepackage{booktabs}
\usepackage{multirow}
\usepackage{color}
\usepackage{paralist}
\usepackage{xiaoli}
\usepackage{wrapfig}
\newcommand{\vecx}{\mathbf{x}}

\newcommand{\vecw}{\mathbf{w}}
\newcommand{\poly}{\text{poly}}
\newcommand{\mat}[1]{\mathbf{#1}}
\newcommand{\vect}[1]{\mathbf{#1}}

\newcommand{\ABS}[1]{\left|#1\right|}
\newcommand{\ABSL}[1]{|#1|}
\newcommand{\SGN}[1]{\mathsf{sgn}[#1]}
\newcommand{\TWONL}[1]{\|#1\|_2}
\newcommand{\R}{\mathbb{R}} 
 
\begin{document}
\title{Sub-linear Time Support Recovery for Compressed Sensing using Sparse-Graph Codes}
  
\author{Xiao~Li,
	Dong~Yin,
        Sameer~Pawar,
        Ramtin~Pedarsani,
        and~Kannan~Ramchandran% <-this % stops a space
\thanks{X. Li is with Cubist Systematic Strategies. Email: xiaoli@eecs.berkeley.edu. Part of this work was done when X. Li was a postdoc at UC Berkeley.}
\thanks{D. Yin and K. Ramchandran are with the Department
of EECS at UC Berkeley. Email: \{dongyin, kannanr\}@eecs.berkeley.edu.}% <-this % stops a space
\thanks{S. Pawar is with Intel Corporation. Email: sameeronnet@gmail.com. Part of this work was done when S. Pawar was a graduate student at UC Berkeley.}% <-this % stops a space
\thanks{R. Pedarsani is with the Department of ECE at UC Santa Barbara. Email: ramtin@ece.ucsb.edu.}% <-this % stops a space
\thanks{This work was supported by grants NSF CCF EAGER 1439725, and NSF CCF 1116404 and MURI CHASE Grant No. 556016.}% <-this % stops a space
\thanks{Several parts of this paper were presented in 2015 IEEE International Symposium on Information Theory (ISIT)~\cite{li2015sub}, 2016 IEEE International Conference on Acoustics, Speech and Signal Processing (ICASSP)~\cite{li2016recovering}, and 2016 54th Annual Allerton Conference on Communication, Control, and Computing~\cite{yin2016compressed}.}}

\maketitle

\begin{abstract}
We study the support recovery problem for compressed sensing, where the goal is to reconstruct the sparsity pattern of a high-dimensional $K$-sparse signal $\mathbf{x}\in\mathbb{R}^N$, as well as the corresponding sparse coefficients, from low-dimensional linear measurements with and without noise. Our key contribution is a new compressed sensing framework through a new family of carefully designed {\it sparse measurement matrices} associated with minimal measurement costs and a low-complexity recovery algorithm. Specifically, the measurement matrix in our framework is designed based on the well-crafted {\it sparsification} through capacity-approaching {\it sparse-graph codes}, where the sparse coefficients can be recovered efficiently in a few iterations by performing simple error decoding over the observations. We formally connect this general recovery problem with sparse-graph decoding in packet communication systems, and analyze our framework in terms of the measurement cost, computational complexity and recovery performance. Specifically, we show that in the noiseless setting, our framework can recover any arbitrary $K$-sparse signal in $O(K)$ time using $2K$ measurements asymptotically with a {\it vanishing error probability}. In the noisy setting, when the sparse coefficients take values in a finite and quantized alphabet, our framework can achieve the same goal in time $O(K\log(N/K))$ using $O(K\log(N/K))$ measurements obtained from measurement matrix with elements $\{-1,0,1\}$. When the sparsity $K$ is {\it sub-linear} in the signal dimension $K=O(N^\delta)$ for some $0<\delta<1$, our results are order-optimal in terms of measurement costs and run-time, both of which are sub-linear in the signal dimension $N$. The sub-linear measurement cost and run-time can also be achieved with continuous-valued sparse coefficients, with a slight increment in the logarithmic factors. More specifically, in the continuous alphabet setting, when $K=O(N^\delta)$ and the magnitudes of all the sparse coefficients are bounded below by a positive constant, our algorithm can recover an arbitrarily large $(1-p)$-fraction of the support of the sparse signal using $O(K\log(N/K)\log\log(N/K))$ measurements, and $O(K\log^{1+r}(N/K))$ run-time, where $r$ is an arbitrarily small constant. For each recovered sparse coefficient, we can achieve $O(\epsilon)$ error for an arbitrarily small constant $\epsilon$. In addition, if the magnitudes of all the sparse coefficients are upper bounded by $O(K^c)$ for some constant $c<1$, then we are able to provide a strong $\ell_1$ recovery guarantee for the estimated signal $\widehat{\vecx}$: $\|\widehat{\vecx} - \vecx\|_1 \le \kappa \|\vecx\|_1$, where the constant $\kappa$ can be arbitrarily small. This offers the desired scalability of our framework that can potentially enable real-time or near-real-time processing for massive datasets featuring sparsity, which are relevant to a multitude of practical applications. 
\end{abstract}

\IEEEpeerreviewmaketitle 

%\graphicspath{{./figure/}}
 
\section{Introduction}

A classic problem of interest is that of estimating an unknown vector $ \mathbf{x}$ of length $N$ from noisy observations
\begin{align}\label{formulation_koisy}
	\mathbf{y} = \mathbf{A} \mathbf{x} + \mathbf{w},
\end{align}
where $\mathbf{A}$ is an $M\times N$ known matrix typically referred to as the {\it measurement matrix} and $\mathbf{w}$ is an additive noise vector. We refer to $N$ as the {\it signal dimension}. In general, if $\mathbf{x}$ has no additional structure, it is impossible to recover $\mathbf{x}$ from fewer measurements than the signal dimension. However, if the signal is known to be sparse with respect to some basis, wherein only $K$ coefficients are non-zero or significant with $K\ll N$, it is possible to recover the signal from much fewer measurements. This has been studied extensively in the literature under the name of {\it compressed sensing} \cite{donoho2006compressed}. The compressed sensing problem of reconstructing high-dimensional signals from lower dimensional observations arises in diverse fields, such as medical imaging \cite{lustig2007sparse}, optical imaging \cite{candes2013phase}, speech and image processing \cite{elad2010sparse}, data streaming and sketching \cite{gilbert2010sparse}, etc. 

A large variety of measurement designs and reconstruction algorithms have been proposed in the literature to exploit the inherent sparsity of signals to recover them from low-dimensional linear measurements. Clearly, the design of good measurement matrices and efficient reconstruction algorithms are critical (see Section \ref{sec:summary_results} for a brief review of existing methods). The key to achieve this goal boils down to two questions of interest:
\begin{itemize}
	\item[\bf Q1)] Measurement cost: what is the minimum number of measurements $M$ required to guarantee recovery?
	\item[\bf Q2)] Computational cost: how fast can one reconstruct the signal given $M$ measurements from some $\mathbf{A}$?
\end{itemize}
The answer to {\bf Q1} is well understood under information-theoretic settings (e.g. \cite{gastpar2000necessary,wainwright2009information,akccakaya2010shannon}). In the presence of noise, the predominant result indicates a minimum measurement cost of ${O}(K\log (N/K))$ for exact support recovery, here referred to as the {\it order-optimal scaling}. For {\bf Q2}, it is desirable if the computational complexity scales linearly with the measurement cost ${O}(K\log (N/K))$. However, there are no existing schemes that achieve ${O}(K\log (N/K))$ costs in both measurements and run-time in the worst case. More specifically, in existing methods, for any fixed measurement matrix, one can always find a $K$-sparse signal such that the algorithm fails to recover the sparse coefficients using ${O}(K\log (N/K))$ measurements and run-time. To relax this worst-case assumption, an intriguing question is:
\begin{center}
\textit{``Under probabilistic settings, is it possible to achieve the order-optimal scaling in both the measurement cost and the computational run-time?''}
\end{center}
%We propose a novel compressed sensing framework based on sparse-graph codes, and show that {\it with high probability}, our framework guarantees successful recovery in time ${O}(K \log N)$ using $M={O}(K\log N)$ measurements. Under this sparsity regime, our scaling is order-optimal and the run-time becomes sub-linear in $N$. 
In this work, we answer this question in the affirmative under the \emph{sub-linear} sparsity regime $K={O}(N^\delta)$ for any constant $\delta\in(0, 1)$. To the best of our knowledge, this is the first constructive design for noisy compressed sensing that achieves the same order-optimal costs in both measurements and complexity under probabilistic guarantees. Meanwhile, we note that our algorithm also works in the linear sparsity regime where $K=O(N)$, with $O(K\log(N))$ costs in both measurements and run-time. In this regime, our algorithm brings new insights to the design of measurement matrix for compressed sensing, and the measurement cost and run-time are still order-optimal up to logarithmic factors.

\subsection{Design Philosophy}

%\begin{wrapfigure}{r}{0.5\textwidth}
%\begin{minipage}{0.5\textwidth}
\begin{figure}[t]
\begin{center}
\includegraphics[width=0.5\linewidth]{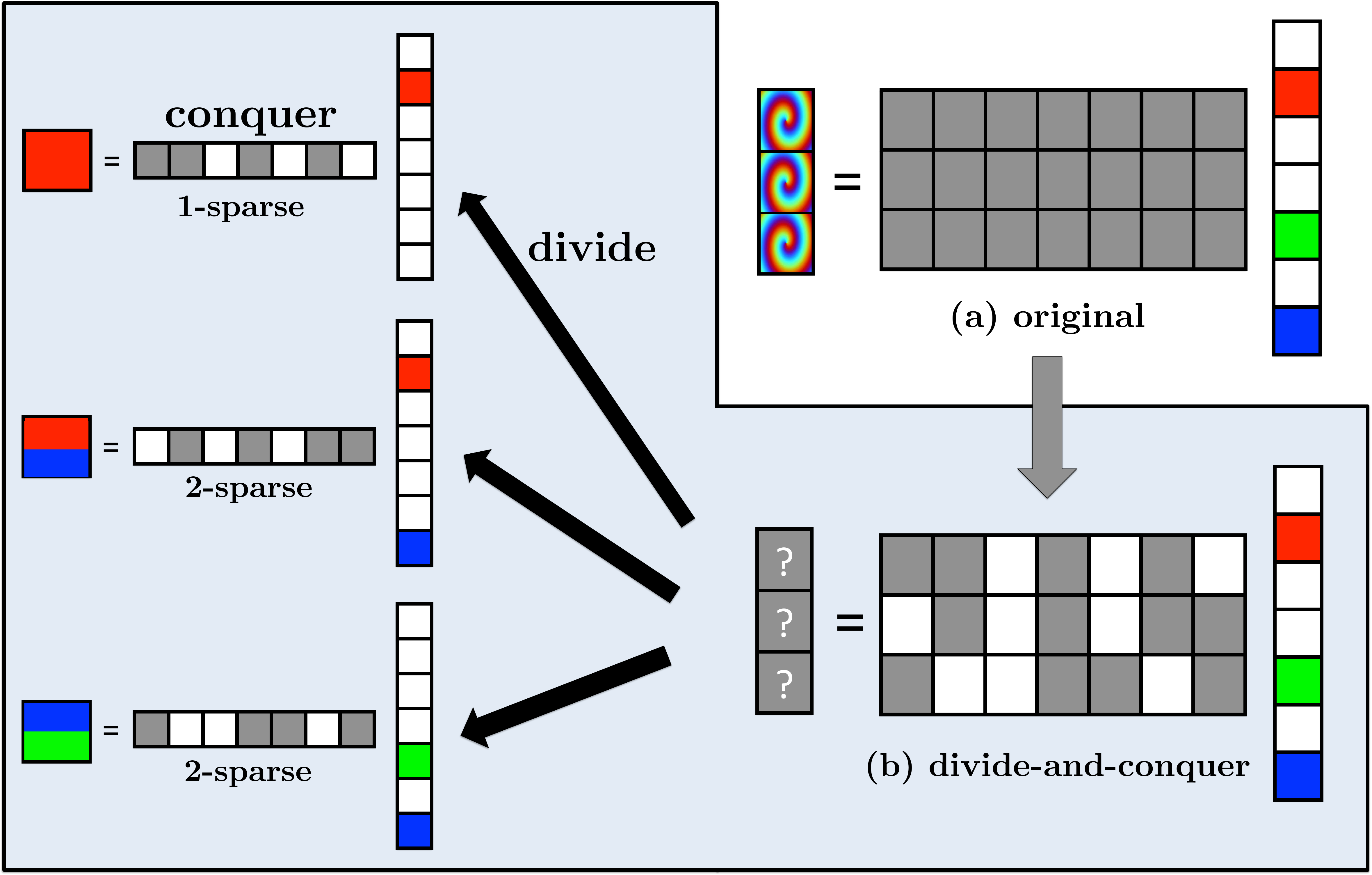}
\caption{\footnotesize A conceptual diagram of the ``divide-and-conquer'' philosophy used in our design. Zero entries are colored in {\it white} and the non-zero entries in the sparse vector are colored in {\it red}, {\it green} and {\it blue} respectively. We have a $3$-sparse recovery problem in sub-figure (a), where the measurement matrix is colored in {\it grey} to indicate an arbitrary design. The resulting measurements are colored as {\it mixtures} because of the arbitrary mixing of different color components (red, green, blue). In sub-figure (b), we sparsify the measurement matrix by placing three zeros in each row shown as the {\it white} spots. The resulting measurement matrix divides the $3$-sparse recovery problem into multiple sub-problems, where one of the sub-problems involves only one color that can be easily identified. In this example, the first measurement contains a single {\it red} color, whereas the second and third measurements contain a {\it mixture or red and blue} and a {\it mixture of blue and green} respectively. If the decoder knows that the first measurement contains a single {\it red} color, it can peel off its contribution from the {\it mixture of red and blue} in the second measurement, which forms a new measurement containing a single {\it blue} color. %Similarly, the blue color can be peeled off from the third measurement such that the {\it green} color is decoded. 
}\label{fig:idea_diagram}
\end{center}
\end{figure}
%\end{minipage} 
%\end{wrapfigure}
 
We take a simple but powerful ``divide-and-conquer'' approach to the problem by viewing compressed sensing through a ``sparse-graph coding'' lens. Our design philosophy is depicted in \figref{fig:idea_diagram} as a cartoon illustration, where we use different colors to distinguish the entries in the sparse vector, namely, we choose {\it red, green} and {\it blue} respectively for the non-zero entries, and {\it white} for zero entries. A conventional design in compressed sensing is to generate weighted linear measurements of the sparse vector through a carefully designed {\it measurement matrix} \cite{baraniuk2007compressive}. In this example, all the entries of the measurement matrix are colored in {\it grey} to indicate an arbitrary design and the corresponding measurements are some generic mixtures of {\it red, green} and {\it blue}, as shown in \figref{fig:idea_diagram}-(a). 
 
We design the measurement matrix by sparsifying each row of the measurement matrix with zero patterns guided by sparse-graph codes, indicated by the {\it white} spots in \figref{fig:idea_diagram}-(b). This new measurement matrix leads to a different set of measurements, where some contain single colors and some contain their mixtures. Our design philosophy is to {\it disperse} the signal into multiple single color measurements (e.g., the red color in the first measurement) and {\it peel} them off from color mixtures (e.g., the red-blue mixture in the second measurement and the blue-green mixture in the third measurement) to decode other unknown colors in the spirit of ``divide-and-conquer''. By analogy, the use of sparse-graph codes essentially {\it divides} the general sparse recovery problem into multiple sub-problems that can be easily {\it conquered} and synthesized for reconstructions. Furthermore, by viewing our design from a coding-theoretic lens, our design can further leverage the properties of sparse-graph codes in terms of both measurement cost (capacity-approaching) and computational complexity (fast peeling-based decoding). This leads to a new family of sparse measurement matrices simultaneously featuring low measurement costs and low computational costs.

\subsection{Objective}
We mainly focus on the recovery of the {\it exact support} of any $K$-sparse $N$-length signal and its sparse coefficients. This so-called {\it support recovery} problem arises in an array of applications such as model selection \cite{candes2009near}, sparse approximation \cite{fuchs2005recovery} and subset selection in regression problems \cite{greenshtein2006best}. Given $\widehat{\mathbf{x}}$ generated by some recovery method, a typical metric for {\it support recovery} is the error probability $\Pf$ of failing to recover the {\it exact support} of the signal:
\begin{align}\label{def_Pe}
	\Pf \defn \Prob{\supp{\widehat{\mathbf{x}}}\neq\supp{\mathbf{x}}},
\end{align}
where $\supp{\cdot}$ represents the support of some vector  $\supp{\mathbf{x}} \defn \left\{k:~x[k]\neq 0,~0\le k\le N-1\right\}$. The probability $\Pf$ is evaluated with respect to the randomness associated with the noise $\mathbf{w}$ {\it and} the measurement matrix $\mathbf{A}$. In other words, for any given $K$-sparse signal $\mathbf{x}$, our design generates a measurement matrix $\mathbf{A}$ (from a specific random ensemble\footnote{Note that this is what is known as the ``for-each'' guarantee \cite{gilbert2010sparse} in contrast to the ``for-all'' guarantee in some compressed sensing contributions, where a single measurement matrix is used for all sparse signals once generated.}) and produces an estimate $\widehat{\mathbf{x}}$ whose support matches {\it exactly} that of $\mathbf{x}$ with probability $1-\Pf$ approaching $1$ asymptotically in $K$ and $N$. In addition to support recovery, we also target accurate recovery of the sparse coefficients. In the noiseless setting and the noisy setting where the sparse coefficients take quantized values, we aim to recover the exact values of the sparse coefficients. In the continuous alphabet setting, we aim to get strong $\ell_\infty$ and $\ell_1$ norm recovery guarantees.

%Under ``for-all'' guarantees, the recovery algorithm achieves high probability of success using {\it one} measurement matrix for {\it all} signals, while under ``for-each'' guarantees, the recovery algorithm achieves high probability of success using a {\it random matrix} for {\it each} signal. 

\subsection{Contributions}\label{sec:contributions}
Our key contribution is the proposed new compressed sensing design framework for support recovery, with ${O}(K\log (N/K))$ costs for {\it both} measurements and run-time in the presence of noise. The measurement cost and computational complexity are obtained under the assumption that the sparse coefficients take values in a quantized alphabet, which can have arbitrarily fine but finite precision, and is practical in most cases of interest. Moreover, with a slight increment in the logarithmic factor, our results can be extended to the continuous alphabet setting, where we can obtain recovery guarantees in the $\ell_0$ and $\ell_1$ norms: for each recovered sparse coefficient, we can achieve $O(\epsilon)$ error for an arbitrarily small constant $\epsilon$; if the magnitudes of all the sparse coefficients are upper bounded by $O(K^c)$ for some constant $c<1$, then the estimated signal $\widehat{\vecx}$ satisfies $\|\widehat{\vecx} - \vecx\|_1 \le \kappa \|\vecx\|_1$, where the constant $\kappa$ can be arbitrarily small. In the noiseless setting, our measurement cost is can be reduced to $2K$ asymptotically, and run-time is reduced to $O(K)$ accordingly. \emph{When $K$ is sub-linear in $N$, and more specifically $K=O(N^\delta)$ for some $0<\delta<1$, our results are order-optimal and furthermore, sub-linear in the signal dimension $N$.} This offers the desired scalability of the algorithm that can potentially enable real-time or near-real-time processing for massive datasets featuring sparsity, which are relevant to a multitude of practical applications. Here, using the big-O notation\footnote{Recall that a single variable function $f(x)$ is said to be ${O}(g(x))$, if for a sufficiently large $x$ the function $|f(x)|$ is bounded above by $|g(x)|$, i.e., $\lim_{x\rightarrow\infty} |f(x)| < c|g(x)|$ for some constant $c$. Similarly, $f(x) = \Omega(g(x))$ if $\lim_{x\rightarrow\infty} |f(x)| > c|g(x)|$ and $f(x) = o(g(x))$ if the growth rate of $|f(x)|$ as $x\rightarrow\infty$, is negligible as compared to that of $|g(x)|$, i.e. $\lim_{x\rightarrow\infty} |f(x)|/|g(x)| = 0$.}, we briefly summarize our technical result as follows.

\begin{table}[h]
\begin{center} 
\begin{tabular}{|c|c|c|c|c|c|c|c|}
  \hline
   & Measurement  & Complexity & Recovery Guarantee\\
  \hline
  Noiseless & $2(1+\epsilon)K$ &   ${O}(K)$  & Support \& exact value\\
  \hline
  Noisy (quantized alphabet) & $O(K \log (N/K))$ &   ${O}(K\log (N/K))$  &  Support \& exact value\\
  \hline
  Noisy (continuous alphabet) & $O(K \log(N/K)\log\log(N/K))$ &   ${O}(K\log^{1+r} (N/K))$  &  Support \& $\ell_\infty$, $\ell_1$ norm bound\\
  \hline  
\end{tabular}
\end{center}
\caption{Measurement cost and complexity of our framework when $K=O(N^{\delta})$, $\delta\in(0,1)$ ($\epsilon>0$ and $r>0$ are arbitrarily small constants)}
\label{Table_Results}
\end{table}
Here, we also note that one can directly apply our algorithm in the linear sparsity setting, i.e., $K=O(N)$. In this scenario, the $\log(N/K)$ and $\log\log(N/K)$ factors in Table~\ref{Table_Results} are replaced with $\log(N)$ and $\log\log(N)$, respectively. Therefore, the measurement and time costs of our algorithm are still order-optimal up to logarithmic factors.

We now provide some intuition about our results. Recall that the idea is to use sparse-graph codes to structure the measurement matrix in order to generate different measurements containing isolated $1$-sparse coefficients, as well as their mixtures. From \figref{fig:idea_diagram}, these $1$-sparse coefficients (e.g., the red color in the first measurement) can be peeled off from their mixtures (e.g., the red and blue mixture in the second measurement), which forms new $1$-sparse coefficients for further peeling. This divide-and-conquer approach allows us to tackle a $K$-sparse recovery problem by solving a series of $1$-sparse problems of dimension $N$. Therefore, the challenge is to keep this peeling process going until all $1$-sparse components have been recovered. Hence we invoke sparse-graph codes principles to study this ``turbo'' peeling process theoretically to guarantee the success of decoding. As a result, we can focus on solving each $1$-sparse problem. Clearly, depending on the specific measurement matrix used, there are many ways to solve these $1$-sparse problems in $N$ dimension. 

In the noiseless setting, we choose the first two rows of the {\it Discrete Fourier Transform (DFT) matrix} as the measurement matrix before being sparsified by sparse-graph codes, and solve the $1$-sparse problem by leveraging spectral estimation techniques \cite{pawar2013computing}. We have two measurements to estimate the unknown index and the unknown value of the $1$-sparse coefficient, which is equivalent to estimating the frequency and amplitude of a complex discrete sinusoid from the DFT matrix. Therefore, in the noiseless setting, the frequency can be estimated by simply examining the relative phase between the two measurements, which only requires ${O}(1)$ measurements and computations. Then the unknown value of the coefficient can be obtained easily given the frequency.

To motivate our noisy result, we begin with another approach in the noiseless scenario by using a simple $\log_2 N \times N$ {\it binary indexing matrix}, which contains the binary index vector of each column included in the set of $N$ columns divided in the sub-problem. Using this measurement matrix, there are $\log_2 N$ measurements in each sub-problem. By taking the absolute values of the measurements, in the noiseless setting, we can directly obtain the signs of the measurements as the binary index of the $1$-sparse coefficient (assuming that the coefficient is positive\footnote{When the sign of the coefficient is unknown, we can use an extra row consisting of all one's to provide a reference sign.}). In fact, the signs of the measurements can be viewed as a length-$\log_2 N$ message bits for obtaining the unknown location of the $1$-sparse coefficient. Therefore in the noisy setting, according to the channel coding theorem, we can encode the binary indexing matrix using good channel codes with $N$ codewords of block length $O(\log_2 N)$ such that it can still be decoded correctly in the presence of noise with high probability. If the channel code has a {\it linear} decoding time in its block length ${O}(\log N)$, then we can achieve ${O}(\log N)$ costs for both measurements and computations for solving each $1$-sparse problem. Since $K=O(N^\delta)$, our results are order-optimal because $O(\log N) = O(\log (N/K))$, where the big-O constant changes according to $\delta$.

Finally, since there are in total $K$ sparse coefficients to estimate, the overall measurement and computational costs are further multiplied by a factor of $K$, which gives our result.

\subsection{Notation and Organization}

Throughout this paper, we use $\mathbb{R}$ and $\mathbb{C}$ to denote the real and complex fields. For any non-negative integer $n$, we denote by $[n]$ the set $\{0,1,\ldots, n-1\}$. Any boldface lowercase letter such as $\mathbf{x}\in\mathbb{C}^N$ represents a vector containing the complex elements\footnote{Here, we slightly abuse the symbol $[n]$. When attached to a lowercase letter, e.g., $x[n]$, $[n]$ represents the index of elements in a vector; otherwise, $[n]$ represents the set $\{0,\ldots, n-1\}$.} $\mathbf{x}=[x[0],\cdots,x[N-1]]^T$, and a boldface uppercase letter, such as $\mathbf{X}\in\mathbb{C}^{M\times N}$, represents a matrix with elements $X_{i,j}$ for $i\in[M]$ and $j\in[N]$. We denote the support of a vector $\mathbf{x}$ by $\supp{\mathbf{x}}$. For any subset $\Gamma$ of $[N]$, we define $\vecx_\Gamma$ as a vector with elements given by
$$
x_\Gamma[k] = \begin{cases}
x[k] & \text{ if } k\in\Gamma, \\
0 & \text{ otherwise.}
\end{cases}
$$
The inner product between two vectors is defined as $\ip{\mathbf{x}}{\mathbf{y}}=\sum_{k\in[N]}x[k](y[k])^\ast$ with arithmetic over $\mathbb{C}$. Let $\mathcal{A}$ be a set. We denote the cardinality of $\mathcal{A}$ by $|\mathcal{A}|$, and the complement of $\mathcal{A}$ by $\mathcal{A}^c$.

%The operator $\mathrm{supp}(\mathbf{x})=\{n:~x_k\neq 0, n=1,\cdots,N\}$ takes the support set of the vector $\mathbf{x}$ and $|\cdot|$ takes the cardinality of a certain set. The notation $\mathcal{X}_2$ refers to the finite field consisting of $\{0, 1\}$, with defined operations such as summation and multiplication modulo 2. Furthermore, we let $\mathcal{X}_2^n$ be the $n$-dimensional vector with each element from $\mathcal{X}_2$ and the addition of the vectors done element-wise over this field. The inner product of two binary indices $i$ and $j$ is defined by $\ip{i}{j}=\sum_{t=0}^{n-1} i_tj_t$ with arithmetic over $\mathcal{X}_2$, and 

%Given a real-valued vector $\mathbf{x}\in\mathbb{R}^N$ with $N = 2^n$, the $i$-th entry of $\mathbf{x}$ (i.e., $x[i]$) is interchangeably represented by $x[\bdsb{i}]$, where $\bdsb{i}=[i_{n-1},\cdots,i_0]^T\in\mathcal{X}_2^n$ is the index vector containing the binary representation of $i$, with $i_0$ and $i_{n-1} $ being the least significant bit (LSB) and the most significant bit (MSB), respectively. 
%The sign function is defined as $\mathrm{sign}(x)=1$ if $x>0$ and $-1$ if $x<0$, and $0$ if $x=0$. 
%Finally, given $\mathbf{S}\in\mathbb{R}^{M_2\times N}$ and $\mathbf{H}\defn [\mathbf{h}_1;\cdots;\mathbf{h}_{M_1}] \in\mathbb{R}^{M_1\times N}$ where $\mathbf{h}_m$ is the $m$-th row of the matrix, the row tensor product is defined as
%

This paper is organized as follows. We first summarize our main technical results in Section \ref{sec:summary_results}, followed by a brief overview of existing sparse recovery methods in Section~\ref{sec:related_work}. In Section \ref{sec:framework_overview}, for illustration purpose we provide a concrete example of our design framework using sparse-graph codes, followed by the analysis of the peeling decoder for sparse support recovery. Based on the example, we propose the principle and mathematical formulation of our measurement design in Section \ref{sec:meas_design}. We provide the general framework of the peeling decoding algorithm, and the density evolution analysis in Section~\ref{sec:analysis_peeling_decoder}. Then, we proceed to discuss specific constructions for our noiseless recovery results in Section \ref{sec:noiseless}, and further the noisy recovery results in the quantized alphabet and continuous alphabet settings in Section \ref{sec:noisy} and Section~\ref{sec:noisy_continuous}, respectively. We provide numerical results in Section \ref{sec:numerical} to corroborate our noisy recovery performance, and make conclusions in Section~\ref{sec:conclusion}.
 
%\begin{figure}[t]
%\centering
%\includegraphics[width=0.8\linewidth]{aa.pdf}
%\caption{Relationship between our results with existing literature in terms of design goals and philosophy highlighted in blue.}\label{fig:diagram_context}
%\end{figure}

\section{Main Results}\label{sec:summary_results}
In this section, we summarize the main results in this paper. We consider the problem of recovering the sparse\footnote{More generally, we also allow the signal to be sparse in any linear transform domain. If the signal is sparse in the transform domain, one can pre-multiply the measurement matrix $\mathbf{A}$ on right by the appropriate inverse transform.} signal $\mathbf{x}$ from the measurements
obtained in \eqref{formulation_koisy}. In particular, we are interested in support recovery for both the noiseless and noisy settings. Our design is characterized by the triplet $(M, T, \Pf)$, where $M$ is the measurement cost, $T$ is the computational complexity in terms of arithmetic operations, and $\Pf$ is the failure probability defined in \eqref{def_Pe}. 

\begin{thm}[\bf Noiseless Recovery]\label{thm_koiseless_recovery}
For any $\epsilon>0$, with probability at least $1-O(1/K)$, our framework can recover any $K$-sparse signal $\mathbf{x}$ in time $T=O(K)$ with $M=2(1+\epsilon)K$ measurements if $\mathbf{w}=\mathbf{0}$.
\end{thm}
Details of the noiseless recovery algorithm is provided in Section \ref{sec:noiseless}.

%When it comes to the noisy settings\footnote{Part of this result has been presented at ICASSP 2016, Shanghai.}, we further assume that all the non-zero coefficients belong to a set\footnote{This finite constellation assumption is imposed to simplify our analysis and its cardinality can be arbitrarily large but finite, which subsumes all practical digital signals that have been quantized with finite precision (essentially any signal processed by a digital computer).} $\mathcal{X}=\{Ae^{\mathrm{i}\theta}: A\in\mathcal{A},\theta\in\Theta\}$ where $\mathcal{A}\defn\{\Amin+\ell\rho\}_{\ell=0}^{L_1-1}$ for some arbitrarily large but finite $L_1>1$, $\rho>0$ and $\Amin>0$ while $\Theta\defn\{2\pi \ell/L_2\}_{\ell=0}^{L_2-1}$ for some arbitrarily large but finite $L_2>0$. Furthermore, our result can be further extended to the scenario when $x[k]$ take continuous values \cite{yin2013}. Denote the minimum signal-to-noise ratio (SNR) by $\SNRmin \defn \rho^2/\sigma^2$ (assuming $\rho<\Amin$ without loss of generality). 

When it comes to the noisy settings, we assume that the elements in the noise vector $\vecw$ are i.i.d. Gaussian distributed with mean $0$ and variance $\sigma^2$. We further consider two cases in the noisy setting: the quantized alphabet setting and the continuous alphabet setting. In the quantized alphabet setting, all the non-zero coefficients belong to a finite set $\mathcal{X}=\{\pm \rho, \pm 2\rho, \ldots, \pm B\rho\}$, and the minimum signal-to-noise ratio (SNR) is denoted by $\SNRmin \defn \rho^2/\sigma^2$. Our main result is as follows.

\begin{thm}[\bf Noisy Recovery, Quantized Alphabet]\label{thm_sub-linear_recovery}
Let $K=O(N^\delta)$ for some $\delta\in(0,1)$. With probability at least $1- O(1/K)$, our framework can recover any $K$-sparse signal with quantized alphabet $\mathcal{X}$ in time $T=O(K \log (N/K))$ with $M= O(K \log(N/K))$ measurements, where the big-O constant depends on $\SNRmin$ and the sparsity regime $\delta$.
\end{thm}
We provide the details in Section \ref{sec:noisy} and Appendix \ref{sec:noisy_recovery_perf_analysis}. In addition, when $K=O(N)$, the run-time and measurement cost become $M = O(K\log(N))$ and $T= O(K\log(N))$, respectively.

In the continuous alphabet setting, we assume that all the sparse coefficients have absolute values at least $\beta>0$, i.e., for any $k\in\supp{\vecx}$, we have $| x[k] | \ge \beta$. We provide the performance guarantee for recovering an arbitrarily large fraction of the support, as well as the $\ell_\infty$ and $\ell_1$ norm recovery guarantees. 
%In this scenario, when we aim to guarantee that the all the recovered elements satisfy $| \widehat{x}[k] - | \le O(\epsilon)$, then, the we define the minimum signal-to-noise ratio $\SNRmin$ as , where $\epsilon$

\begin{thm}[\bf Noisy Recovery, Continuous Alphabet]\label{thm_continuous_recovery}
Let $K=O(N^\delta)$ for some $\delta\in(0,1)$. Let $\Gamma$ be the support of $\vecx$, and $\widehat{\vecx}$ be the recovered signal with support $\widehat{\Gamma}$. Suppose that for some $\epsilon>0$, $\beta = \Omega(\max\{\epsilon, (\sigma+\epsilon)^2\})$, and that $\|\vecx\|_\infty \le O(K^c)$ for some constant $c\in(0,1)$. Then, using $M= O(K\log(N/K)\log\log(N/K))$ measurements, our algorithm satisfies: 
\begin{itemize}
\item $\widehat{\Gamma} \subset \Gamma$ (no false discovery)
\item $| \widehat{\Gamma} | \ge (1-p) K$, for arbitrarily small constant $p>0$ (recovering an arbitrarily large fraction of the support)
\item $\|\widehat{\vecx}_{\widehat{\Gamma}} - \vecx_{\widehat{\Gamma}} \|_\infty \le O(\epsilon)$ ($\ell_\infty$ norm recovery guarantee) 
%\item $\|\widehat{\vecx} - \vecx\|_1 \le (1+o(1))p\|\vecx\|_1 + O(K\epsilon)$ ($\ell_1$ norm recovery guarantee)
\item $\|\widehat{\vecx} - \vecx\|_1 \le \kappa \|\vecx\|_1$, for an arbitrarily small constant $\kappa > 0$ ($\ell_1$ norm recovery guarantee)
\end{itemize}
with probability at least $1-O(1/\poly(N))$. Further, our algorithm runs in time $T=O(K\log^{1+r}(N/K))$ with an arbitrary small constant $r>0$.
\end{thm}

The details of the continuous alphabet setting are provided in Section~\ref{sec:noisy_continuous}. Again, we mention that in the linear sparsity regime where $K=O(N)$, the measurement cost and run-time become $M= O(K\log(N)\log\log(N))$ and $T=O(K\log^{1+r}(N))$, respectively. In the following discussion, we focus on the sub-linear sparsity regime where $K=O(N^\delta)$. In the continuous alphabet setting, the definition of minimum signal-to-noise ratio $\SNRmin$ is changed to $\SNRmin := \frac{\epsilon^2}{\sigma^2}$, where $\epsilon$ is the accuracy in the $\ell_\infty$ norm in Theorem~\ref{thm_continuous_recovery}. In the $\ell_1$ recovery guarantee, the constant $\kappa$ depends on $\epsilon$, $\beta$, and $p$, and can be made arbitrarily small by tuning the design parameters in the algorithm. Here, since we focus on the regime where $K$ and $N$ approach infinity, we hide the dependence on $\epsilon$, $p$, $\delta$, $\SNRmin$ in the big-O notation in the measurement cost and run-time. As one can see, the continuous alphabet setting is more complicated than the quantized alphabet setting, and in Theorem~\ref{thm_continuous_recovery}, we only guarantee to recover an arbitrarily large fraction of the support of $\vecx$. However, recovering the full support is indeed possible by running the algorithm $O(\log K)$ times independently, and collecting all the recovered sparse coefficients. In this case, we can recover the full support with $M= O(K\log^2(N/K)\log\log(N/K))$ measurements and time $T=O(K\log^{2+r}(N/K))$.
%, and the $\ell_\infty$ norm and $\ell_1$ norm recovery guarantee becomes $\|\widehat{\vecx} - \vecx\|_\infty \le O(\epsilon)$ and $\|\widehat{\vecx} - \vecx\|_1 \le O(K\epsilon)$, respectively. 
Furthermore, the reason that the $\log\log(N/K)$ term appears in the measurement cost is that, we design a concatenated code in order to solve the $1$-sparse problem. We would like to mention that the use of this code is mainly for theoretical reason. Under a mild conjecture on the existence of a code with universal decoding algorithm and linear complexity, we can further eliminate the $\log\log(N/K)$ factor. With this conjecture, our measurement cost for large fraction recovery becomes $M= O(K\log(N/K))$ and computational complexity becomes $T = O(K\log(N/K))$; and the measurement cost and computational complexity for full support recovery become $M= O(K\log^2(N/K))$ and $T = O(K\log^2(N/K))$, respectively. For comparison, we list the results for the continuous alphabet setting in Table~\ref{tab:compare_continuous}.
%More details are relegated to Section~\ref{sec:noisy_continuous}.

\begin{table}[!h]
\label{tab:result}
\centering
\begin{tabular}{|c|c|c|}
\hline
Recovery & Measurement & Complexity\\
\hline
Large fraction & $ O(K\log(N/K)\log\log(N/K)) $ & $ O(K\log^{1+r}(N/K)) $ \\
\hline
Large fraction with conjecture & $O(K\log(N/K))$ & $O(K\log(N/K))$ \\
\hline
Full recovery & $O(K\log^2(N/K)\log\log(N/K))$ & $O(K\log^{2+r}(N/K))$ \\
\hline 
Full recovery with conjecture & $O(K\log^2(N/K))$ & $O(K\log^2(N/K))$ \\
\hline
\end{tabular}
\caption{Measurement cost and computational complexity in the continuous alphabet setting, $K=O(N^\delta)$}
\label{tab:compare_continuous}
\end{table}

%\begin{rmk}
%Since we assume an arbitrarily large but finite constellation $\mathcal{X}$ for each non-zero coefficient, we show that the $K$-sparse signal $\mathbf{x}$ can in fact be recovered perfectly, even from the noisy measurements with high probability. The recovery algorithm is equally applicable to signals with arbitrary coefficients over the complex field, but the analysis becomes overly cumbersome without offering more insights of our design. Hence we do not pursue it in this paper. However, the numerical tests in Section \ref{sec:numerical} present empirical evidence of our design succeeding with high probability in recovering the support of signals with arbitrary non-zero complex coefficients, where the $\ell_1$-approximation error of the sparse signal is also small.
%\end{rmk}
%

\section{Related Works}\label{sec:related_work}
In this section, we review the relevant works in the literature.
It is worth noting that only with a few exceptions, most of the existing compressed sensing and sparse recovery results have been predominantly developed for sparse approximation under the $\ell_2/\ell_1$-norm or $\ell_1/\ell_1$-norm approximation error metrics\footnote{$\ell_p/\ell_q$-norm guarantees refer to the error metrics measured with respect to the best $K$-term approximation error $\left\|\mathbf{x}_K-\mathbf{x}\right\|$ (i.e., the vector $\mathbf{x}_K$ is the best $K$-term approximation containing the $K$ most significant entries in the sparse vector $\mathbf{x}$), where the recovered sparse signal $\widehat{\mathbf{x}}$ satisfies $\left\|\widehat{\mathbf{x}}-\mathbf{x}\right\|_p\leq \kappa\left\|\mathbf{x}_K-\mathbf{x}\right\|_q$ for some absolute constant $\kappa>0$.}, with a relatively much lower coverage of support recovery  \cite{candes2006robust,candes2007dantzig,wu2012optimal,davenport2011introduction,gilbert2010sparse}. Meanwhile, necessary and sufficient conditions for support recovery have been studied in different regimes under various distortion measures using optimal decoders \cite{reeves2008sampling,gastpar2000necessary,wainwright2009information,aeron2010information,akccakaya2010shannon}, $\ell_1$-minimization methods \cite{wainwright2009sharp,candes2009near} and greedy methods \cite{cai2011orthogonal}. For example, it is shown in \cite{wainwright2009information} that ${O}(K\log(N/K))$ measurements are sufficient and necessary for support recovery when the measurement matrix consists of independent identically distributed (i.i.d.) Gaussian entries under Gaussian noise. Similar conditions under other signal and measurement models are also reported in \cite{fletcher2009necessary,wang2010information,jin2011limits}. Nonetheless, constructive recovery schemes that specifically target {\it support recovery} are relatively scarce \cite{hormati2009estimation,haupt2011robust,wang2010information}, especially those that come with order-optimal measurement costs and low computational complexities (see~\cite{sarvotham2006sudocodes,khajehnejad2011summary,gilbert2012approximate,gilbert2006algorithmic}). %To the best of our knowledge, this is the first result that achieves order-optimal measurement costs with order-optimal run-time as well. 
In the following, we categorize and briefly review the relevant works.

\subsection{Convex Relaxation Approach} The classic formulation for sparse recovery from linear measurements is through an $\ell_0$-norm minimization, which is a non-convex optimization problem. This problem has been known to be notoriously hard to solve. Convex optimization techniques relax the original combinatorial problem to a convex $\ell_1$-norm minimization problem, where computationally efficient algorithms are designed to solve this relaxed problem.
It has been shown that as long as the measurement matrices satisfy the Restricted Isometry Property (RIP) or mutual coherence (MC) conditions, the $\ell_1$-relaxation of the original problem has exactly the same sparse solution as the original combinatorial problem. This class of methods is known to provide a high level of robustness against the measurement noise, and furthermore, do not depend on the structure of measurement matrices. Popular algorithms in this class include LASSO \cite{tibshirani1996regression}, Iterative Hard Thresholding (IHT) \cite{blumensath2009iterative}, fast iterative shrinkage-thresholding algorithm (FISTA) \cite{beck2009fast}, message passing \cite{donoho2009message}, Dantzig selector \cite{candes2007dantzig} and so on. Most of the existing results along this line measurement matrices that are characterized by a measurement cost of ${O}(K\log(N/K))$ and a computational complexity $O(\mathrm{poly}(N))$. 

\subsection{Greedy Methods} Another class of methods, referred to as greedy iterative algorithms, attempts to solve the original $\ell_0$-minimization problem directly using successive approximations of the sparse signal through various heuristics. Examples include Orthogonal Matching Pursuit (OMP) \cite{tropp2007signal}, CoSaMP \cite{needell2009cosamp}, Regularized OMP (ROMP) \cite{needell2009uniform}, Stagewise OMP (StOMP) \cite{donoho2012sparse} and so on. Similar to convex relaxation approaches, this class also does not depend on the structure of the measurement. Although greedy algorithms are generally faster in practical implementations than the techniques based on convex relaxations, the common computational cost still scales as ${O}(\mathrm{poly}(N))$ for both noiseless and noisy settings, with a few exceptions that incur near-linear run-time ${O}(N\log N)$ (e.g., StOMP algorithm \cite{donoho2012sparse}). Besides, the measurement matrix is typically stated in terms of MC conditions\footnote{The measurement scaling of ${O}(K\log (N/K))$ for greedy pursuit methods exists under relaxed settings (e.g. bounded noise scenarios or probabilistic guarantees \cite{tropp2005signal}). While there are some results on OMP based on the RIP, it is still ongoing work (see \cite{davenport2011introduction}).} which require ${O}(K^2)$ measurements. This phenomenon is commonly referred to as the square-root bottleneck, where the limit of sparsity for successful recovery is on the order of $K={O}(\sqrt{N})$ even if measurement matrices achieving the MC lower bound are used (i.e. the Welch bound \cite{welch1974lower}). 

\subsection{Coding-theoretic Approach} This class of methods borrows the insights from modern coding theory to facilitate measurement designs and recovery algorithms. Compressed sensing measurement designs have been extensively studied from a coding-theoretic lens. For instance, \cite{howard2008fast,applebaum2009chirp} exploit the algebraic properties of Reed-Muller codes and Delsarte Goethals codes, \cite{akccakaya2008frame} uses a generalization of Reed-Solomon codes, and \cite{dimakis2012ldpc} establishes the connection between the channel decoding problem and the convex relaxation approach. Meanwhile, a multitude of work has emerged based on {\it expander graphs} \cite{xu2007efficient,jafarpour2009efficient}, a popular design element in modern coding theory, which achieves near-linear time\footnote{Using the same measurement design based on expanders, $\ell_1$-minimziation can also be shown to achieve similar performance in polynomial time\cite{berinde2008combining}.} recovery ${O}(N\log (N/K))$ using ${O}(K\log (N/K))$ measurements in the noiseless setting. Motivated by expander-based designs, researchers have proposed greedy approximation schemes that achieve similar costs, such as Expander Matching Pursuit (EMP) \cite{indyk2008near} and Sparse Matching Pursuit (SMP) \cite{berinde2008practical}. Last but not least, there is a wide range of recovery algorithms using modern decoding principles such as list decoding \cite{parvaresh2008explicit,pham2009sublinear}, efficient error-correcting codes via message passing \cite{sarvotham2006sudocodes,bakshi2012sho,zhang2008compressed}. Recently, \cite{donoho2012information}
uses spatially-coupled LDPC codes in the measurement design and an approximate message passing decoding algorithm for recovery, which achieves the information-theoretically optimal measurement cost ${O}(K)$ given by \cite{wu2012optimal} under a source coding setting. However, the decoding complexity remains polynomial time in $N$. Particularly relevant to our work are those based on fast verification-based decoding  \cite{zhang2012verification,sarvotham2006sudocodes,finiasz2012private}, where the sparse coefficients are solved by verifying and correcting each symbol iteratively. The Sudocodes design \cite{sarvotham2006sudocodes} introduces a noiseless scheme with ${O}(K\log N)$ measurements and sub-linear time computations ${O}(K\log K \log N)$ through a two-part verification decoding procedure. Further, \cite{zhang2012verification} proposes a general high rate LDPC design with applications in compressed sensing, which provably provides guarantees for a broad class of measurement matrices under verification-based decoding, where the Sudocodes \cite{sarvotham2006sudocodes} is mentioned as a special case therein. Further, \cite{khajehnejad2011summary} proposed an algorithm that achieves a sample complexity of ${O}(K\log N \log\log N)$ and run-time ${O}(\mathrm{poly}(K\log N))$ using a well-designed measurement matrix based on the proposed ``summary-based'' structure. Although our design shares certain elements in terms of the code properties being used, our approach differs significantly in designing the verification decoding schemes to achieve {\it sub-linear} time both {\it in the absence} and {\it presence} of noise, as well as the associated performance analysis. 

\subsection{Group Testing and Data Stream Computing}
This class of methods exploit linear ``sketches'' of data for sparsity pattern recovery in {\it group testing}~\cite{du1993combinatorial} and {\it data stream computing} \cite{charikar2004finding}. 
The major difference in this class of methods is that it mostly deals with noiseless measurements and that the measurement matrix can be freely designed to facilitate recovery. 
In group testing, the common scenario is that we need to devise a collection of tests to find $K$ anomalous items from $N$ total items, where the typical goal is to recover the {\it support} of the underlying sparse vector and minimize the number of tests performed (measurements taken) \cite{indyk2010efficiently}. In particular, \cite{cormode2006combinatorial} develops a compressed sensing design using group testing principle with ${O}(K\log^2N)$ measurements and ${O}(K\log^2N)$ operations. On the other hand, the goal of data stream computing is to maintain a short linear sketch of the network flows for approximating the sparse vector with some distortion measure. Examples include the count-min/count-sketch methods \cite{haupt2011robust} and so on. Typical results in this bulk of literature require ${O}(K\log (N/K))$ measurements and near-linear time ${O}(N\log N)$ (see \cite{gilbert2010sparse}). While there is a subset of sketching algorithms that achieve sub-linear time with ${O}(K\log(N/K))$ and ${O}(K\log^{{O}(1)} N)$ operations \cite{gilbert2007one,gilbert2012approximate,gilbert2006algorithmic}, these results typically provide {\it constant} failure probability guarantees for noiseless\footnote{Although sketching algorithms are not derived specifically to address noisy measurements, they could potentially be quite robust to various forms of noise.} measurements and sparse approximation instead of support recovery. % and $\ell_2/\ell_1$ or $\ell_1/\ell_1$ approximations instead of support recovery. %In contrast, our design addresses the exact recovery of the sparse support in both noiseless and noisy settings with sub-linear costs. Last but not least, although our measurement costs and computational complexity are stated in an order sense, the constants in the big-O notations are demonstrated to be very small in simulations.
%
%
 
%\todo{Do we need to add comparison with our results?}
%\subsection{Comparisons with Our Results}
%It can be seen that the computational complexities of all existing results (except for the group testing and sketching algorithms) scale at least linearly with the signal dimension $N$, whereas our design achieves sub-linear time recovery for both noiseless and noisy settings. This is because these methods generally do not depend on the structure of the measurement matrix as long as certain RIP and MC conditions are met, which incurs the complexity of touching all the entries in $\mathbf{x}$. Our design accommodates a new family of measurement matrices using sparse-graph codes, whose structure can be efficiently exploited in recovery. This idea is vastly different from convex relaxation approaches and greedy methods, but 

\section{Main Idea of Compressed Sensing using Sparse-Graph Codes}\label{sec:framework_overview}

%With the measurement design set up, we connect the sparse recovery problem in \eqref{formulation_koisy} to the decoding of sparse-graph codes, which forms the basis of our recovery algorithm. 

%We first interpret how such design can be cast as an instance of sparse-graph decoding and then propose our recovery algorithm and analyze its performance.

In this section, we present our design philosophy depicted in \figref{fig:idea_diagram} with more details, and describe the main idea of our measurement design and recovery algorithm through a simple example in the noiseless setting. 
We illustrate the principle of our recovery algorithm by connecting support recovery with sparse-graph decoding using an ``oracle'' (described below). Then, using the insights gathered from the oracle-based decoding algorithm, we explain how we can get rid of the ``oracle'' using the same example.

\subsection{Oracle-based Sparse-Graph Decoding}\label{sec:simple_example}

Consider a simple illustration consisting of a sparse signal $\mathbf{x}$ of length $N=16$ with $K=5$ non-zero coefficients $x[1]=1$, $x[3]=4$, $x[5]=2$, $x[10]=3$ and $x[13]=7$. To illustrate the principle of our recovery algorithm, we construct a bipartite graph with $16$ left nodes and $9$ right nodes. The graph has the following properties:
\begin{itemize}
	\item Each {\it left node} labeled with $k$ is assigned a value $x[k]$ for $k\in[N]$;
	\item Each {\it left node} is connected to the {\it right nodes} according to the  {\it sparse} bipartite graph\footnote{Since the values of the right nodes are not affected by the left nodes carrying zero coefficients, we show only the edges from the left nodes with non-zero values $x[k]\neq 0$.} in \figref{fig:example_bipartite};
	\item Each {\it right node} labeled with $r$ is assigned a value $y_r$ equal to the complex sum of its left neighbors, similar to the parity-check constraints of the LDPC codes. 
\end{itemize}

\begin{figure}[t]
\begin{center}
\centering
\includegraphics[width=0.3\linewidth]{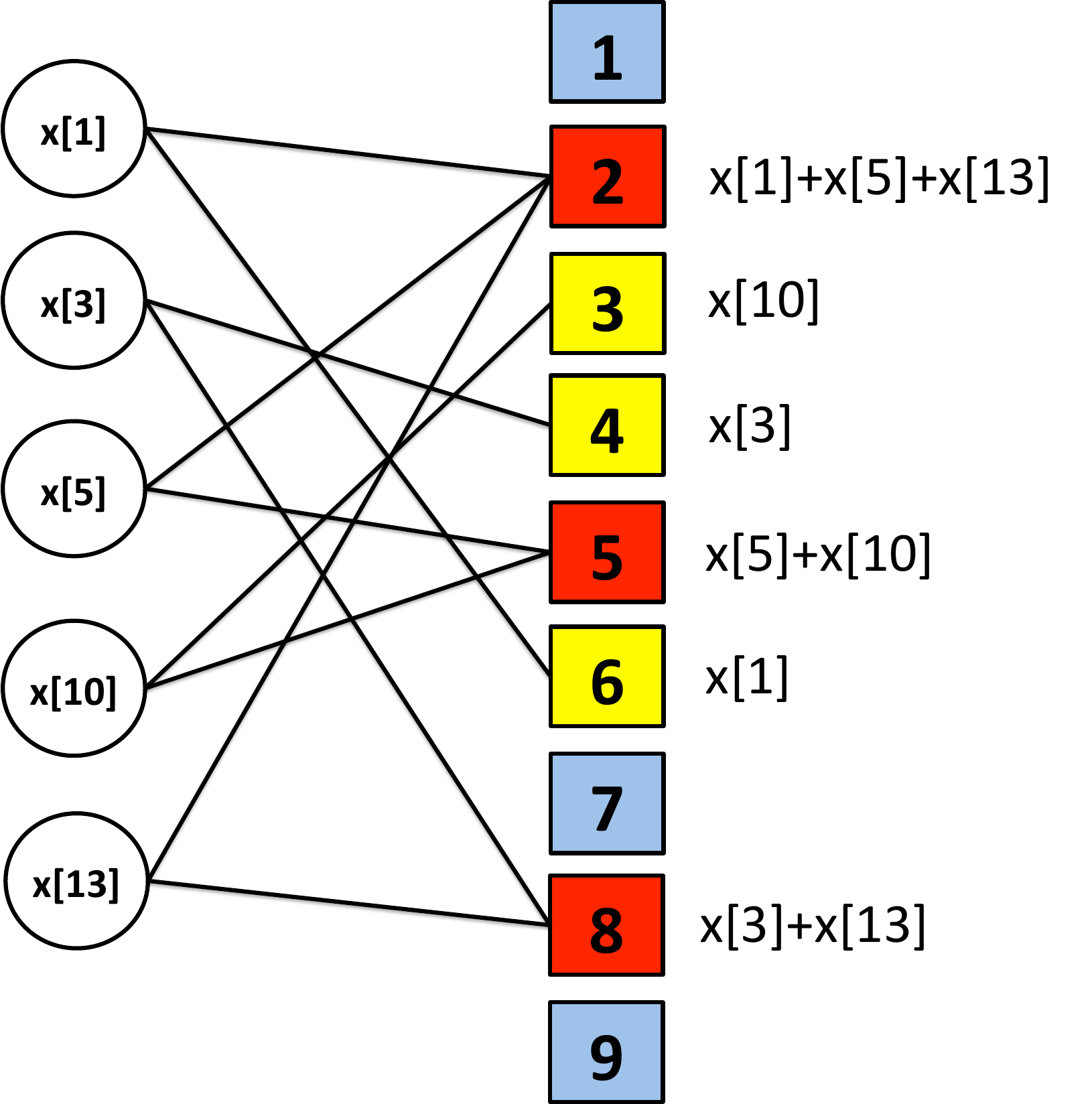}
\caption{Example consisting of $5$ left nodes with $2$ edges randomly connected to the right nodes. Blue represents ``zero-ton'', yellow represents ``single-ton'' and red represents ``multi-ton''.}\label{fig:example_bipartite}
\end{center}
\end{figure}

Now we briefly introduce how this bipartite graph helps us recover the $20$-length sparse signal $\mathbf{x}$ on the left nodes from the $9$ measurements associated with the right nodes:
\begin{align*}
	y_1&=y_7=y_9=0,\\
	y_2&=x[1]+x[5]+x[13],\\
	y_3&=x[10],\\
	y_4&=x[3],\\
	y_5&=x[5]+x[10],\\
	y_6&=x[1],\\
	y_8&=x[3]+x[13].
\end{align*}
Depending on the connectivity of the sparse bipartite graph, we categorize the measurements associated with the right nodes into the following types:
\begin{enumerate}
	\item {\bf Zero-ton}: a right node is a zero-ton if it does not involve any non-zero coefficient (e.g., {\it blue} in \figref{fig:example_bipartite}). 
	\item {\bf Single-ton}: a right node is a single-ton if it involves only one non-zero coefficient (e.g., {\it yellow} in \figref{fig:example_bipartite}).  More specifically, we refer to the index $k$ of the non-zero coefficient $x[k]$ and its associated value $x[k]$ as the {\bf index-value pair} $(k,x[k])$ for that single-ton.
	\item {\bf Multi-ton}: a right node is a multi-ton if contains more than one non-zero coefficient (e.g., {\it red} in \figref{fig:example_bipartite}). 
\end{enumerate}

To help illustrate our decoding algorithm, we assume that there exists an ``oracle''  that informs the decoder exactly which right nodes are {\it single-tons}. More importantly, the oracle further provides the index-value pair for that single-ton. In this example, the oracle informs the decoder that right nodes labeled $3$, $4$ and $6$ are single-tons with index-value pairs $(10, x[10])$, $(3,x[3])$ and $(1,x[1])$ respectively. Then the decoder can subtract their contributions from other right nodes, forming new single-tons. Therefore generally speaking, with the oracle information, the peeling decoder repeats the following steps similar to \cite{richardson2001capacity,finiasz2012private}:
\begin{itemize}
	\item[\bf Step (1)] select all the edges in the bipartite graph with right degree $1$ (identify single-ton bins);
	\item[\bf Step (2)] remove (peel off) these edges and the corresponding pair of variable and right nodes on these edges.
	\item[\bf Step (3)] remove (peel off) all other edges connected to the left nodes that have been  removed in {\bf Step (2)}. 
	\item[\bf Step (4)] subtract the contributions of the left nodes from right nodes removed in {\bf Step (3)}.
\end{itemize}
Finally, decoding is successful if all the edges are removed from the graph.

\subsection{Getting Rid of the Oracle}\label{sec:ratio_test}
Since the oracle information is critical in the peeling process, we proceed with our example and explain briefly how to obtain such information without an oracle. Clearly, we need more measurements to obtain such oracle information in its absence. Therefore, instead of simply assigning the simple {\it sum} to each right node, we assign a {\it vector-weighted sum} to the right nodes, where each left node (say $k$) is weighted by the $k$-th column of a {\bf bin detection matrix} $\mathbf{S}$. For example, we can choose the bin detection matrix $\mathbf{S}$ as
\begin{align*}
	\mathbf{S}
	&=
	\begin{bmatrix}
		1 & 1 & 1 & 1 & 1 & \cdots & 1 \\
		1 & W & W^2 & W^3 & W^4 & \cdots & W^{15}
	\end{bmatrix},
\end{align*}	
where $W=e^{\mathrm{i}\frac{2\pi}{N}}$ is the $N$-th root of unit with $N=16$. Note that this is simply the first two rows of the $20\times 20$ DFT matrix. In this way, each right node (say $r$) is assigned a $2$-dimensional vector $\mathbf{y}_r=[y_r[0],y_r[1]]^T$ and we call each vector a {\bf measurement bin}. For example, the measurements at right node $1$, $2$ and $3$ become
\begin{align*}
	\mathbf{y}_1&=\mathbf{0},\\
	\mathbf{y}_2&=
	x[1]\times
	\begin{bmatrix}
	1\\
	W
	\end{bmatrix}
	+
	x[5]\times	
	\begin{bmatrix}
	1\\
	W^5
	\end{bmatrix}
	+
	x[13]\times
	\begin{bmatrix}
	1\\
	W^{13}
	\end{bmatrix},\\
	\mathbf{y}_3&=
	x[10]	\times
	\begin{bmatrix}
	1\\
	W^{10}
	\end{bmatrix}.
\end{align*}

Now with these bin measurements, one can effectively determine if a right node is a zero-ton, a single-ton or a multi-ton. Although this procedure is formally stated in Section \ref{sec:noiseless} in our noiseless recovery results, here as an illustration, we go through the procedures for right nodes $1$, $2$ and $3$:
\begin{itemize}
	\item {\bf zero-ton bin}: consider the zero-ton right node $1$. A zero-ton right node can be identified easily since the measurements are all zero
	\begin{align}
		\mathbf{y}_1&=\mathbf{0}.
	\end{align}
	\item {\bf single-ton bin}: consider the single-ton right node $3$. A single-ton can be verified by performing a simple ``ratio test'' of the two dimensional vector:
	\begin{align*}
		\widehat{k} &= \frac{\angle{y_3[1]}/{y_3[0]}}{2\pi/16} = 10,\\
		%& = -\frac{\frac{2\pi}{20}\times (-10)}{2\pi/16} = 10\\
		\widehat{x}[\widehat{k}] &= y_3[0] = 3.
	\end{align*}
	Another unique feature is that the measurements would have identical magnitudes $|y_3[0]|=|y_3[1]|$. Both the ratio test and the magnitude constraints are easy to verify for all right nodes such that the index-value pair is obtained for peeling.
	\item {\bf multi-ton bin}: consider the multi-ton right node $2$. A multi-ton can be easily identified by the ratio test
	\begin{align*}
		\widehat{k} &= \frac{\angle{y_2[1]}/{y_2[0]}}{2\pi/16} = 12.59.
	\end{align*}
	Furthermore, the magnitudes are not identical $|y_2[0]|\neq |y_2[1]|$. Therefore, if the ratio test does not produce a non-zero integer and the magnitudes are not identical, we can conclude that this right node is a multi-ton.
\end{itemize}

%\begin{wrapfigure}{r}{0.53\textwidth}
%\vspace{-0.8cm}
%\begin{minipage}{0.56\textwidth}
\begin{algorithm}[H]
  \caption{Peeling Decoder}\label{alg:peeling}
  \begin{algorithmic}
    %\REQUIRE \# of query groups $C$; \# of peeling iterations $I$; $B=\eta s$.
    %\STATE ${\tt Given}$: observations $\bdsb{U}_c[\bdsb{j}]$ for $c\in[C]$ and $\bdsb{j}\in\GF^b$
%    \FOR{ $c=1$ to $C$}
%    	\FOR{$p=1$ to $P$}
%	\STATE ${\tt Obtain}$ $\bdsb{U}_c[\bdsb{j}]$ for $c\in[C]$ and $\bdsb{j}\in\mathbb{F}_2^b$ by \algref{alg:subsampling}
%	\ENDFOR		
%    \ENDFOR
    \FOR{$i=1$ to $I$}
	\FOR{$r =1$ to $R$}
			\STATE {\bf identify} if $\mathbf{y}_r^{(i)}$ is a single-ton bin;
			%among $\{ \mathcal{H}_{\textrm{Z}},  \mathcal{H}_{\textrm{M}},  \mathcal{H}_{\textrm{S}}(k, x[k])\}$;
			\IF{$\mathbf{y}_r^{(i)}$ is a single-ton}
				\STATE {\bf mark} the index-value pair $(\widehat{k}, \widehat{x}[\widehat{k}])$;	
			\FOR{$r'=1$ to $R$}
				\STATE {\bf locate right nodes} $r'$ connected to $\widehat{k}$ in the graph;
				\STATE {\bf peel off} $\mathbf{y}_{r'}^{(i+1)} = \mathbf{y}_{r'}^{(i)} - \widehat{x}[\widehat{k}]\mathbf{s}_{\widehat{k}}$, where $\mathbf{s}_{\widehat{k}}$ is the $\widehat{k}$-th column of the bin detection matrix $\mathbf{S}$;
			\ENDFOR
			\ELSE 
				 \STATE continue to next bin $r$.
			\ENDIF			
	\ENDFOR
    \ENDFOR
  \end{algorithmic}
\end{algorithm}
%\vspace{-0.5cm}
%\caption{Peeling Decoder}
%\end{minipage}
%\vspace{-1cm}
%\end{wrapfigure}

This simple example shows how the problem of recovering the $K$-sparse signal $\mathbf{x}$ can be cast as an instance of sparse-graph decoding, as briefly summarized in \algref{alg:peeling}. Note that the sparse bipartite graph in this example only shows the idea of peeling decoding, but does not guarantee successful recovery for an arbitrary signal. Furthermore, this example also suggests that it is possible to obtain the index-value pair of any single-ton without the help of an ``oracle'' through a properly chosen bin detection matrix. We will address later how to construct sparse bipartite graphs to guarantee successful decoding (Section \ref{sec:analysis_peeling_decoder}) and how to choose appropriate bin detection matrices for different schemes. In the following, we first present our general measurement design in Section \ref{sec:meas_design}, which is the cornerstone of our compressed sensing framework.

\section{Measurement Matrix Design}\label{sec:meas_design}
Before delving into specifics, we define the {\it row-tensor} operator $\boxtimes$ to help explain our measurement design. Given a matrix $\mathbf{S}=[\mathbf{s}_0,\cdots,\mathbf{s}_{N-1}]\in\mathbb{C}^{M_2\times N}$ and a matrix $\mathbf{H}=[\mathbf{h}_0,\cdots,\mathbf{h}_{N-1}]\in\mathbb{C}^{M_1\times N}$, the row-tensor operation $\mathbf{H}\boxtimes\mathbf{S}$ is defined such that each row of $\mathbf{H}$ is augmented element-wise by performing a tensor product with each corresponding column in the matrix $\mathbf{S}$. Mathematically, the {\it row-tensor product} is a $M_1M_2\times N$ matrix given as
\begin{align*}
	\mathbf{H}\boxtimes\mathbf{S}	
	\defn
	\begin{bmatrix}	
		\mathbf{h}_0\otimes\mathbf{s}_0 &
		\cdots &
		\mathbf{h}_{N-1}\otimes\mathbf{s}_{N-1}
	\end{bmatrix},
\end{align*}
where $\otimes$ is the standard Kronecker product. For example, let $\mathbf{H}$ be a sparse matrix with random coding patterns of $\{0,1\}$ and $\mathbf{S}$ be chosen as the first two rows of a DFT matrix as in the simple example
\begin{align}
	\mathbf{H} 
	&=
	\begin{bmatrix}
		1 & 1 &  0 & 1 & 0 & 1 & 0\\
		0 & 1 & 0 & 1 & 0 & 0 & 1\\
		1 & 0 & 0 & 1 & 1 & 1 & 1
	\end{bmatrix},\quad
	\mathbf{S}
	=
	\begin{bmatrix}
		1 & 1 & 1 & 1 & 1 & 1 & 1\\
		1 & W & W^2 & W^3 & W^4 & W^5 & W^6
	\end{bmatrix}
\end{align}
with $W=e^{\mathrm{i}\frac{2\pi}{7}}$. Then the row-tensor product is given by
\begin{align}
	\mathbf{H}\boxtimes\mathbf{S} 
	&=
	\begin{bmatrix}
		1 & 1 & 0 & 1 & 0 & 1 & 0\\
		1 & W & 0 & W^3 & 0 & W^5 & 0\\
		0 & 1 & 0 & 1 & 0 & 0 & 1\\
		0 & W & 0 & W^3 & 0 & 0 & W^6\\
		1 & 0 & 0 & 1 & 1 & 1 & 1\\
		1 & 0 & 0 & W^3 & W^4 & W^5 & W^6\\				
	\end{bmatrix}.
\end{align}
Since $\mathbf{H}$ has three rows of coding patterns, the product $\mathbf{H}\boxtimes\mathbf{S}$ contains three blocks of matrices, where each block is the corresponding sparsified version of $\mathbf{S}$ by the coding pattern in each row of $\mathbf{H}$. 

\begin{defi}[\bf Measurement Matrix]\label{def_hybrid_sensing_matrix}
Let $M=RP$ for some positive integers $R$ and $P$. Given a $R \times N$ coding matrix $\mathbf{H}$ and a $P\times N$ {\it bin detection matrix} $\mathbf{S}$, the $M\times N$ measurement matrix $\mathbf{A}$ is designed as
\begin{align}\label{meas_matrix}
	\mathbf{A} 
	&= \mathbf{H}\boxtimes\mathbf{S},
\end{align}
where $\boxtimes$ is the row-tensor product, and the coding matrix and {\it bin detection matrix} are specified below.
\begin{itemize}
	\item The {\bf coding matrix} $\mathbf{H}=[H_{r,n}]_{R\times N}$ is the $R\times N$ adjacency matrix of a bipartite graph $\mathcal{G}$ consisting of $N$ left nodes $V_1\defn [N]$ and $R$ right nodes $V_2 \defn [R]$ with an edge set $\mathcal{E} \defn V_1\times V_2$;	
	\item The {\bf bin detection matrix} $\mathbf{S}\defn [\mathbf{s}_0,\cdots,\mathbf{s}_{N-1}]$ is a $P\times N$ matrix explicitly given in Sections \ref{sec:noiseless} and \ref{sec:noisy}.
\end{itemize}
\end{defi}

%
%%
%\begin{figure}[h]
%\centering
%\includegraphics[width=1\linewidth]{hybrid_design_kew.pdf}
%\caption{Example of the measurement matrix design.}\label{fig.hybrid_design}
%\end{figure}
%%

%The {\it bin detection matrix} $\mathbf{S}$ is constructed differently for Theorem \ref{thm_koiseless_recovery},\ref{thm_sample_optimal} and \ref{thm_sub-linear_recovery}. The constructions will be specified later in the corresponding sections on noiseless and noisy recovery. 
%By choosing the bipartite graph $\mathbf{H}$ and sensing matrices $\mathbf{S}$ appropriately, the support of any $K$-sparse signal can be recovered with different measurement costs and computational costs. %Since the bipartite graph $\mathbf{H}$ (i.e., coding matrix) has a fixed number of right nodes $\eta K$ for a given the graph redundancy parameter $\eta>0$, the measurement cost $M$ of our design depends on the size of the {\it bin detection matrix} $\mathbf{S}$. For example, the design achieves the desirable cost $M=2K$ for the noiseless setting if $\mathbf{S}$ has $P=2$ rows, while achieving the measurement scaling $M={O}(K\log N)$ if the matrix $\mathbf{S}$ has $P={O}(\log N)$ rows. 

\begin{prop}\label{prop_divide-and-conquer}
The measurement $\mathbf{y}=\mathbf{A}\mathbf{x}+\mathbf{w}$ is divided into $R$  measurement bins as $\mathbf{y} = [\mathbf{y}_1^T,\cdots,\mathbf{y}_R^T]^T$ with
\begin{align}\label{divide-and-conquer}
	%\mathbf{y}_r = \sum_{k\in\mathcal{K}_j}x[k]\mathbf{s}_k + \mathbf{w}_j,\quad j=1,\cdots,\eta K,
	%\mathbf{y}_r = \mathbf{S}\mathbf{z}_r+\mathbf{w}_r,\quad r =1,\cdots,R,
	\mathbf{y}_r = \mathbf{S}\mathbf{z}_r+\mathbf{w}_r,\quad r =1,\cdots,R
\end{align}
where $\mathbf{w}_r$ is the noise in the $r$-th measurement bin and $\mathbf{z}_r=[z_r[0],\cdots,z_r[N-1]]^T$ is a reduced sparse vector
\begin{align}
	z_r[k] = 
	\begin{cases}
		x[k], & k\in\mathcal{N}(r)\\
		0, &k\notin\mathcal{N}(r)
	\end{cases},
\end{align}
and $\mathcal{N}(r)$ is the set of left nodes connected to right node $r=1,\cdots,R$ in the graph $\mathcal{G}$. 
\end{prop}
\begin{proof}
The proof is straightforward and hence omitted.
%	See Appendix \ref{proof_prop_divide-and-conquer}.
%By Definition \ref{def_hybrid_sensing_matrix}, the measurement matrix $\mathbf{A}$ is divided into $\eta K$ sub-matrices $\mathbf{A}=[\mathbf{A}_1^T,\cdots,\mathbf{A}_{\eta K}^T]^T$, where each sub-matrix $\mathbf{A}_j$ is a $P\times N$ matrix given by
%\begin{align}
%	\mathbf{A}_j = \mathbf{S}\cdot\diag{\mathbf{h}_j},\quad j=1,\cdots,\eta K.
%\end{align}
%By letting $\mathbf{z}_r\defn \diag{\mathbf{h}_j}\mathbf{x}$, the measurements from the $j$-th sub-matrix $\mathbf{A}_j$, called a ``measurement bin'', can be viewed as a decoupled sparse problem with the same {\it bin detection matrix} $\mathbf{S}$. In other words, we have $z_j[k]=x[k]$ if $k\in\supp{\mathbf{z}_r}$ and $z_j[k]=0$ if otherwise. Thus the formulation in \eqref{divide-and-conquer} follows.
%The support information of the original sparse problem $\supp{\mathbf{x}}$ is retained since $\bigcup_{j=1}^{\eta K} \mathcal{K}_j = \supp{\mathbf{x}}$.
\end{proof}

Since the vector $\mathbf{x}$ is by itself sparse on a support that may or may not overlap with the coding pattern given by the graph $\mathcal{G}$, the resulting equivalent sparse vector $\mathbf{z}_r$ in each bin $r$ is even sparser with a reduced support $\supp{\mathbf{x}}\cap\mathcal{N}(r)$. If the coding pattern happens to make $\mathbf{z}_r$ a $1$-sparse vector, we have a much easier problem to solve. Then we can use the recovered $1$-sparse coefficient to recover other coefficients iteratively. Therefore, we need to distinguish the type of each bin in order to determine if $\mathbf{z}_r$ is $1$-sparse, which can be regarded as a separate hypothesis in the presence of noise $\mathbf{w}_r$:
\begin{enumerate}
	\item $\mathbf{y}_r$ is a {\bf zero-ton} bin if $\supp{\mathbf{z}_r}=\varnothing$, denoted by $\mathbf{y}_r \sim \mathcal{H}_{\textrm{Z}}$;
	\item $\mathbf{y}_r$ is a {\bf single-ton} bin with the index-value pair $(k,x[k])$ if $\supp{\mathbf{z}_r}=\{k\}$ for some $k\in[N]$ and $z_r[k]=x[k]$, denoted by $\mathbf{y}_r \sim \mathcal{H}_{\textrm{S}}(k,x[k])$;
	\item $\mathbf{y}_r$ is a {\bf multi-ton} bin if $\left|\supp{\mathbf{z}_r}\right| \geq 2$, denoted by $\mathbf{y}_r \sim \mathcal{H}_{\textrm{M}}$.
\end{enumerate}

The spirit of divide-and-conquer is also manifested in this general design since the design of coding matrix ensures fast decoding by peeling, while the bin detection matrix ensures the correct detection of various bin hypotheses. These two designs are completely modular and can be designed independently depending on the applications. Now, given the above general measurement design, the following questions are of particular interests: 
\begin{enumerate}
	\item Given $N$ left nodes and $R$ right nodes, how to construct a bipartite graph that guarantees a ``friendly'' distribution of single-tons, zero-tons and multi-tons for successful peeling?
	\item Given the sparsity $K$ of the bipartite graph, what is the minimum number of right nodes $R$ to guarantee successful peeling?
	\item How to choose the {\it bin detection matrix} $\mathbf{S}$ in general for providing the oracle information, especially when the measurements are noisy?
\end{enumerate} 

In the following, we answer these questions in details and discuss the specific constructions for $\mathbf{H}$ and $\mathbf{S}$. In Section \ref{sec:analysis_peeling_decoder}, we first present the peeling decoder analysis that guides the design of the bipartite graphs and the associated coding matrix $\mathbf{H}$, and then discuss the constructions of the bin detection matrix $\mathbf{S}$ for both noiseless and noisy scenarios in Section \ref{sec:noiseless}, \ref{sec:noisy}, and Section~\ref{sec:noisy_continuous}.

\section{Sparse Graph Design and Peeling Decoder}\label{sec:analysis_peeling_decoder}

As mentioned above, the design of the coding matrix, or namely the sparse bipartite graph, is independent of the design of the bin detection matrix since they target different architectural objectives of the decoding algorithm. Simply put, the coding matrix (i.e. the sparse graph) can be designed assuming that there is an oracle present at decoding, while the bin detection matrix helps replace the oracle, which can be designed independently. Therefore, in this section we focus on the design of the coding matrix and study the sparse bipartite graphs that guarantee successful oracle-based decoding.

\subsection{Sparse Graph Design for Compressed Sensing}
The design of sparse bipartite graphs for peeling decoders has been studied extensively in the context of erasure-correcting sparse-graph codes \cite{luby2001efficient,richardson2001capacity}. In this section, for simplicity we consider {\it the ensemble of left ${d}$-regular bipartite graphs} $\mathcal{G}_{\rm reg}^N(R, {d})$ consisting of $N$ left nodes (unknown coefficients $x[k]$ for $k\in[N]$) and $R$ right nodes (compressed measurements $\mathbf{y}_r$ for $r=1,\cdots,R$), where each left node $k\in[N]$ is connected to ${d}$ right nodes $r=1,\cdots,R$ uniformly at random and the number of right nodes is linear in the sparsity $R=\eta K$. We call $\eta$ the {\it redundancy parameter}. %Recall that the simple example depicted in \figref{fig:example_bipartite} is an instance of the regular graph ensemble with ${d}=2$ consisting of $N=20$ left nodes (including the left nodes with non-zero values) and $\eta K= 9$ right nodes.

%\begin{wrapfigure}{r}{0.4\textwidth}
%\vspace{-0.2cm}
%\begin{minipage}{0.4\textwidth}
\begin{figure}[h]
\begin{center}
\includegraphics[width=0.32\linewidth]{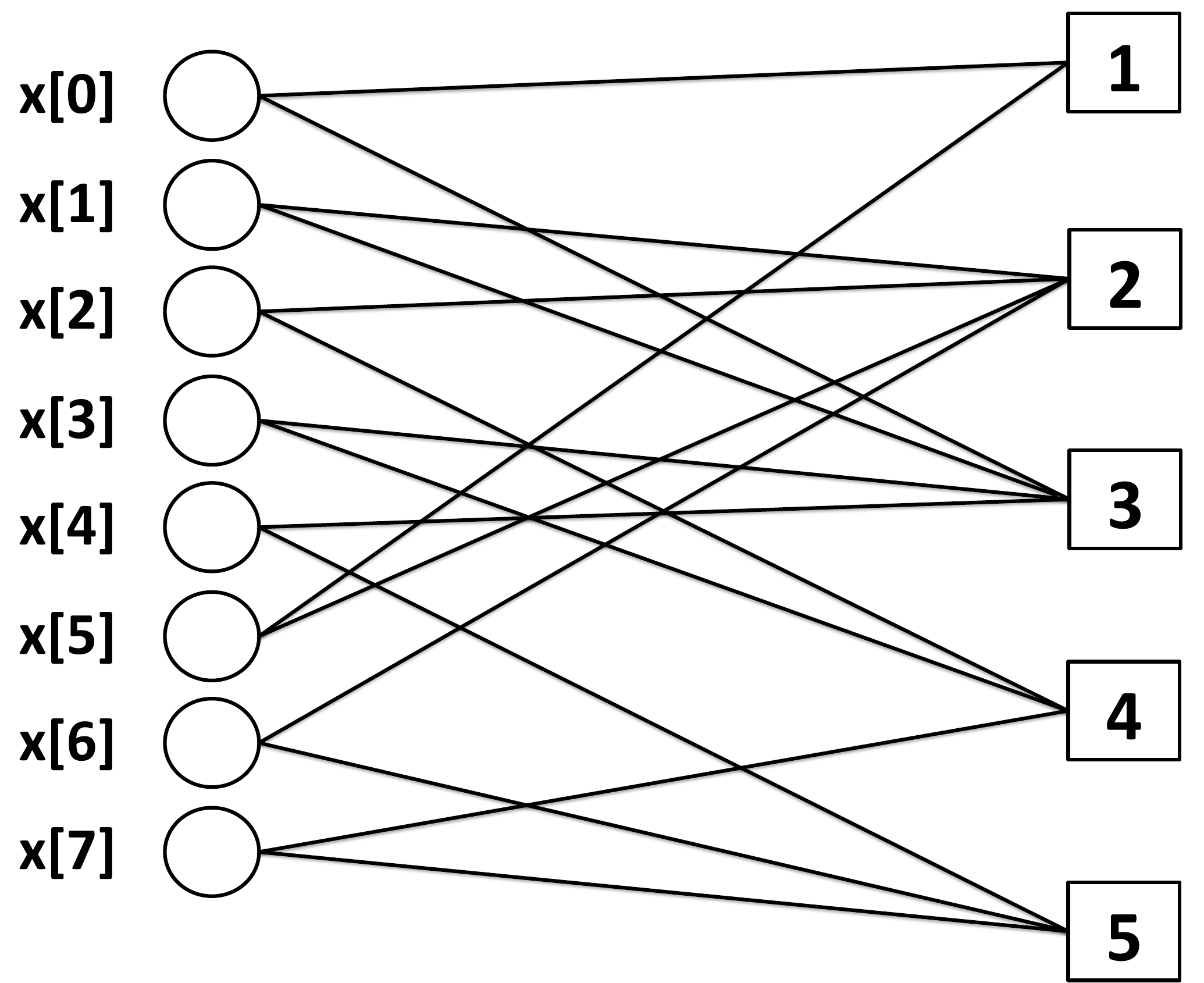}\\
\includegraphics[width=0.28\linewidth]{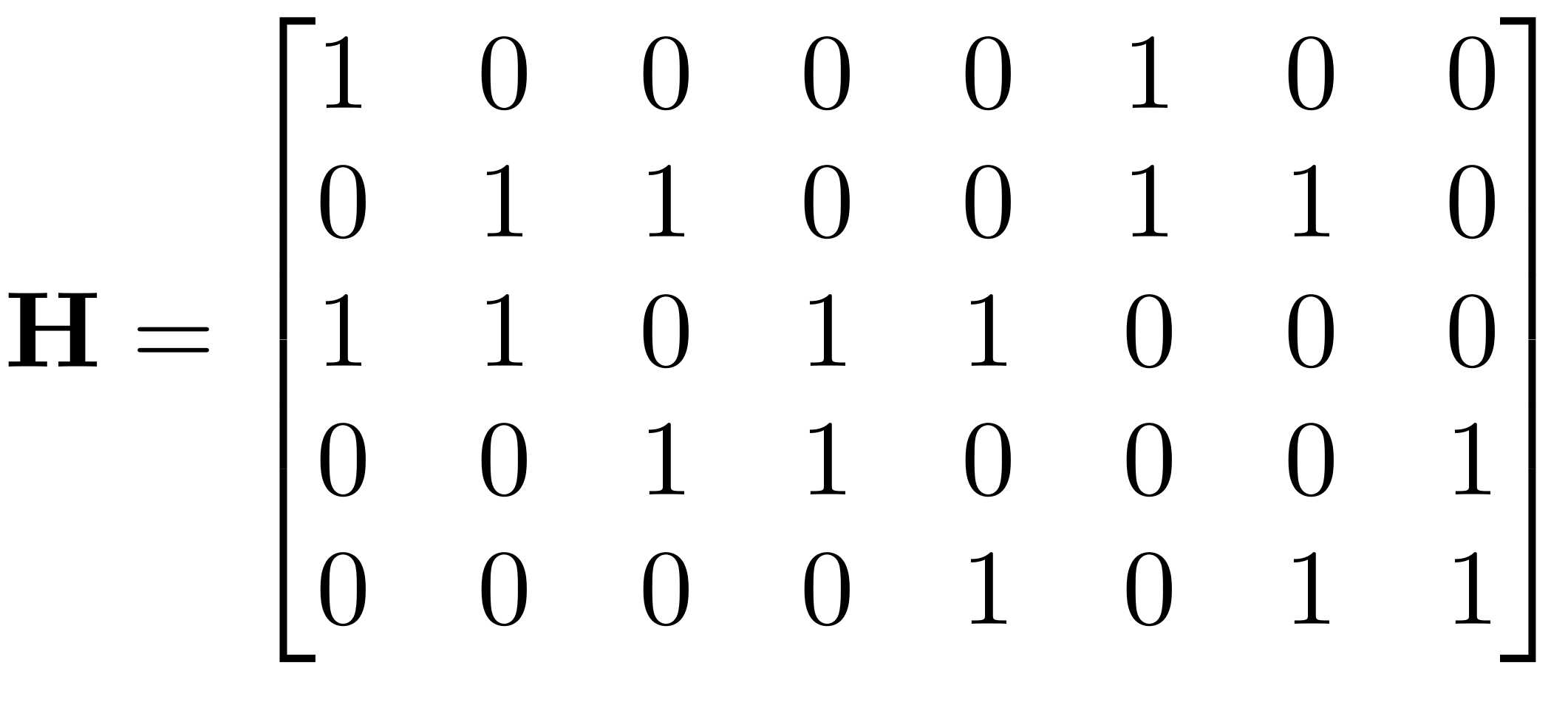}
\caption{An example of the bipartite graph from the regular graph ensemble with ${d}=2$ left degrees, consisting of $N=8$ left nodes and $R=5$ nodes, where the left nodes are labeled by the signal $\mathbf{x}=[x[0],\cdots,x[7]]^T$.}\label{fig:example_coding_matrix}
\end{center}
\end{figure}

%\end{minipage}
%\end{wrapfigure}

The coding matrix $\mathbf{H}$ constructed from the regular graph ensemble conforms with a random ``balls-and-bins'' model, where each row of $\mathbf{H}$ corresponds to a ``bin'' (i.e., right node) and each column of $\mathbf{H}$ corresponds to a ``ball'' (i.e., left node). If the $(r,k)$-th entry $H_{r,k}=1$, then we say that the $k$-th ball is thrown into the $r$-th bin. In the ``balls-and-bins'' model associated with the regular ensemble $\mathcal{G}_{\rm reg}^N(R, {d})$, each ball $k\in[N]$ is thrown uniformly at random to ${d}$ bins. In the context of LDPC codes, the $k$-th coefficient $x[k]$ (variable node) appears in the parity check constraints in $d$ right nodes (check nodes) chosen uniformly at random. For example, consider a smaller example with $N=8$ left nodes and $R=5$ nodes, where $\mathbf{x}=[x[0],\cdots,x[7]]^T$ is some generic signal vector. Then, an instance from the $2$-regular ensemble $\mathcal{G}_{\rm reg}^8(5, 2)$ and the associated coding matrix $\mathbf{H}$ are shown in \figref{fig:example_coding_matrix}.

In our compressed sensing design, the sparse bipartite graph for peeling is the ``pruned'' graph after removing the left nodes with zero values. For example, if the signal is $4$-sparse with non-zero coefficients $x[1]$, $x[4]$, $x[5]$ and $x[6]$, then the {\it ``pruned'' graph} is reduced to that in \figref{fig:example_coding_matrix_pruned} on the right from the {\it full graph} on the left. Another example of a ``pruned'' graph  has been shown in \figref{fig:example_bipartite}, which is associated with a $5$-sparse signal and a left $2$-regular graph with $N=20$ left nodes and $R=9$ right nodes. 

\begin{figure}[h]
\centering
\includegraphics[width=0.6\linewidth]{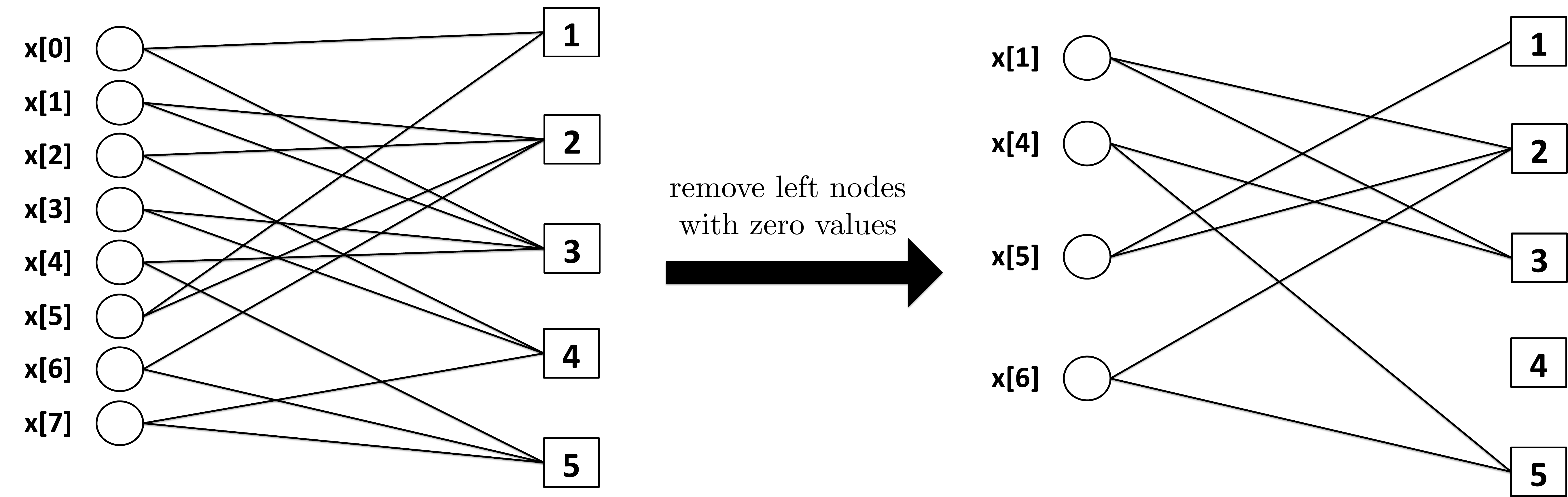}
\caption{The ``pruned'' bipartite graph when the signal $\mathbf{x}=[x[0],\cdots,x[7]]^T$ is $4$-sparse with non-zero coefficients $x[1]$, $x[4]$, $x[5]$ and $x[6]$.}\label{fig:example_coding_matrix_pruned}
\end{figure}
%
%\begin{align}
%	\mathbf{H}=
%	\begin{bmatrix}
%      1  &   0  &   0  &   0    & 0   &    1  &   0   &  0\\
%     0   &  1  &   1  &   0    & 0    &    1  &   1   &  0\\
%     1   &  1  &   0   &  1   &  1    &    0  &   0   &  0\\
%     0  &   0  &   1   &  1  &   0    &    0  &   0   &  1\\
%     0  &   0  &   0   &  0 &    1    &    0  &   1   &  1
%	\end{bmatrix}.
%\end{align}

Given some $K$-sparse signal $\mathbf{x}$, the {\it pruned} graph in \figref{fig:example_coding_matrix_pruned}, instead of the {\it full graph} in \figref{fig:example_coding_matrix}, determines the peeling decoder performance. However, the pruned graph depicted in \figref{fig:example_coding_matrix_pruned} does not lead to successful decoding since the peeling is stuck with all multi-tons after removing the single-ton from right node $\# 1$. The intuition is that there are $4$ nodes on the left with degree $2$ but only $5$ nodes on the right. Therefore there is a high probability for each right node to connect to more than one left node (i.e., in this case only one right node has degree $1$). In general, given the left degree ${d}$ of the ensemble and the sparsity $K$, the graph needs to contain a sufficient number of right nodes to guarantee the success of the peeling decoder by choosing the redundancy parameter $\eta$ properly. In the following, we study the peeling decoder performance over the pruned graphs from the regular ensemble $\mathcal{G}_{\rm reg}^N(R, {d})$ and shed light on how to specify the parameter $\eta$ appropriately.
%In other words, we can determine how many times to throw each ``ball'', but we cannot guarantee in the end how many balls land in each bin due to the pruning by the signal sparsity.
%This is a major distinction between our design and those in LDPC codes literature \cite{richardson2001capacity,luby2001efficient}, where the {\it full graph} is used for peeling decoding. In the context of sparse-graph codes design, the full graph can be designed exactly by specifying the {\it degree distributions} (explained later) of both the left nodes and the right nodes. For example, in \figref{fig:example_coding_matrix_pruned}, the left nodes and right nodes can be designed to have exactly degree $5$ and degree $8$ respectively so that the full graph becomes left-regular and right-regular. However, one cannot guarantee the pruned graph to maintain the same regularity by design. 
%In our framework, we need to have good {\it full graphs} such that the pruned graph induced by the sparsity pattern of the signal $\mathbf{x}$ guarantees successful decoding. 

\subsection{Oracle-based Peeling Decoder Analysis using the Regular Ensemble $\mathcal{G}_{\rm reg}^N(R, {d})$}\label{sec:peeling_decoder_analysis}
In this section, we show that for the compressed sensing problem, if the redundancy parameter $\eta=R/K$ and the left regular degree ${d}$ are chosen properly for the regular graph ensemble $\mathcal{G}_{\rm reg}^N(R, {d})$, then for an arbitrary $K$-sparse signal $\mathbf{x}$, all the edges of the  {\it pruned graph} can be peeled off in ${O}(K)$ peeling iterations {\it with high probability}. The formal statement is given in Theorem \ref{thm_peeling_decoder}. In other words, we show that as long as the {\it full graph} is chosen properly, the {\it pruned graph} can lead to successful decoding with high probability for any given sparse signal. Our analysis is similar to the arguments in \cite{luby2001efficient,richardson2001capacity} using the {\it density evolution} analysis from modern coding theory, which tracks the average density\footnote{The density here refers to fraction of the remaining edges, or namely, the number of remaining edges divided by the total number of edges in the graph.} of the remaining edges in the pruned graph at each peeling iteration of the algorithm. 

The proof techniques to analyze the peeling decoder in our framework are similar to those from \cite{luby2001efficient} and \cite{richardson2001capacity}, except that the graph we have is the ``pruned'' version with a sub-linear fraction $K$ left nodes given adversarially by the input. Hence, this leads to some differences in the analysis from those in \cite{richardson2001capacity,luby2001efficient}, such as the degree distributions of the graphs (explained later) and the expansion properties of the graphs. As a result, we present an independent analysis here for our peeling decoder. In the following, we provide a brief outline of the proof elements highlighting the main technical components.
%Therefore, we construct our proof step-by-step as follows:
\begin{itemize}
	\item {\bf Density evolution}: %(mean performance analysis over cycle-free graphs for any finite iterations $i$) 
	We analyze the performance of our peeling decoder over a {\it typical graph} (i.e., cycle-free) of the
ensemble $\mathcal{G}_{\rm reg}^N(R, {d})$ for a fixed number of peeling iterations $i$. We assume that a local neighborhood of every edge in the graph is cycle-free (tree-like) and derive a recursive equation that represents the average density of remaining edges in the pruned graph at iteration $i$.
	\item {\bf Convergence to density evolution}: %(concentration analysis around the mean performance for any finite iterations $i$)
	Using a Doob martingale argument as in \cite{richardson2001capacity} and \cite{pedarsani2017phasecode}, we show that the local neighborhood of most edges of a randomly chosen graph from the ensemble $\mathcal{G}_{\rm reg}^N(R, {d})$ is cycle-free with high probability. This proves that with high probability, our peeling decoder removes all but an arbitrarily small fraction of the edges in the pruned graph (i.e., the left nodes are removed at the same time after being decoded) in a constant number of iterations $i$. %The main difference of our convergence analysis compared to \cite{richardson2001capacity} is that the right edge degree distribution in our ensemble is induced by the left degrees ${d}$ and the sparsity pattern of $\mathbf{x}$, while the right degree is regular in \cite{richardson2001capacity}.
	\item {\bf Graph expansion property} for complete decoding: 
	We show that if the sub-graph consisting of the remaining edges is an ``expander'' (as will be defined later in this section), and if our peeling decoder successfully removes all but a sufficiently small fraction of the left nodes from the pruned graph, then it removes all the remaining edges of the ``pruned'' graph successfully. This completes the decoding of all the non-zero coefficients in $\mathbf{x}$.
\end{itemize}

\subsubsection*{Density Evolution}
Density evolution, a powerful tool in modern coding theory, tracks the average density of remaining edges that are not decoded after a fixed number of peeling iteration $i>0$. We describe the concept of {\it directed neighborhood} of a certain edge in the pruned graph up to depth $\ell=2i$. This concept is important in the density evolution analysis since the peeling of an edge in the $i$-th iteration depends solely on the removal of the edges from this neighborhood in the previous $i-1$ iterations. The {\it directed neighborhood} $\mathcal{N}_{\textrm{e}}^\ell$ at depth $\ell$ of a certain edge $e = (v, c)$ is defined as the induced sub-graph containing all the edges and nodes on paths $e_1,\cdots, e_\ell$ starting at a variable node $v$ (left node) such that $e_1 \neq e$. An example of a directed neighborhood of depth $\ell=2$ is given in Fig. \ref{fig:localtree}. 

\begin{figure}[h]
\begin{center}
\includegraphics[scale=0.16]{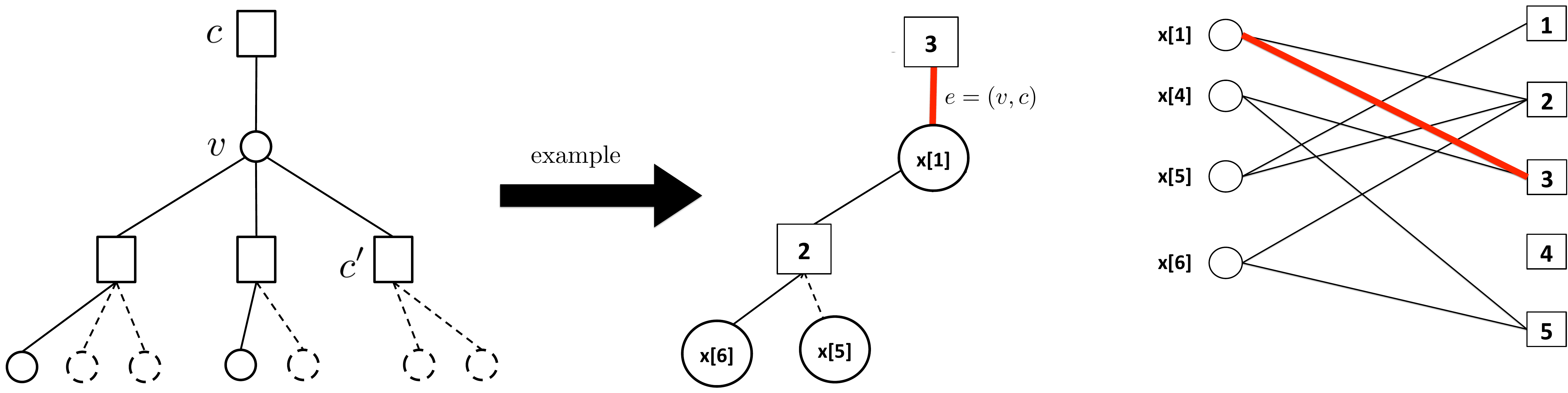}
\caption{On the left sub-figure, we illustrate the directed neighborhood of depth $2$ of an edge $e = (v,c)$, namely $\mathcal{N}_{\textrm{e}}^{2}$, while on the right we show this neighborhood for our example depicted in \figref{fig:example_coding_matrix_pruned}. The dashed lines on the left correspond to nodes/edges removed at the end of iteration $i-1$. The edge between $v$ and $c$ can be potentially removed at iteration $i$ as one of the check nodes (right nodes) $c'$ is a single-ton (it has no more variable nodes remaining at the end of iteration $i-1$). In our example, unlike the check node $c'$ on the left, the edge $e=(x[1] , 3)$ cannot be removed since the  check node is still a multi-ton (i.e., $x[6]$ and $x[1]$ are still attached).}
\label{fig:localtree}
\end{center}
\end{figure}

To analyze the performance of the peeling decoder over the pruned graph, we need to understand the edge degree distributions on the left and right for the pruned graph. Let $\rho_{j}$ be the fraction of edges in the pruned graph connecting to right nodes with degree ${j}$. Clearly, the total number of edges is $K {d}$ in the pruned graph since there are $K$ left nodes in the pruned graph and each left node has degree ${d}$. Therefore, since the expected number of edges connected to right nodes with degree ${j}$ can be obtained as $\Prob{\textrm{a right node has degree}~{j}} R {j}$, the fraction $\rho_{j}$ can be obtained as
\begin{align}
	\rho_{j} 
	= \frac{\Prob{\textrm{a right node has degree}~{j}} R{j} }{K{d}} =\frac{{j} \eta }{{d}} \Prob{\textrm{a right node has degree}~{j}},
\end{align}
where we have used $R=\eta K$ and $\eta$ is the redundancy parameter. According to the ``balls-and-bins'' model, the degree of a right node follows the binomial distribution $B({d}/(\eta K), K)$, and as $K$ approaches infinity can be well approximated by a Poisson variable as
\begin{align}
	\Prob{\textrm{a right node has degree}~{j}} \approx \frac{({d}/\eta)^{j} e^{-{d}/\eta}}{{j}!}.
\end{align}
As a result, the fraction $\rho_{j}$ of edges connected to right nodes having degree ${j}$ is 
\begin{align}\label{rho_d}
	\rho_{j} 
	= \frac{({d}/\eta)^{{j}-1} e^{-{d}/\eta}}{({j}-1)!}.
\end{align}

Now let us consider the local neighborhood $\mathcal{N}_{\textrm{e}}^{2i}$ of an arbitrary edge $e=(v,c)$  with a left regular degree ${d}$ and right degree distribution given by $\{\rho_{j}\}_{j=1}^{K}$. If the sub-graph corresponding to the neighborhood $\mathcal{N}_{\textrm{e}}^{2i}$ of the edge $e=(v,c)$ is a {\it tree} or namely {\it cycle-free}, then the peeling procedures over different bins in the first $i$ iterations (see Section \ref{sec:simple_example}) are independent, which can greatly simplify our analysis. Density evolution analysis is based on the assumption that this neighborhood is cycle-free (tree-like), and we will prove later (in the next subsection) that all graphs in the regular ensemble behave like a tree when $N$ and $K$ are large and hence the actual density evolution concentrates well around the density evolution result.

Let $p_i$ be the probability of this edge being present in the pruned graph after $i>0$ peeling iterations. If the neighborhood is a tree as in \figref{fig:tree_DE}, the probability $p_i$ can be written with respect to the probability $p_{i-1}$ recursively. 
\begin{align}
	p_i 
	&= \left(1-\sum_{j} \rho_{j} (1-{p}_{i-1})^{{j}-1}\right)^{{d}-1},\quad i = 1,2,3,\cdots.
\end{align}
% However, this is not necessarily true for an arbitrary ``pruned'' graph constructed from the regular graph ensemble $\mathcal{G}_{\rm reg}^N(R, {d})$ and the $K$-sparse signal. 
The term $\sum_{j} \rho_{j} (1-{p}_{i-1})^{{j}-1}$ can be simplified using the right degree generating polynomial
\begin{align}
	\rho(x) \defn \sum_{j} \rho_{j} x^{{j}-1} = e^{-(1-x)\frac{{d}}{\eta}},
\end{align}
where we have used \eqref{rho_d} to derive the second expression. 

%\begin{wrapfigure}{r}{0.4\textwidth}
%\vspace{-0.6cm}
%\begin{minipage}{0.4\textwidth}
\begin{figure}[h]
\begin{center}
\includegraphics[width=0.3\linewidth]{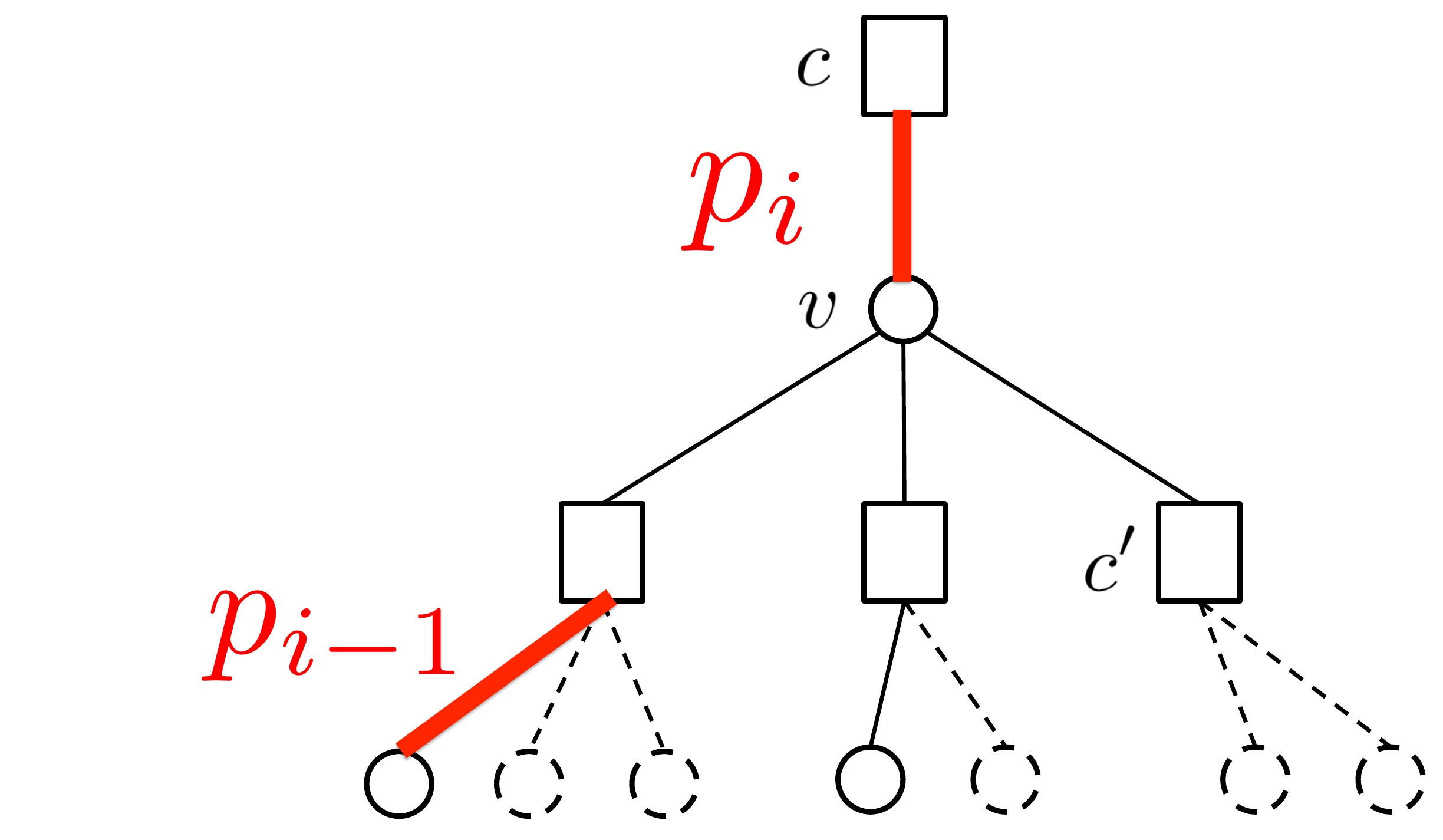}
\caption{The schematic of density evolution in a local tree-like neighborhood.}\label{fig:tree_DE}
\end{center}
\end{figure}
%\end{minipage}
%\vspace{-0.5cm}
%\end{wrapfigure}

Therefore, the density evolution equation for our peeling decoder can be obtained as
\begin{align}\label{density_evolution}
p_i &= f({p}_{i-1}) = \left(1- e^{-\frac{{d}}{\eta}{p}_{i-1}}\right)^{{d}-1},\quad i = 1,2,3,\cdots.
\end{align}
An example of the density evolution with ${d}=3$ and different values of $\eta$ is given in \figref{fig:DE}. 
\begin{figure}[p]
\centering
\begin{subfigure}{0.4\textwidth}
\includegraphics[width=1\linewidth]{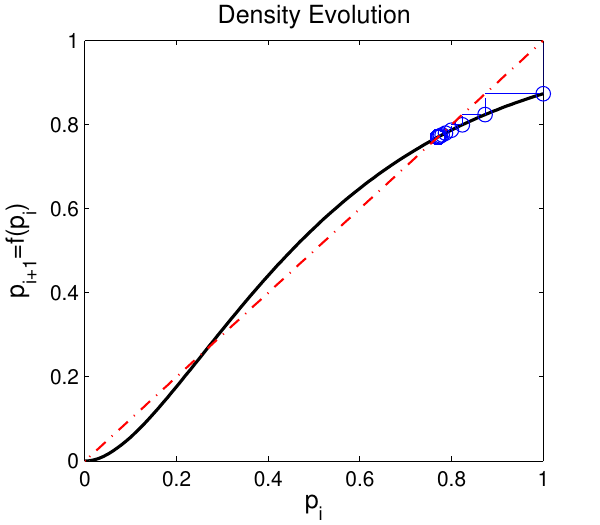}
\caption{$\eta=1.1$}
\end{subfigure}
 \begin{subfigure}{0.4\textwidth}
\includegraphics[width=1\linewidth]{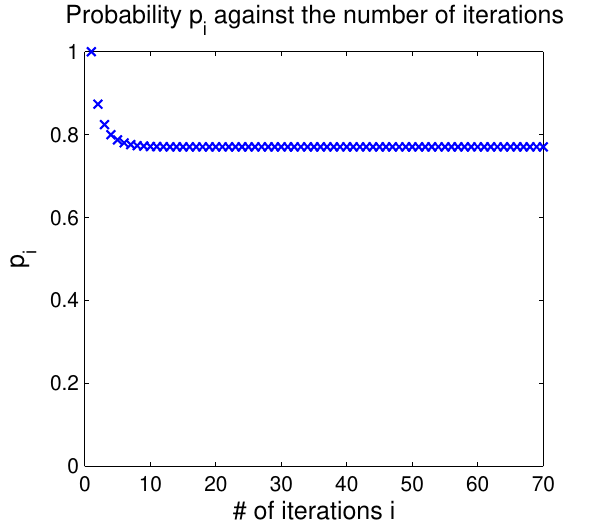}
\caption{$\eta=1.1$}
\end{subfigure}
\begin{subfigure}{0.4\textwidth}
\includegraphics[width=1\linewidth]{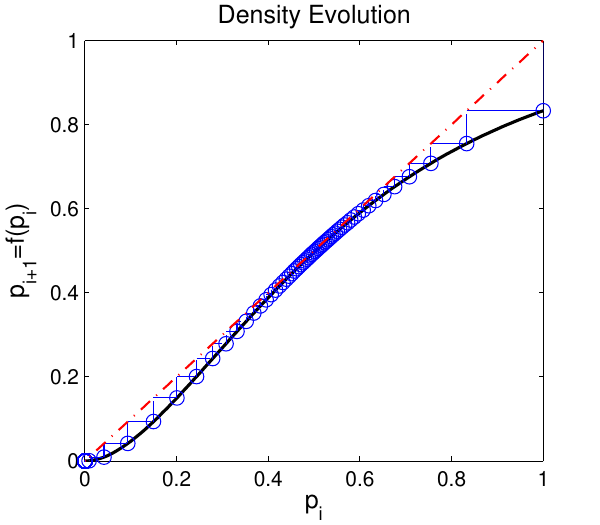}
\caption{$\eta=1.23$}
\end{subfigure}
\begin{subfigure}{0.4\textwidth}
\includegraphics[width=1\linewidth]{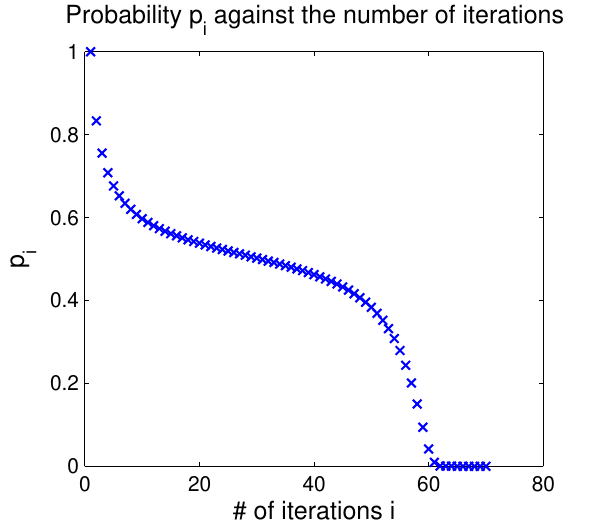}
\caption{$\eta=1.23$}
\end{subfigure}
\begin{subfigure}{0.4\textwidth}
\includegraphics[width=1\linewidth]{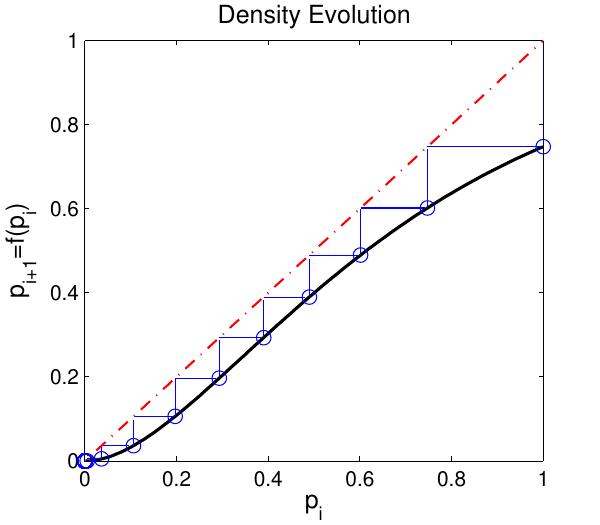}
\caption{$\eta=1.5$}
\end{subfigure}
\begin{subfigure}{0.4\textwidth}
\includegraphics[width=1\linewidth]{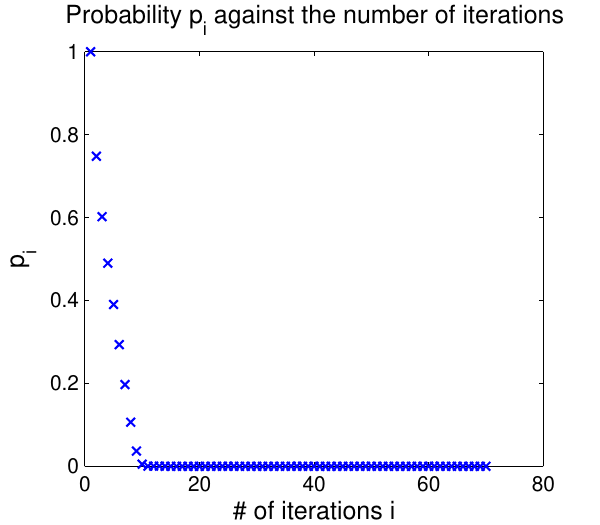}
\caption{$\eta=1.5$}
\end{subfigure}
\caption{The density evolution $f(p_i)$ and the probability $p_i$ at each iteration $i$, where we have shown the case with ${d}=3$ and $\eta=1.1$, $\eta=1.23$, $\eta=1.5$. In the density evolution figures (a)-(c)-(e), the red line is the line $p_{i+1}=p_i$ while the black line is the actual density evolution recursion $f(p_i)$ against $p_i$. The blue circles that ``zig-zag'' between the red line and the black line are the specific $p_i$'s that are achieved at each peeling iteration. It can be seen from (a) that when $\eta$ is small (i.e. $\eta =1.1$), the density evolution reaches a fixed point at around $p_i\approx 0.8$. On the other hand, when $\eta$ is greater than the threshold $1.23$ given by Table \ref{Table_beta}, the density $p_i$ reaches $0$ very quickly in (a) when $\eta=1.5$. The values of $p_i$ marked by the blue circles in (a)-(c)-(e) are further plotted against the peeling iterations $i$ in (b)-(d)-(f), where in the case with $\eta=1.5$ the density $p_i$ approaches $0$ after less than $10$ iterations.}\label{fig:DE}
\end{figure}
Clearly, the probability $p_i$ can be made arbitrarily small for a sufficiently large but finite $i>0$ as long as ${d}$ and $\eta$ are chosen properly. One can find the minimum value $\eta$ for a given ${d}$ to guarantee $p_i<p_{i-1}$, which is shown in Table \ref{Table_beta}. Due to lack of space we only show up to ${d}=6$.

\begin{lem}[Density evolution]
Denote by $\mathcal{T}_{i}$ the event where the local $2i$-neighorhood $\mathcal{N}_{\textrm{e}}^{2{i}}$ of every edge in the graph is tree-like and let $Z_{i}$ be the total number of edges that are not decoded after $i$ (an arbitrarily large but fixed) peeling iterations. For any $\varepsilon>0$, there exists a finite number of iteration $i>0$ such that 
	\begin{align}\label{mean_analysis_result}
		\mathbb{E}[Z_{i}|\mathcal{T}_{i}] = K {d} \varepsilon/4,
	\end{align}
where the expectation is taken with respect to the random graph ensemble $\mathcal{G}_{\rm reg}^N(R, {d})$ with the left regular degree ${d}$ and the redundancy parameter $\eta=R/K$ chosen from Table \ref{Table_beta} below.
\begin{table}[h]
\begin{center} 
\begin{tabular}{|c|c|c|c|c|c|c|c|c|}
  \hline
  ${d}$ & 2 & 3& 4 & 5 & 6\\% &7 \\
  \hline
 minimum $\eta$ & 2.0000 &   1.2219 &   1.2948 &   1.4250  &  1.5696\\ %  &  1.7192\
%  minimum $M$ & $2K$ &   $1.23K$ &   $1.30K$ &  $1.43K$  &  $1.57K$ \\ %  &  1.7192\\
  \hline
\end{tabular}
\end{center}
\caption{Minimum value for $\eta$ given the regular degree ${d}$ according to density evolution.
}\label{Table_beta}
\end{table}
\end{lem}
Based on this lemma, we can see that if the pruned bipartite graph has a local neighborhood that is tree-like up to depth $2i$ for every edge, the peeling decoder on average peels off all but an arbitrarily small fraction of the edges in the graph. We prove this lemma below.
	
\begin{proof}
	%See Appendix \ref{sec:density_evolution}.
Let $Z_{i}^{(\textrm{e})}\in\{0,1\}$ be the random variable denoting the presence of edge $e$ after $i$  iterations, thus 
\begin{align}\label{sum_edges}
	Z_{i} = \sum_{e=1}^{K{d}} Z_{i}^{(\textrm{e})}.
\end{align}
%Since each edge is peeled off independently given the event $\mathcal{T}_{i}$, 
The expected number of remaining edges over cycle-free graphs can be obtained as
\begin{align}\label{mean_analysis}
	\mathbb{E}\left[Z_{i}|\mathcal{T}_{i}\right] = \sum_{e=1}^{K{d}} \mathbb{E}\left[Z_{i}^{(\textrm{e})}|\mathcal{T}_{i}\right]  = K {d} p_i,
\end{align}
where by definition $p_i = \Prob{Z_{i}^{(\textrm{e})}=1|\mathcal{T}_{i}}$ is the {\it conditional probability} of an edge in the $i$-th peeling iteration conditioned on the event $\mathcal{T}_{i}$ studied in the density evolution equation \eqref{density_evolution}. We are interested in the evolution of such probability $p_i$. In the following, we prove that for any given $\varepsilon>0$, there exists a finite number of iterations $i>0$ such that $p_i \leq \varepsilon/4$, which leads to our desired result in \eqref{mean_analysis_result}.	
\end{proof}

\subsubsection*{Convergence to Density Evolution}	
Given the mean performance analysis (in terms of the number of undecoded edges) over cycle-free graphs through density evolution, now we provide a {\it concentration analysis} on the number of the undecoded edges $Z_{i}$ for {\it any graph from the regular ensemble} at the $i$-th iteration, by showing that $Z_{i}$ converges to the density evolution result.
\begin{lem}
Over the probability space of all graphs from $\mathcal{G}_{\rm reg}^N(R, {d})$,  let $p_i$ be as given in the density evolution \eqref{density_evolution}. Given any $\varepsilon>0$ and a sufficiently large $K$, there exists a constant $c_4>0$ such that
\begin{align}
	&{\tt (i)}~\quad\quad \mathbb{E}[Z_{i}] < K {d} \varepsilon/2 \label{mean_on_general_graph}\\
	&{\tt (ii)} \quad ~~~  \Prob{\left|Z_{i}-\mathbb{E}[Z_{i}]\right|>K {d} \varepsilon/2} 
	\leq
	2\exp\left(-c_4 \varepsilon^2 K^{\frac{1}{4i+1}}\right)\\
	&{\tt (iii)} \quad ~ \Prob{\left|Z_{i}- K {d} \varepsilon/2  \right|>K {d} \varepsilon/2} 
	\leq
	2\exp\left(-c_4 \varepsilon^2 K^{\frac{1}{4i+1}}\right)\label{concentration_DE}
\end{align}
\end{lem}
\begin{proof}  
	The details of the proof are given in Appendix \ref{sec:concentration_analysis}, but here we provide an outline of the proof. The {\it concentration analysis} is performed with respect to the number of the remaining edges for an {\it arbitrary graph from the ensemble} by showing that $Z_{i}$ converges to the mean analysis result. This proof is done in two steps:
\begin{itemize}
	\item {\bf Mean analysis on general graphs from ensembles}: first, we use a counting argument similar to \cite{pedarsani2017phasecode} to show that any random graph from the ensemble $\mathcal{G}_{\rm reg}^N(R,{d})$ behaves like a {\it tree} with high probability. Therefore, the expected number of remaining edges over all graphs can be made arbitrarily close to the mean analysis $|\mathbb{E}[Z_{i}]-\mathbb{E}[Z_{i}|\mathcal{T}_{i}] |<K {d} \varepsilon/4$ such that
	\begin{align}\label{mean_on_general_graph}
		\mathbb{E}[Z_{i}]<K {d} \varepsilon/2
	\end{align}
	as long as $N$ and $K$ are greater than some constants.
	\item {\bf Concentration to mean by large deviation analysis}: we use a Doob martingale argument as in \cite{richardson2001capacity} to show that the actual number of remaining edges $Z_{i}$ concentrates well around its mean $\mathbb{E}[Z_{i}]$ with an exponential tail in $K$ such that $\Prob{\left|Z_{i}-\mathbb{E}[Z_{i}]\right|>K {d} \varepsilon/2}  \leq 2\exp\left(-c_4 \varepsilon^2 K^{\frac{1}{4i+1}}\right)$ for some constant $c_4>0$.
\end{itemize}
Then finally, it follows that $\Prob{\left|Z_{i}- K {d} \varepsilon/2  \right|>K {d} \varepsilon/2} \leq2\exp\left(-c_4 \varepsilon^2 K^{\frac{1}{4i+1}}\right)$.
\end{proof}

\subsubsection*{Graph Expansion for Complete Decoding}
From previous analyses, it has already been established that with high probability, our peeling decoder terminates with an arbitrarily small fraction of edges undecoded
	\begin{align}\label{undecoded_edges}
		Z_{i} &< K {d} \varepsilon,\quad \forall \varepsilon>0,
	\end{align}
where ${d}$ is the left degree. In this section, we show that all the undecoded edges can be completely decoded if the sub-graph consisting of the remaining undecoded edges is a ``good-expander''. First, we introduce the concept of graph expanders.
\begin{defi}[Expander Graph]
A bipartite graph with $K$ left nodes and regular left degree ${d}$ is called a $(\varepsilon,1/2)$-expander if for all subsets $\mathcal{S}$ of left nodes with $|\mathcal{S}|\leq \varepsilon K$, there exists a right neighborhood of $\mathcal{S}$ in the graph, denoted by $\mathcal{N}(\mathcal{S})$, that satisfies $|\mathcal{N}(\mathcal{S})| > {d} |\mathcal{S}|/2$.
\end{defi}

\begin{lem}\label{lem_graph_expander}
For a sufficiently small constant $\varepsilon>0$ and $d\geq 3$, the pruned graph of $\mathcal{G}_{\rm reg}^N(R,{d})$ resulting from any given $K$-sparse signal $\mathbf{x}$ is an $(\varepsilon,1/2)$-expander with probability at least $1-{O}(1/K)$.
\end{lem}
\begin{proof}
	See Appendix \ref{sec:expander_graph}.
\end{proof}

%\begin{cor}
%If the sub-graph consisting of the remaining $Z_{i}$ undecoded edges is a $(\varepsilon,1/2)$-expander, then all the undecoded edges will be peeled off completely.
%\end{cor}
Without loss of generality, let the $Z_{i}$ undecoded edges be connected to a set of left nodes $\mathcal{S}$. Since each left node has degree ${d}$, it is obvious from \eqref{undecoded_edges} that $|\mathcal{S}| = Z_{i}/d < K\varepsilon$ with high probability. Note that our peeling decoder fails to decode the set $\mathcal{S}$ of left nodes if and only if there are no more single-ton right nodes in the neighborhood of $\mathcal{S}$. A sufficient condition for all the right nodes in $\mathcal{N}(\mathcal{S})$ to have at least one single-ton is that the average degree of the right nodes in the set $\mathcal{N}(\mathcal{S})$ is strictly less than $2$, which implies that $|\mathcal{S}|{d}/|\mathcal{N}(\mathcal{S})| < 2$ and hence $|\mathcal{N}(\mathcal{S})| > |\mathcal{S}|{d}/2$. Since we have shown in Lemma \ref{lem_graph_expander} that any pruned graph from the regular ensemble $\mathcal{G}_{\rm reg}^N(R, {d})$ is a $(\varepsilon,1/2)$-expander with high probability such that $|\mathcal{N}(\mathcal{S})| > {d} |\mathcal{S}|/2$, there will be sufficient single-tons to peel off all the remaining edges. 
\begin{thm}\label{thm_peeling_decoder}
Given the ensemble $\mathcal{G}_{\rm reg}^N(\eta K, {d})$ with $d\geq 3$ and $\eta$ chosen based on Table \ref{Table_beta}, the oracle-based peeling decoder peels off all the edges in the pruned graph in ${O}(K)$ iterations with probability at least $1-{O}(1/K)$.
\end{thm}
\begin{proof} 
The oracle-based peeling decoder fails when: (1) the number of remaining edges in the $i$-th iteration cannot be upper bounded as $Z_i < K {d} \epsilon$ as in \eqref{concentration_DE}, or (2) the number of remaining edges can be upper bounded by $Z_i < K {d} \epsilon$ as in \eqref{undecoded_edges} but the remaining sub-graph is not a $(\varepsilon,1/2)$-expander. Event (1) occurs with an exponentially small probability so the total error probability is dominated by event (2). From Lemma \ref{lem_graph_expander}, we have that event (2) occurs with probability ${O}(1/K)$, which approaches $0$ asymptotically. Last but not least, since there are a total of ${O}(K)$ edges in the pruned graph, and there is at least one edge being peeled off in each iteration with high probability, the total number of iterations required to peel of the graph is ${O}(K)$.
\end{proof}
%Now, we only need to show for any small sets of variable nodes up to size $\varepsilon_\star K$, the sub-graph consisting of these left nodes is a good expander that has many right neighbors. Let $\mathcal{S}$ be any set of variable nodes with $|\mathcal{S}|\leq \varepsilon_\star K$. If the number of right neighbors of $\mathcal{S}$ is greater than $\bar{d} |\mathcal{S}|/2$ with $\bar{d}$ being the average degree of the nodes in $\mathcal{S}$, then at least one of these check nodes has at most one neighbor in $\mathcal{S}$ with probability at least $1-{O}(1/K)$. Therefore, the peeling can continue because of the existence of single-tons.

\section{Noiseless Recovery}\label{sec:noiseless}

In the noiseless setting, we consider a different graph ensemble to construct the coding matrix $\mathbf{H}$. If we use the regular graph ensemble $\mathcal{G}_{\rm reg}^N(R,{d})$ mentioned earlier to construct the coding matrix $\mathbf{H}$, the measurement cost is $M=RP$ with $R=\eta K$. Since each node has at least $P=2$ measurements from the bin detection matrix $\mathbf{S}$, the measurement cost would be at least $2\eta K$. According to Table \ref{Table_beta}, given sufficiently large $N$ and $K$, the minimum achievable $\eta$ for successful decoding is $\eta=1.23$ when ${d}=3$, and hence the minimum measurement cost is at least $M\geq 2.46K$ if the regular ensemble is used. In order to achieve the minimum redundancy parameter $\eta \rightarrow 1$, bipartite graphs with {\it irregular} left degrees need to be considered. 

\subsection{Measurement Design}

For the noiseless setting particularly, we construct the coding matrix $\mathbf{H}$ using an irregular graph ensemble rather than the regular graph ensemble $\mathcal{G}_{\rm reg}^N(R,{d})$ with better constants in our measurement costs. In the irregular graph ensemble $\mathcal{G}_{\rm irreg}^N(R,D)$, each left node has irregular left degrees ${j}=2,\cdots,D+1$, where $D+1$ is the maximum left degree. To describe the construction of the irregular graph ensemble, we use the left degree sequence $\{\lambda_{j}\}_{{j}=2}^{D+1}$, where $\lambda_{j}$ is the fraction of edges\footnote{The graph is specified in terms of fractions of edges of each degree due to its notational convenience later on.} of degree ${j}$ on the left\footnote{An edge of degree ${j}$ on the left (right) is an edge connecting to a left (right) node with degree ${j}$.}. For instance, the left degree sequence for the regular ensemble $\mathcal{G}_{\rm reg}^N(R,{d})$ is $\lambda_{j}=1$ for ${j}={d}$ and $0$ if $j\neq d$.

\begin{defi}[Irregular Graph Ensemble $\mathcal{G}_{\rm irreg}^N(R,D)$ for Noiseless Recovery]\label{lem_graph_ensemble_irregular}
Given $N$ left nodes and $R=(1+\epsilon)K$ right nodes for an arbitrary $\epsilon>0$, the edge set in the irregular graph ensemble $\mathcal{G}_{\rm irreg}^N(R,D)$ is characterized by the degree sequence
\begin{align}
	\lambda_{j} = \frac{1}{H(D)({j}-1)},\quad {j}= 2,\cdots,D+1
\end{align}
where $D >1/\epsilon$ and $H(D)=\sum_{{j}=1}^{D} {1}/{j}$ is chosen such that $\sum_{{j}\geq 2} \lambda_{j}=1$.
%\begin{itemize}
%	\item a maximum left node degree $D\geq 1/(1-\eta)$;
%	\item an edge degree sequence $\lambda_d$ given by
%	\begin{align}
%		\lambda_d = \frac{1}{H(D)(d-1)},\quad d= 2,\cdots,D+1
%	\end{align}
%	where $H(D)=\sum_{d=1}^{D} {1}/{d}$.
%	\item an average left degree $\bar{d} = H(D)(1+{1}/{D})$.
%\end{itemize}
\end{defi}

\begin{thm}\label{thm_peeling_decoder_irregular}
Consider the ensemble $\mathcal{G}_{\rm irreg}^N(R, D)$ for our construction. The oracle-based peeling decoder peels off all the edges in the pruned graph in ${O}(K)$ iterations with probability at least $1-{O}(1/K)$.
\end{thm}
\begin{proof}
    See Appendix \ref{sec:proof_thm_peeling_irregular}.
\end{proof}

Given the coding matrix $\mathbf{H}$ constructed from the irregular ensemble, we choose the {\it bin detection matrix} $\mathbf{S}$ as 
\begin{align}\label{Fourier_sensing_koiseless}
	\mathbf{S}
	\defn
	\begin{bmatrix}
		1 & \cdots & 1 & \cdots & 1\\
		1 & \cdots & W^n & \cdots & W^{N-1}
	\end{bmatrix}
	\times
	\diag{{F}_0,{F}_1,\cdots,{F}_{N-1}}, 
\end{align}
where $W=e^{\mathrm{i}\frac{2\pi}{N}}$ is the $N$-th root of unity and ${F}_k$ for $k\in[N]$ is a random variable drawn from some continuous distribution. The bin detection matrix is therefore the first $2$ rows of the $N\times N$ DFT matrix with each column scaled by a random variable. This is similar to the example we used in Section \ref{sec:ratio_test}, except for the random scaling on each column. We have briefly shown in Section \ref{sec:ratio_test} how to obtain the oracle information in the noiseless setting using a similar bin detection matrix. In the following, we restate the procedures more formally to be self-contained. 

Using the two measurements in each bin $\mathbf{y}_r=[y_r[0],y_r[1]]^T$ for $r=1,\cdots,R$, we perform the following tests to reliably identify the single-ton bins and obtain the correct index-value pair for any single-ton:

\begin{itemize}
	\item {\bf Zero-ton Test}: since there is no noise, it is clear that the bin is a zero-ton if $\left\|\mathbf{y}_r\right\|^2=0$.
	\item {\bf Multi-ton Test}:
The measurement bin is a multi-ton as long as $|y_r[1]|\neq |y_r[0]|$ and/or $\angle{y_r[1]}/{y_r[0]} \neq {0\mod 2\pi/N}$. The multi-ton test fails when the relative phase is a multiple of $2\pi/N$, which corresponds to the following condition according to the measurement model in \eqref{divide-and-conquer}
\begin{align}
	\frac{y_r[1]}{y_r[0]} = \frac{\sum_{k\in[N]} H_{r,k} x[k] {F}_k e^{\mathrm{i}\frac{2\pi n}{N}}}{\sum_{k\in[N]} H_{r,k} x[k] {F}_k} = e^{\mathrm{i}\frac{2\pi \ell}{N}},\quad \textrm{for some $\ell\in[N]$}
\end{align}
where $H_{r,k}$ is the $(r,k)$-th entry in the coding matrix $\mathbf{H}$. Clearly,  this event is measure zero under the continuous distribution of ${F}_k$ for $k\in[N]$.
	\item {\bf Single-ton Test}:
	After the zero-ton and multi-ton tests, if $|y_r[1]|=|y_r[0]|$ and $\angle {y_r[1]}/{y_r[0]} = {0\mod 2\pi/N}$, the measurement bin is detected as a single-ton with the index-value pair:
	\begin{align}
		\widehat{k}_r &= \frac{N}{2\pi}\angle\frac{y_r[1]}{y_r[0]},\quad
		\widehat{x}[\widehat{k}_r] = y_r[0]/F_{\widehat{k}_r}.
	\end{align}
This gives us the index-value pair of the single-ton for peeling. 

\end{itemize}

\subsection{Some Numerical Examples}

\subsubsection*{Density Evolution Threshold}

%.
%\begin{wrapfigure}{r}{0.55\textwidth}
%\vspace{-0.8cm}
\begin{figure}[h]
\begin{center}
\includegraphics[width=0.5\linewidth]{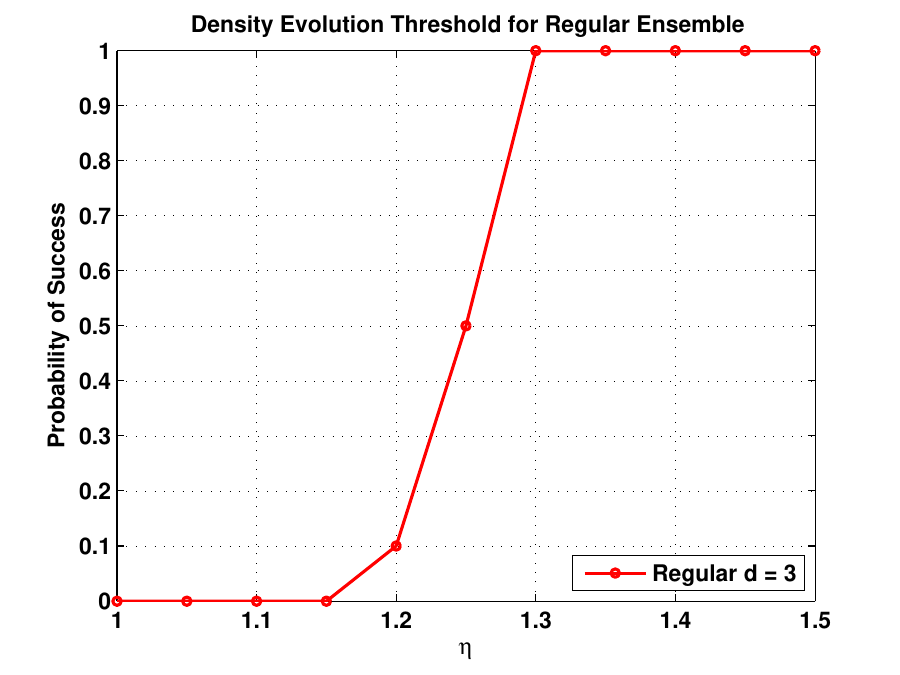}
\vspace{-0.6cm}
\caption{Probability of success against the redundancy parameter $\eta$ for the regular ensemble $\mathcal{G}_{\rm reg}^N(\eta K,3)$ with $N=0.1~\mathrm{million}$.}
\label{fig:density_evolution}
\end{center}
\end{figure}
%\vspace{-0.5cm}
%\end{wrapfigure}	
We examine the density evolution result using the noiseless design in Section \ref{sec:noiseless} in the absence of noise. We generate a sparse vector $\mathbf{x}$ with $K=500$ and $N=10^5$ for all the experiments. To understand the effects of the graph ensemble on density evolution, we numerically trace the probability of success $1-\Pf$ against the redundancy parameter $\eta=R/K$ of the regular graph ensemble $\mathcal{G}_{\rm reg}^N(R,{d})$. For simplicity, we fix the left node degree ${d}=3$ and vary the redundancy parameter $\eta=R/K$ from $1$ to $1.5$. It can be seen that the threshold for $R/K=\eta$ empirically matches with the density evolution analysis for regular graphs in Section \ref{sec:peeling_decoder_analysis}, where the algorithm succeeds with some probability from $\eta = 1.2$ and reaches probability one after $\eta=1.3$. 

\subsubsection*{Illustration of Density Evolution}
We demonstrate the density evolution process by showing the peeling iterations of recovering a $280\times 280$ grayscale ``Cal'' image consisting of pixels taking values within $[0,1]$. In this setting, we have the input dimension $N=280\times 280=78400$ and the sparsity $K=3600$, and the image in \figref{fig:original_image} is free from noise. To recover this Cal image using our framework, we exploit the noiseless design in Section \ref{sec:noiseless}. In particular, the coding matrix $\mathbf{H}$ is constructed using the regular graph ensemble $\mathcal{G}_{\rm reg}^N(R, d)$ with a regular degree ${d}=3$ and a redundancy $R=1.5K$, while the bin detection is the first two rows of an $N$-point DFT matrix such that $P=2$. Therefore, the total measurement cost is $M=RP=3K=10800\approx N\times 13.7\%$. It can be seen from \figref{fig:density_evolution_peeling} that when the density evolution threshold is met $\eta = 1.5 > 1.23$, the image is quickly recovered from a few iterations, where the first $3$ iterations almost capture most of the sparse coefficients while iteration $4$ and $5$ are cleaning up the very few remaining coefficients.

\begin{figure}[h]
\begin{center}
\begin{subfigure}{0.32\textwidth}
\includegraphics[width=1\linewidth]{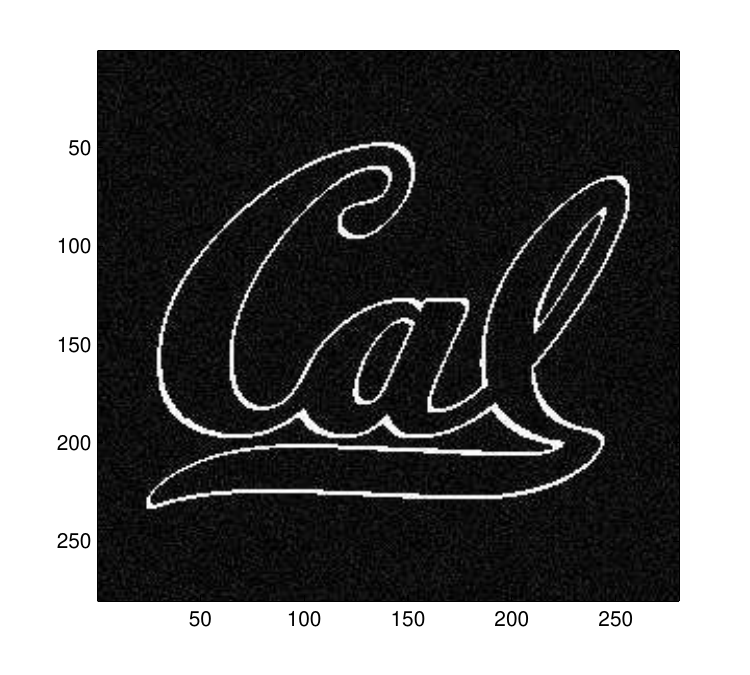}
\caption{Original Image}\label{fig:original_image}
\end{subfigure}
\begin{subfigure}{0.32\textwidth}
\includegraphics[width=1\linewidth]{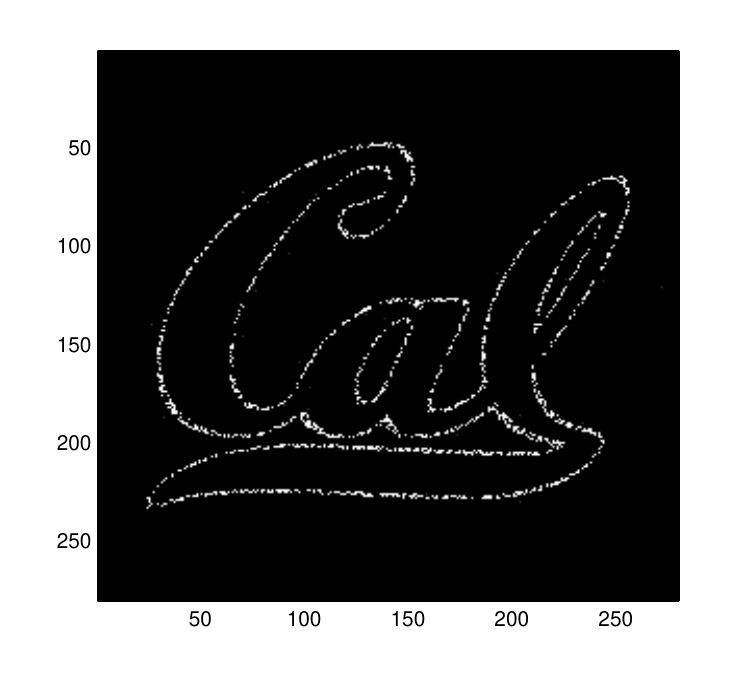}
\caption{Peeling Iteration 1}
\end{subfigure}
\begin{subfigure}{0.32\textwidth}
\includegraphics[width=1\linewidth]{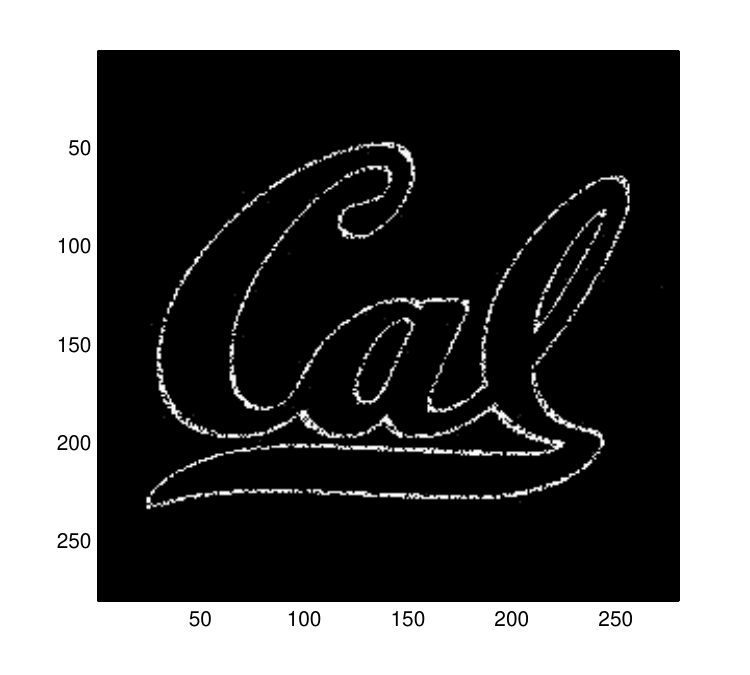}
\caption{Peeling Iteration 2}
\end{subfigure}
\begin{subfigure}{0.32\textwidth}
\includegraphics[width=1\linewidth]{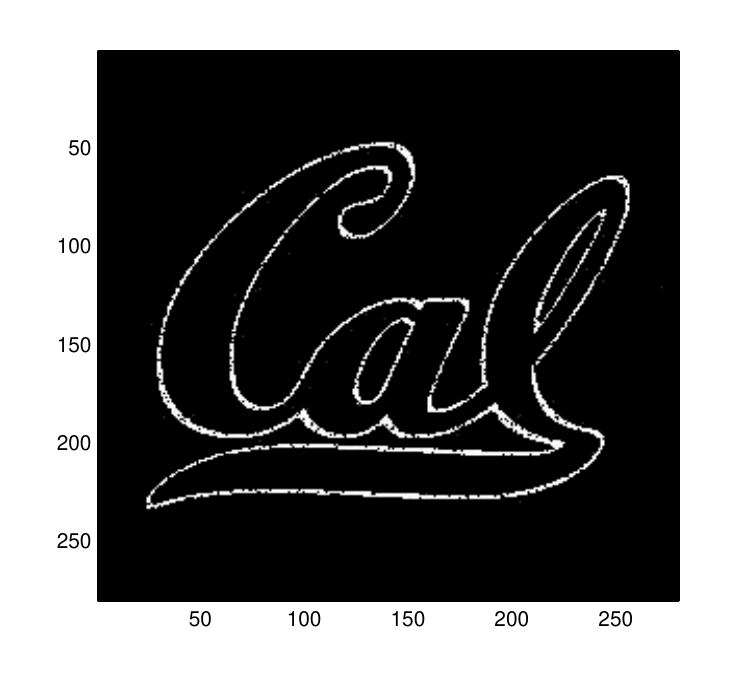}
\caption{Peeling Iteration 3}
\end{subfigure}
\begin{subfigure}{0.32\textwidth}
\includegraphics[width=1\linewidth]{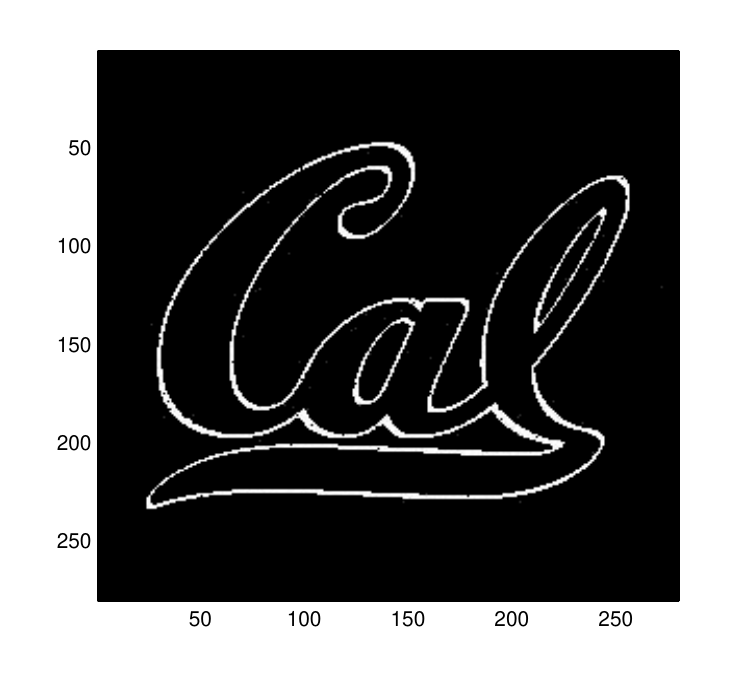}
\caption{Peeling Iteration 4}
\end{subfigure}
\begin{subfigure}{0.32\textwidth}
\includegraphics[width=1\linewidth]{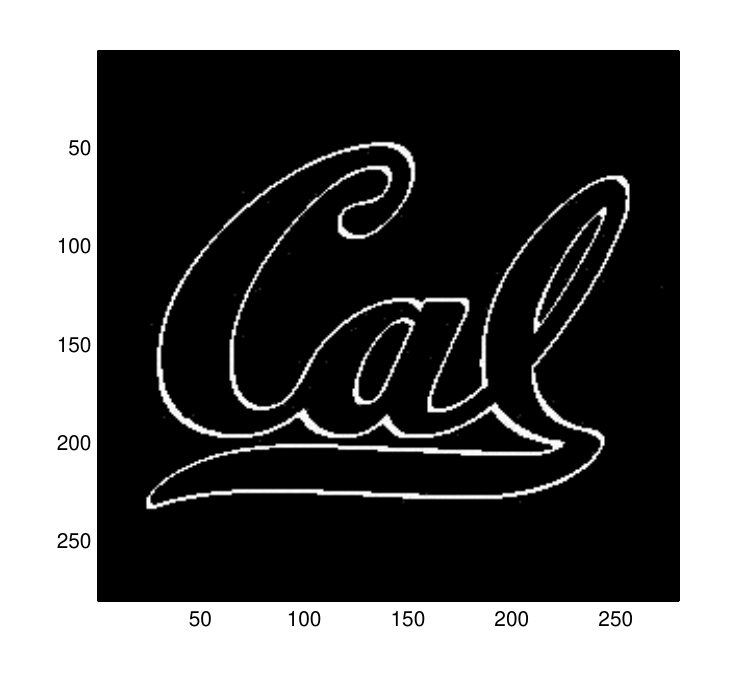}
\caption{Peeling Iteration 5}
\end{subfigure}
\caption{Illustration of density evolution through peeling iterations over the recovery of the ``
Cal'' image}
\label{fig:density_evolution_peeling}
\end{center}
\end{figure}	

\section{Noisy Recovery in the Quantized Alphabet Setting}\label{sec:noisy}
In this section, we extend the noiseless design to the noisy design in the quantized alphabet setting. More specifically, we assume that all the sparse coefficients in $\vecx$ are elements in a finite set $\mathcal{X}=\{\pm \rho, \pm 2\rho, \ldots, \pm B\rho\}$. We first discuss the construction of the coding matrix $\mathbf{H}$. Note that we can certainly use the irregular graph ensemble as in the noiseless case to design our coding matrix $\mathbf{H}$ for the noisy case as well, because it gives sharper measurement bounds. However, since we are providing order-wise results for the measurement costs, we consider the regular graph ensemble $\mathcal{G}_{\rm reg}^N(R, {d})$ for constructing $\mathbf{H}$ because of its simplicity. In the following, we discuss the constructions of the {\it bin detection matrix} $\mathbf{S}$ in the noisy setting. 

Since the procedures are the same for any measurement bin at any iteration, we drop the bin index $r$ in \eqref{divide-and-conquer} and use the italic font $\bdsb{y}$ to denote a generic bin measurement $\mathbf{y}_r$ using the following model
\begin{align}\label{equiv.model.bin}
	\bdsb{y}  
	%&= \sum_{\bdsb{k} \in \mathcal{K}} x[k] \mathbf{s}_k + \mathbf{w}
	= \mathbf{S}\mathbf{z} + \mathbf{w}
\end{align}
for some bin detection matrix $\mathbf{S}=[\mathbf{s}_0,\cdots,\mathbf{s}_{N-1}]$ and some sparse vector $\mathbf{z}$. For example, in the first iteration at bin $r$, the sparse vector equals $\mathbf{z}=\mathbf{z}_r$ given in \eqref{divide-and-conquer}. As the peeling iterations proceed, the non-zero coefficients in $\mathbf{z}$ will be peeled off and potentially left with a $1$-sparse coefficient. Therefore, at each iteration, we perform the bin detection routine to verify if $\mathbf{z}$ has become a $1$-sparse signal (i.e. resolve the bin hypothesis) and obtain the associated index-value pair $(\widehat{k},\widehat{x}[\widehat{k}])$. In the presence of noise, we propose the following robust detection scheme for each bin.

\begin{defi}[\bf Robust Bin Detection Algorithm]\label{defi_RBI_Algorithm}
The detection is performed in a ``guess-and-check'' manner as:
\begin{itemize}
%	\item[\bf Step 1)] {\bf zero-ton verification}: {\it is it just noise?} Given the noisy measurements, the decoder needs to first detect whether there is indeed an existing signal component (if not, there is no decoding necessary). To do so, it is expected that the energy $\|\bdsb{y}\|^2$ to be small relative to the energy of a single-ton. Therefore, this idea is used to verify zero-tons:
%		\begin{align}
%			\bdsb{y}\sim\mathcal{H}_{\textrm{Z}}, ~~~~\mathrm{if}~\frac{1}{P}\left\|\bdsb{y}\right\|^2\leq 
%			\tau_0
%		\end{align}
%		where $\tau_0$ is some threshold specified later.
%		
%
%
	\item[\bf Step 1)] {\bf single-ton search} $\psi: \bdsb{y} \rightarrow (\widehat{k}, \widehat{x}[\widehat{k}])$ estimates the index-value pair $(\widehat{k},\widehat{x}[\widehat{k}])$ assuming that the underlying bin is a single-ton. This procedure depends on the bin detection matrix $\mathbf{S}$, and is explained in the next section.
	\item[\bf Step 2)] {\bf single-ton verification} determines whether the single-ton assumption is valid using the estimates $(\widehat{k},\widehat{x}[\widehat{k}])$: 
\begin{align}\label{singleton-verification}
	\bdsb{y}\sim\mathcal{H}_{\textrm{S}}(\widehat{k},\widehat{x}[\widehat{k}])\quad
	&\textrm{if}~\frac{1}{P}\left\|\bdsb{y} - \widehat{x}[\widehat{k}] \mathbf{s}_{\widehat{k}} \right\|^2 
	\leq \left(1+ \gamma \SNRmin\right) \times \sigma^2,
\end{align}
where $\gamma \in (0,1)$ is some constant, and $\SNRmin = \rho^2/\sigma^2$.	
\end{itemize}

\end{defi}

This ``guess-and-check'' procedure is already manifested in the noiseless design, where the bin detection matrix $\mathbf{S}$ leads to a simple ratio test to accomplish both the single-ton search and verification. More specifically, the matrix $\mathbf{S}$ from the noiseless design is a properly chosen {\it codebook} for encoding the unknown value and location of the $1$-sparse coefficient, where each column of $\mathbf{S}$ is a {\it codeword}. On one hand, the first row of both designs is an all-one vector, which captures directly the unknown value (but not the index). On the other hand, the noiseless design encodes the index information into a single $N$-PSK symbol (i.e. $W^k=e^{-\mathrm{i}\frac{2\pi k}{N}}$ for $k\in[N]$). The perspective of treating $\mathbf{S}$ as a codebook is very insightful for designing the single-ton search for the noisy scenario, where the goal is to decode the index-value pair (i.e. the codeword transmitted $\mathbf{s}_k$) from its noisy observation $\bdsb{y}$ through a Gaussian channel with an unknown channel gain $x[k]$ (see \figref{fig:single-ton_search}).
%
%\begin{wrapfigure}{r}{0.5\textwidth}
%\begin{minipage}{0.52\textwidth}
\begin{figure}[h]
\begin{center}
\includegraphics[scale=0.26]{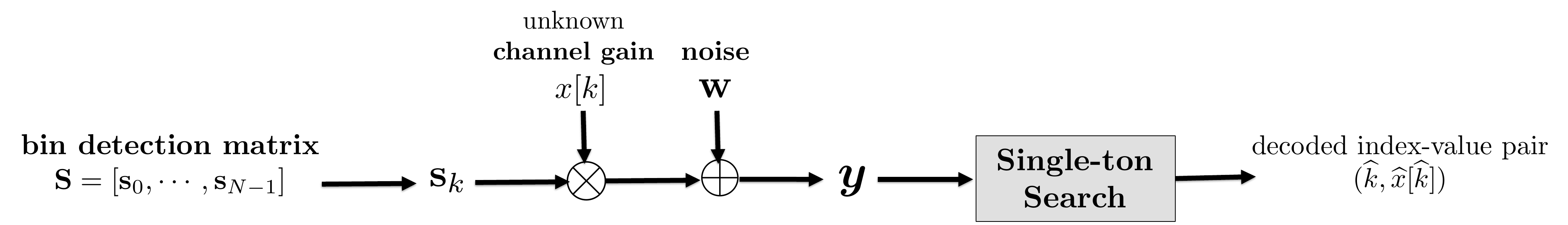}
\caption{An illustration of the single-ton search.}
\label{fig:single-ton_search}
\end{center}
\end{figure}
%\end{minipage}
%\end{wrapfigure}
%

%
%
%\begin{wrapfigure}{r}{0.5\textwidth}
%\vspace{-0.8cm}
%\begin{minipage}{0.5\textwidth}
%\begin{algorithm}[H]
%% \DontPrintSemicolon
%% \SetAlgoNoLine
%%  \SetKwFor{ParFor}{for}{in parallel}{end for}
%  \caption{ Robust Bin Detection Algorithm}\label{alg:RBI}
%  \begin{algorithmic}
%    \STATE ${\tt Input:}$ Observations $\bdsb{y}$ of some bin, $\SNRmin$ and $\sigma^2$.
%    \STATE ${\tt Set:}$ Parameter $\gamma\in(0,1)$;
%    \STATE ${\tt Output:}$ Bin type $\widehat{\mathcal{H}}$ and index-value pair $(\widehat{k},\widehat{x}[\widehat{k}])$
%    \STATE
%    \STATE {\bf perform singleton search} $\psi: \bdsb{y} \rightarrow (\widehat{k}, \widehat{x}[\widehat{k}])$	    
%	    \IF{$\left\|\bdsb{y} - \widehat{x}[\widehat{k}] \mathbf{s}_{\widehat{k}} \right\|^2/P \leq (1+\gamma\SNRmin)\sigma^2$} 
%	        \RETURN $\widehat{\mathcal{H}}=\mathcal{H}_{\textrm{S}}(\widehat{k},\widehat{x}[\widehat{k}])$
%            \ENDIF    
%  \end{algorithmic}
%\end{algorithm}
%\vspace{-0.5cm}
%\caption{Robust Bin Detection Algorithm}
%\vspace{-0.5cm}
%\end{minipage}
%\end{wrapfigure}
%

To guarantee the success of peeling in the presence of noise, the codebook needs to be designed differently from the noiseless case such that it can be robustly decoded. In the following, we first introduce a simple randomized construction for this purpose with no computational constraints, and then explain how to derive a low complexity scheme based on the randomized construction.

\subsection{A Simple Random Construction}\label{sec:random_family}

In the presence of noise, the randomized design exploits fully randomized linear codes to resolve different bin hypotheses and obtain the index-value pair.

\begin{defi}[]\label{def_random_sensing}
The $P\times N$ {\it bin detection matrix} $\mathbf{S} = \left[S_{i,j}\right]_{P \times N}$ consists of i.i.d. Gaussian entries $\mathcal{N}(0,1)$.
\end{defi}

Using this randomized construction, the single-ton search can be performed as follows. %We employ a generalized likelihood test in the {\it single-ton search}, similar to the spirit of joint-typicality decoding. 
For each possible coefficient index $k$, we obtain the maximum likelihood (ML) of the coefficient as:
\begin{align}\label{X_MLE}
	\alpha_k= \frac{\mathbf{s}_k^T \bdsb{y}}{\left\|\mathbf{s}_k\right\|^2}.
\end{align}
Substituting the estimate of the coefficient $\alpha_k$ into the likelihood of the single-ton hypothesis in Proposition \ref{prop_divide-and-conquer}, we choose the index $k$ that minimizes the residual energy:
\begin{align}\label{residual_min}
	  \widehat{k} = \arg\min_{k\in\mathcal{N}(r)}~\left\|\bdsb{y} -  \alpha_k \mathbf{s}_k \right\|^2.
\end{align}
The search is over the coding pattern in the $r$-th bin $k\in\mathcal{N}(r)$, which is known a priori. With the estimated index $\widehat{k}$, the coefficient is obtained by aligning it to the closest alphabet symbol in $\mathcal{X}$
\begin{align}\label{x_detection}
	%\widehat{x}[\widehat{k}] = \arg\min_{x\in\mathcal{X}}~\left\|\bdsb{y} - x \mathbf{s}_{k} \right\|^2
	\widehat{x}[\widehat{k}] = \min_{x\in\mathcal{X}} \left\|\alpha_{\widehat{k}}-x\right\|^2.
\end{align}

\begin{lem}[]\label{lem_random_sensing}
Using the $P\times N$ bin detection matrix $\mathbf{S}$ in Definition \ref{def_random_sensing}, the algorithm in Definition \ref{defi_RBI_Algorithm} succeeds in identifying the presence of a single-ton and its index-value pair correctly in time $O((N/K) \log (N/K))$, with probability at least $1-O(1/K^2)$ as long as $K=O(N^\delta)$ for some $\delta\in(0,1)$ and
\begin{align}
	\begin{cases}
		P \geq 16(1+\SNRmin^{-1}) \frac{(1+2\delta) }{(1- \delta)} \log \left(\frac{N}{K}\right), & \SNRmin \gg 1\\
		P \geq 16 \SNRmin^{-2}\frac{(1+2\delta) }{(1- \delta)} \log \left(\frac{N}{K}\right), & \SNRmin \ll 1
	\end{cases}.
\end{align}
\end{lem}

\begin{proof}
	See Appendix \ref{proof_lem_random_sensing}. 
\end{proof}
Since the detection scheme incurs an error with probability at most $O(1/K^2)$, the overall probability of making an error throughout the peeling iterations across $K$ bins is at most $O(1/K)$, which is on par with the error probability of the oracle-based peeling decoder. Therefore, our scheme achieves an overall failure probability of $\Pf = O(1/K)$, which approaches zero asymptotically. Now let us briefly comment on the measurement cost and computational complexity. There are a total of $R=\eta K$ bins and each bin has $P=O(\log (N/K) )$ measurements, the randomized construction leads to a measurement cost of $M=\eta K P = {O}(K\log (N/K))$. %The measurement cost ${O}(K\log N)$ is near-optimal in the order sense  with respect to the fundamental limits for support recovery studied in \cite{wainwright2009information,fletcher2009necessary,wang2010information,jin2011limits}.
%%This is implied by information theory where random codes are oftentimes capacity-achieving, but the decoding complexity is high. 
In terms of computations, this scheme requires an exhaustive search over the entire codebook in each peeling iteration. The size of the codebook for some bin (say $r$) depends on the right node degree $|\mathcal{N}(r)|$. Based on the ``balls-and-bins'' construction, this means that $|\mathcal{N}(r)|$ is well concentrated around ${O}(N/K)$ with an exponential tail. Since each codeword imposes a search complexity of $P={O}(\log (N/K))$ by the maximum likelihood single-ton search, therefore across all ${O}(K)$ peeling iterations, this results in a total complexity of $T={O}(N/K)\times {O}(\log (N/K)) \times {O}(K) = {O}(N \log (N/K))$. 

\subsection{Noisy Bin Detection: Going below Linear Time}

The randomized construction is slow because it does not optimize its choice of codebook to facilitate the decoding procedure of {\bf Step (1)} in Definition \ref{defi_RBI_Algorithm}, which causes the high complexity. The question to ask is: {\it is it possible to maintain similar performances with a run-time complexity that is sub-linear in $N$}?  To reduce the complexity without compromising the measurement cost, the spirit of divide-and-conquer also applies. We use two codebooks, where one uses the randomized construction to deal with single-ton verifications, while the other codebook (introduced next) deals with the single-ton search, which is the key to our fast algorithm. 

\subsubsection{Motivating Example in the Noiseless Case}
To motivate our noisy design, we consider another coding scheme in the noiseless case, where the bin detection matrix is constructed as
\begin{align}
	\mathbf{S} = (-1)^{\mathbf{B}},
\end{align}
where $\mathbf{B}=\begin{bmatrix}\mathbf{b}_0 & \mathbf{b}_1 & \cdots & \mathbf{b}_{N-1}\end{bmatrix}$ is the binary expansion matrix with $n=\lceil\log_2 N\rceil$ such that each column $\mathbf{b}_k$ is an $n$-bit binary representation for all $k\in[N]$. In our running example $N=16$, the $4\times 16$ binary expansion matrix is 
\begin{align}\label{binary_exp_eq}
	\mathbf{B}=
	\begin{bmatrix}
	0 & 0 & 0 & 0 & 0 & 0 & 0 & 0 &\cdots & 1\\
	0 & 0 & 0 & 0 & 1 & 1 & 1 & 1 &\cdots & 1\\
	0 & 0 & 1 & 1 & 0 & 0 & 1 & 1 &\cdots & 1\\
	0 & 1 & 0 & 1 & 0 & 1 & 0 & 1 &\cdots & 1
	\end{bmatrix}
\end{align}
and the bin detection matrix is:
\begin{align}
	\mathbf{S}=
	\begin{bmatrix}
	(-1)^0 & (-1)^0 & (-1)^0 & (-1)^0 &\cdots & (-1)^1\\
	(-1)^0 & (-1)^0 & (-1)^0 & (-1)^0 &\cdots & (-1)^1\\
	(-1)^0 & (-1)^0 & (-1)^1 & (-1)^1 &\cdots & (-1)^1\\
	(-1)^0 & (-1)^1 & (-1)^0 & (-1)^1 &\cdots & (-1)^1
	\end{bmatrix}.
\end{align}

For simplicity, we assume that the values are all known $x[k]=1$ for $k\in\supp{\mathbf{x}}$ but the locations $k$ are unknown. Later we explain how to get rid of this assumption. Given this bin detection matrix and that all $x[k]=1$ by assumptions, right nodes $1$, $2$ and $3$ are associated with measurements $\mathbf{y}_1=\mathbf{0}$, 
\begin{align*}
	\mathbf{y}_2&=
	%x[1]
	\begin{bmatrix}
	(-1)^0\\
	(-1)^0\\
	(-1)^0\\
	(-1)^1
	\end{bmatrix}
	+
	%x[5]
	\begin{bmatrix}
	(-1)^0\\
	(-1)^1\\
	(-1)^0\\
	(-1)^1
	\end{bmatrix}
	+
	%x[13]
	\begin{bmatrix}
	(-1)^1\\
	(-1)^1\\
	(-1)^0\\
	(-1)^1
	\end{bmatrix},
	~
	\mathbf{y}_3=
	%x[10]	
	\begin{bmatrix}
	(-1)^1\\
	(-1)^0\\
	(-1)^1\\
	(-1)^0
	\end{bmatrix}.
\end{align*}
Now, one can easily determine if a right node is a zero-ton, a single-ton or a multi-ton easily. Consider the right node $3$. A single-ton can be verified by checking if $|y_3[1]|=\cdots=|y_3[4]|$ and the unknown index can be obtained by taking the sign\footnote{The sign function is defined slightly different from the usual case:
	\begin{align}
		\sgn{x} = 
		\begin{cases}
			1, & x<0\\
			0, & x \geq 0.
		\end{cases}
	\end{align}} of each measurement $\sgn{y_3[p]}$ such that
	\begin{align}
		\widehat{k} &= \sum_{p=1}^{n} 2^{p-1} \times \sgn{y_3[p]}.
	\end{align} 
On the other hand, consider the measurement $\mathbf{y}_2$ from right node 2. Since it does not satisfy the above criterion, it can be concluded as a multi-ton. 

In the general noiseless case where $x[k]$ is unknown, we can easily modify the simple case by concatenating an extra ``all-one'' row vector with the bin detection matrix $\mathbf{S}$ as
\begin{align}
	\mathbf{S}=
	\begin{bmatrix}
	1        &      1    &       1  &    1     & \cdots & 1\\
	(-1)^0 & (-1)^0 & (-1)^0 & (-1)^0 &\cdots & (-1)^1\\
	(-1)^0 & (-1)^0 & (-1)^0 & (-1)^0 &\cdots & (-1)^1\\
	(-1)^0 & (-1)^0 & (-1)^1 & (-1)^1 &\cdots & (-1)^1\\
	(-1)^0 & (-1)^1 & (-1)^0 & (-1)^1 &\cdots & (-1)^1
	\end{bmatrix}.
\end{align}
Using this bin detection matrix, for the single-ton right node $3$, we would have $$\mathbf{y}_3= x[10]\times	
\begin{bmatrix}1, (-1)^1, (-1)^0, (-1)^1, (-1)^0 \end{bmatrix},$$ which gives us $y_3[0] = x[5]$ and the unknown index $k$ can be obtained as:
\begin{align}
	\widehat{k} &= \sum_{p=1}^{n} 2^{p-1} \times \sgn{y_3[p]}\oplus \sgn{y_3[0]}.
\end{align} 
However, in the presence of noise, these tests no longer work as an oracle. Next we explain how to robustify this coding scheme in the presence of noise.

\subsubsection{General Design in the Noisy Case}\label{sec:noisy_oracle}
In the noiseless case, each codeword in $\mathbf{S}$ is the bipolar $\{\pm 1\}$ image of the corresponding binary code $\mathbf{b}_k$ of the column index $k$, and hence it is not difficult to decode the transmitted message $\mathbf{b}_k$ and recover $k$. However, in the presence of noise, the codebook needs to be re-designed such that it can be robustly decoded.
\begin{defi}[\bf Bin Detection Matrix]\label{def_binary_sensing}
Let $\mathbf{B}$ be the $n \times N$ binary expansion matrix in \eqref{binary_exp_eq} with $n=\lceil\log_2N\rceil$, where the bin detection matrix is constructed as $\mathbf{S} = [\mathbf{S}_0^T,\mathbf{S}_1^T,\mathbf{S}_2^T]^T$, and 
\begin{itemize}
	\item $\mathbf{S}_0=\mathbf{1}_{P\times N}$ is an all-one codebook;
	\item $\mathbf{S}_1=(-1)^{\mathbf{C}}$ and $\mathbf{C}=[\mathbf{c}_0,\cdots,\mathbf{c}_{N-1}]$ is a $P\times N$ linear channel codebook constructed as $\mathbf{C}=\mathbf{G}\mathbf{B}$ by a $P\times n$ generator matrix with a block length $P$, as well as a decoding error probability of $e^{-\zeta P}$ for some error exponent $\zeta>0$; 
	\item $\mathbf{S}_2=[\mathbf{s}_{2,0},\cdots,\mathbf{s}_{2,N-1}]$ is a $P\times N$ random codebook consisting of i.i.d. Rademacher entries $\{\pm 1\}$.
\end{itemize}
\end{defi}

There exist many codes that satisfy the our requirement (strictly positive error exponent), but the challenge is the decoding time. It is desirable to have a decoding time that is linear in the block length $P={O}(n)$ so that the sample complexity and computational complexity can be maintained at ${O}(n)$ for each bin, same as the noiseless case. Excellent examples include the class of {\it expander codes} or (spatially coupled) {\it LDPC codes} that allow for linear time decoding. 
%The reason for choosing expander codes for $\mathbf{S}_1$ is due to its minimum distance properties and linear decoding time (with respect to its block length) using the bit flipping algorithm \cite{richardson2008modern}. 
With this design, we obtain three measurement sets in each bin $\bdsb{y}=[\mathbf{u}_0^T,\mathbf{u}_1^T,\mathbf{u}_2^T]^T$:
\begin{align}
	\mathbf{u}_i = \mathbf{S}_i\mathbf{z} + \mathbf{w}_i,\quad i=0,1,2.
\end{align}
Each measurement set is used differently in the ``guess-and-check'' procedure mentioned in Definition \ref{defi_RBI_Algorithm}.

The {\bf single-ton verification} simply uses the measurement set $\mathbf{u}_2$ to confirm whether the bin is a single-ton, as summarized in \algref{alg:RBI_binary}, while the {\bf single-ton search} uses $\mathbf{u}_0$ and $\mathbf{u}_1$ differently. The single-ton search uses the measurement set $\mathbf{u}_0$ for obtaining the estimate $\widehat{\alpha}$ of $x[k]$, and the measurement set $\mathbf{u}_1$ for obtaining the estimate $\widehat{k}$ of the index $k$. If the underlying bin is indeed a single-ton with an index-value pair $(k,\alpha)$, then the measurement $\mathbf{u}_1$ is the noisy version of some coded message $\mathbf{c}_k=\mathbf{G}\mathbf{b}_k$% (see \figref{fig:single-ton_search}).
\begin{align}
	\mathbf{u}_1
	= 
	\alpha (-1)^{\mathbf{G}\mathbf{b}_k}+\mathbf{w}_1,
\end{align}
where $\mathbf{b}_k$ is the $k$-th column of the binary expansion matrix $\mathbf{B}$.

\begin{prop}\label{prop_BSC}
Given a single-ton bin with an index-value pair $(k,\alpha)$, the sign of the measurement set $\mathbf{u}_1$ satisfies
\begin{align}
	\sgn{\mathbf{u}_1}
	&= \mathbf{G}\mathbf{b}_k \oplus\sgn{\alpha} \oplus \mathbf{e},
\end{align}
where $\mathbf{e}$ is a binary vector containing $P$ bit flips with a cross probability upper bounded as $\Pe=e^{-\frac{|x[k]|^2}{2\sigma^2}}$.
\end{prop}
\begin{proof}
	The proof can be obtained by Gaussian tail bounds, and hence we omit it here due to lack of space.
\end{proof}

%\begin{wrapfigure}{r}{0.45\textwidth}
%\vspace{-0.6cm}
%\begin{minipage}{0.45\textwidth}
\begin{algorithm}[h]
% \DontPrintSemicolon
% \SetAlgoNoLine
%  \SetKwFor{ParFor}{for}{in parallel}{end for}
  \caption{ Robust Bin Detection Algorithm}\label{alg:RBI_binary}
  \begin{algorithmic}
    \STATE ${\tt Input:}$ Observation $\bdsb{y}=[\mathbf{u}_0^T,\mathbf{u}_1^T,\mathbf{u}_2^T]^T$, $\SNRmin$ and $\sigma^2$.
    \STATE ${\tt Set:}$ $\gamma\in(0,1)$ and generator matrix $\mathbf{G}$.
    \STATE ${\tt Output:}$ the index-value pair $(\widehat{k},\widehat{x}[\widehat{k}])$
    \STATE {\bf obtain the coefficient} from $\mathbf{u}_0$:
	\begin{align}\label{x_MLE}
		\widehat{\alpha} = \min_{x\in\mathcal{X}} \left\|\mathbf{u}_0 - x \mathbf{1}_P \right\|^2
	\end{align}
    \STATE {\bf estimate} the index $\mathbf{b}_{\widehat{k}}$ via channel decoding over $\sgn{\mathbf{u}_1} \oplus\sgn{\widehat{\alpha}}  = \mathbf{G}\mathbf{b}_k\oplus \mathbf{e}$
    \STATE {\bf obtain} $\widehat{k}$ from $\mathbf{b}_{\widehat{k}}=[b_{\widehat{k}}[1],\cdots,b_{\widehat{k}}[n]]^T$ such that $\widehat{k} = \sum_{p=1}^{n} 2^{p-1} \times b_{\widehat{k}}[p]$.
	    \IF{$\|\mathbf{u}_2 - \widehat{\alpha} \mathbf{s}_{2,\widehat{k}} \|^2/P \leq (1+\gamma\SNRmin)\sigma^2$} 
	        \RETURN $(\widehat{k},\widehat{x}[\widehat{k}])$
            \ENDIF    
  \end{algorithmic}
\end{algorithm}
%\vspace{-0.1cm}
%\caption{Robust Bin Detection Algorithm}
%\vspace{-0.1cm}
%\end{minipage}
%\end{wrapfigure}

Although $\alpha$ is unknown, it can be estimated using $\mathbf{u}_0$ using \eqref{x_MLE} and therefore, we have $\sgn{\mathbf{u}_1} \oplus\sgn{\widehat{\alpha}}  = \mathbf{G}\mathbf{b}_k\oplus \mathbf{e}$. Because the index $k$ can be obtained from $\mathbf{b}_k$ directly, we only need to decode $\mathbf{b}_k$ reliably over a binary symmetric channel (BSC) with a cross probability $\Pe$. 

\begin{lem}[]\label{lem_binary_sensing}
Using the bin detection matrix $\mathbf{S}$ in Definition \ref{def_binary_sensing}, the algorithm in Definition \ref{defi_RBI_Algorithm} succeeds in identifying the presence of a single-ton and its index-value pair correctly with probability at least $1-O(1/K^2)$ as long as $K=O(N^\delta)$ and $P={O}(\log (N/K))$.
\end{lem}

\begin{proof}
See Appendix \ref{proof_lem_binary_sensing}, where the big-O constant for $P$ is analyzed. 
\end{proof}

\section{Noisy Recovery in the Continuous Alphabet Setting}\label{sec:noisy_continuous}
In this section, we provide details of the noisy recovery algorithm in the continuous alphabet setting. The major challenge with continuous alphabet is that, since it is impossible to obtain the exact values of the sparse coefficients in the presence of noise, the iterative decoding procedure may suffer from error propagation if we do not design and analyze the algorithm carefully. The key idea of our algorithm in the continuous alphabet setting is to use a truncated peeling algorithm so that the error propagation can be controlled. In the following, we first present the construction of the bin detection matrix, and then the modified peeling decoding algorithm.

\subsection{Bin Detection Matrix}\label{sec:bin_detection_matrix}
Similar to the quantized alphabet setting, we still use the regular graph ensemble $\mathcal{G}_{\rm reg}^N(R, {d})$ for constructing the coding matrix $\mathbf{H}$. Meanwhile, the design of the bin detection matrix $\mathbf{S}\in\{-1,1\}^{P\times N}$ is slightly modified in order to better fit the continuous alphabet setting. The matrix $\mat{S}$ consists of two parts, the \emph{location} matrix $\mat{S}_0\in\{-1,1\}^{P_0\times N}$ and the \emph{verification} matrix $\mat{S}_1\in\{-1,1\}^{P_1\times N}$, i.e., $\mat{S} = [\mat{S}_0^T,~\mat{S}_1^T]^T$, and thus, the number of measurements in each bin detection matrix is $P=P_0+P_1$. We denote by $\vect{s}_j$, $\vect{s}_{0,j}$, and $\vect{s}_{1,j}$ the $j$-th column $(j\in[N])$ of $\mat{S}$, $\mat{S}_0$, and $\mat{S}_1$, respectively. Similar to the quantized alphabet setting, we have the following generative model on the measurements in a particular bin (the bin index is omitted):
\begin{align}\label{equiv.model.bin_continuous}
	\bdsb{y} = \mathbf{S}\mathbf{z} + \mathbf{w}.
\end{align}
With the design of $\mat{S}$, the measurement $\bdsb{y}$ consist of two parts, i.e., $\bdsb{y} = [\vect{u}_0^T,\vect{u}_1^T]^T$, where $\vect{u}_i = \mathbf{S}_i \mathbf{z} + \mathbf{w}_i$, $i=0,1$.

Again, the bin detection matrix $\mat{S}$ is used to check whether a bin is a single-ton bin, and if it is, the bin detection matrix $\mat{S}$ finds the index-value pair of the sparse coefficient. Suppose that a particular bin is a single-ton and the sparse coefficient is located at $j$, $j\in[N]$, i.e., $\vect{z}=x[j]\vect{e}_j$, where $\vect{e}_j$ is the $j$-th vector of the standard basis. Then, the measurements of this bin is $\bdsb{y}=x[j]\vect{s}_j+\vect{w}$. As mentioned above, we can divide the measurements into two parts, location measurements $\vect{u}_0$ and verification measurements $\vect{u}_1$, which correspond to the location matrix and verification matrix, respectively. Namely, we have $\vect{u}_0=x[j]\mat{s}_{0,j}+\vect{w}_0$ and $\vect{u}_1=x[j]\mat{s}_{1,j}+\vect{w}_1$. 
%The location matrix is used to find the location index $j$, and the verification matrix is used to check whether this bin is indeed a single-ton and estimate the value of $x[j]$.

The design of the verification matrix is relatively simple. The entries of the verification matrix $\mat{S}_1$ are i.i.d. Rademacher distributed, i.e., all the entries are independent and equally likely to be either $1$ or $-1$. The design of the location matrix $\mat{S}_0$ is more complicated. As we can see, if a bin is indeed a single-ton, then the location measurements $\vect{u}_0$ is a scaled version of $\mat{s}_{0,j}$ with additive Gaussian noise $\vect{w}_0$. Let $\zeta=\Phi(-| x[j] |/\sigma)$, where $\Phi(\cdot)$ is the CDF of standard Gaussian distribution. Taking the sign\footnote{In this section, we use the standard definition of sign, i.e., 
$$
\sgn{x} = \begin{cases}
1, & x \ge 0\\
-1, & x < 0.
\end{cases}
$$
} of all the location measurements and considering the randomness of the Gaussian noise, we can see that for each element $u_{0,k}$ in the location measurements, $k\in[P_0]$, we have
$$
\sgn{u_{0,k}}=
\begin{cases}
\sgn{x[j]}s_{0,k,j}\quad &\text{with probability } 1-\zeta \\
-\sgn{x[j]}s_{0,k,j}\quad &\text{with probability } \zeta.
\end{cases}
$$
Now the problem becomes a channel coding problem in a symmetric channel with symbols $\{+1,-1\}$. The channel is similar to the binary symmetric channel (BSC) except the fact that we are using $\{+1,-1\}$ rather than $\{0,1\}$. For simplicity we will still call this channel a BSC in the following context. Consider the $N$ possible locations of the sparse coefficient as $N$ \emph{messages}. We encode the $N$ messages by $P_0$-bit \emph{codewords} with symbols $\pm 1$, or equivalently, we design a map $f:[N]\rightarrow \{1,-1\}^{P_0}$, and the columns of the location matrix are the codewords of all the messages, i.e., $\vect{s}_{0,j}=f(j)$, $j\in [N]$. If $x[j]<0$, the codeword gets a global sign flip and then we get the modified codeword $\sgn{x[j]}\vect{s}_{0,j}$. Transmitting this modified codeword through a BSC with bit flip probability $\zeta$, we get the received sequence, $\sgn{\vect{u}_0}$. Then we need a decoding algorithm to decode the original codeword $\vect{s}_{0,j}$, up to a global sign flip, and then, there are at most two possible locations of the sparse coefficient. Then, one can use the verification measurements to check whether the bin is indeed a single-ton, find the correct location among the two possible choices, and estimate the value of the sparse coefficient.

Now we describe the encoding and decoding scheme of the location matrix. The code should satisfy four properties:
\begin{itemize}
\item[(i)] The block length of the codewords should be as small as possible. Since we need at least $O(\log(N))$ bits to encode $N$ messges, $P_0$ should be as close to $O(\log(N))$ as possible.
\item[(ii)] The decoding complexity should be as close to $O(\log(N))$ as possible.
\item[(iii)] The decoding algorithm succeeds with high probability; specifically, when there are $O(1)$ bits flipped, we need the probability of successful decoding to be $1-O(1/\poly(N))$.
\item[(iv)] The decoding algorithm should be \emph{universal}, i.e., it should not rely on the exact knowledge of the bit flipping probability.
\end{itemize}

Many of the state-of-the-art capacity achieving codes, such as LDPC codes and Polar codes, satisfy the first two properties. However, in order to have $1-O(1/\poly(N))$ error probability, the decoding algorithms in these codes need exact knowledge of the channel, meaning that these algorithms need the flip probability $\zeta$ as a known input parameter. However, in our problem, $\zeta=\Phi(-\ABSL{x[j]}/\sigma)$, where $\ABSL{x[j]}$ is unknown. This is the reason that we need universal decoding algorithm. In practice, since we have an upper bound of the bit flip probability, $\zeta \le \Phi(-\beta/\sigma)$, it is reasonable to believe that if we use the upper bound as the bit flip probability, the state-of-the-art capacity achieving codes still work well, although there is no theoretical guarantee. For theoretical interests, here we propose a \emph{concatenated} code which satisfies all the four properties provably. The results are given in Lemma \ref{lem:code}. This code is based on Justesen's concatenation scheme \cite{justesen1972class}, linear complexity expander codes \cite{spielman1995linear}, and the Wozencraft's ensemble \cite{massey1963threshold}.

\begin{lem}\label{lem:code}
There exists a concatenated code 
$$
f_c:[N]\rightarrow \{1,-1\}^{P_0}
$$
for BSC with block length $P_0=O(\log(N)\log\log(N))$ and universal decoding algorithm, which can successfully decode with probability $1-O(1/\poly(N))$. The decoding complexity is $O(\log^{1+r}(N))$, where $r>0$ is an arbitrarily small constant.
\end{lem}
\begin{proof}
See Appendix~\ref{sec:prf_lemma_code}.
\end{proof}
With this concatenated code, we can construct the location matrix $\mat{S}_0$ by setting the $j$-th column as the codeword of $j$, i.e., $\vect{s}_{0,j}=f_c(j)$. Meanwhile, we note that this concatenated code is designed mainly for theoretical purpose. In practice, we can use LDPC codes and Polar codes in the location matrix, and in the decoding algorithm use $ \Phi(-\beta/\sigma)$ as an estimate of the bit flip probability of the BSC channel. In fact, if we make the conjecture that there exists a code with block length $P_0=O(\log(N)$ and has uniform decoding algorithm, linear decoding complexity, and success probability $1-O(1/\poly(N))$, then we can remove the $\log\log(N)$ factor in the measurement cost, and reduce the $\log^{1+r}(N)$ factor in the run-time to $\log(N)$.

\subsection{Peeling Decoder with Truncation}\label{sec:peel}

Recall that the basic idea of the peeling decoder is to use the location matrix and verification matrix to identify single-ton bins, and estimate the index-value pairs of the sparse coefficients in the single-ton bins. After identifying a single-ton bin, the decoder peels the sparse coefficient (left node) from its neighborhood measurement bins (right nodes). Then, more bins become single-tons. The decoder continues the peeling process iteratively until no single-ton bin can be found. The major challenge in the continuous alphabet setting is that, the signal components are real-valued, and thus we cannot obtain the exact values of the sparse coefficients. Therefore, error propagation in the peeling process is inevitable. We propose a truncation peeling strategy in order to control the error propagation.

Here, we demonstrate the peeling algorithm with truncation strategy via a simple example in Figure \ref{fig:example2}. The main idea is to fix the maximum number of sparse coefficients that can be peeled from a measurement bin. Denote this maximum number by $D$, which is an input constant parameter of the algorithm. This means that when at least $D$ sparse coefficients have been peeled from a particular bin, we stop using this bin in following iterations, i.e., we ``truncate'' large multi-ton bins that are connected to more than $D$ sparse coefficients. We set $D=2$ in the example in Figure \ref{fig:example2}. 

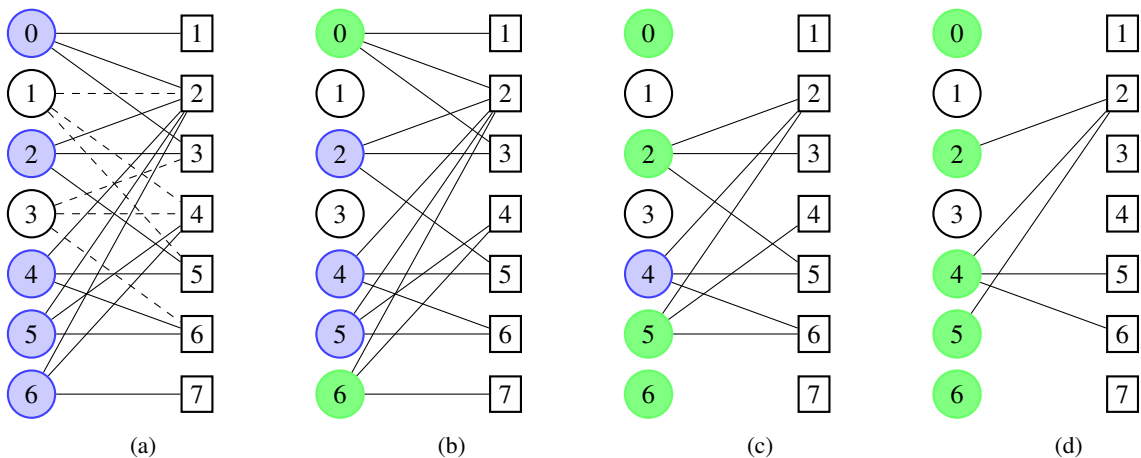
\begin{figure*}[t]
\centering
\begin{subfigure}[b]{0.2\textwidth}
\begin{tikzpicture}[node distance=0.7cm,>=stealth',bend angle=45,auto]
\tikzstyle{ball}=[circle,thick,draw=blue!75,fill=blue!20,minimum size=2.5mm]
\tikzstyle{greenball}=[circle,thick,draw=green!60,fill=green!50,minimum size=2.5mm]
\tikzstyle{emptyball}=[circle,thick,draw=black,fill=white,minimum size=2.5mm]
\tikzstyle{bin}=[rectangle,thick,draw=black,fill=white,minimum size=3mm]
\begin{scope}
\node [ball] (L1) {0};
\node [emptyball] (L2) [below of=L1, yshift=-0.10cm] {1};
\node [ball] (L3)  [below of=L2, yshift=-0.10cm] {2};
\node [emptyball] (L4) [below of=L3, yshift=-0.10cm] {3};
\node [ball] (L5) [below of=L4, yshift=-0.10cm] {4};
\node [ball] (L6) [below of=L5, yshift=-0.10cm] {5};
\node [ball] (L7) [below of=L6, yshift=-0.10cm] {6};	

\node [bin] (R1) [right of=L1, xshift=1.5cm] {1};
\node [bin] (R2) [right of=L2, xshift=1.5cm] {2};
\node [bin] (R3) [right of=L3, xshift=1.5cm] {3};
\node [bin] (R4) [right of=L4, xshift=1.5cm] {4};
\node [bin] (R5) [right of=L5, xshift=1.5cm] {5};
\node [bin] (R6) [right of=L6, xshift=1.5cm] {6};
\node [bin] (R7) [right of=L7, xshift=1.5cm] {7};

\path (L1) edge [left] (R1);
\path (L1) edge [left] (R2);
\path (L1) edge [left] (R3);
\path (L2) edge [left, dashed] (R2);
\path (L2) edge [left, dashed] (R4);
\path (L2) edge [left, dashed] (R5);
\path (L3) edge [left] (R2);
\path (L3) edge [left] (R3);
\path (L3) edge [left] (R5);
\path (L4) edge [left, dashed] (R3);
\path (L4) edge [left, dashed] (R4);
\path (L4) edge [left, dashed] (R6);
\path (L5) edge [left] (R2);
\path (L5) edge [left] (R5);
\path (L5) edge [left] (R6);
\path (L6) edge [left] (R2);
\path (L6) edge [left] (R4);
\path (L6) edge [left] (R6);
\path (L7) edge [left] (R2);
\path (L7) edge [left] (R4);
\path (L7) edge [left] (R7);
\end{scope}
\end{tikzpicture}
\label{subfig:bipartite}\caption{          }
\end{subfigure}
\quad
\begin{subfigure}[b]{0.2\textwidth}
\begin{tikzpicture}[node distance=0.7cm,>=stealth',bend angle=45,auto]
\tikzstyle{ball}=[circle,thick,draw=blue!75,fill=blue!20,minimum size=2.5mm]
\tikzstyle{greenball}=[circle,thick,draw=green!60,fill=green!50,minimum size=2.5mm]
\tikzstyle{emptyball}=[circle,thick,draw=black,fill=white,minimum size=2.5mm]
\tikzstyle{bin}=[rectangle,thick,draw=black,fill=white,minimum size=3mm]
\begin{scope}
\node [greenball] (L1) {0};
\node [emptyball] (L2) [below of=L1, yshift=-0.10cm] {1};
\node [ball] (L3)  [below of=L2, yshift=-0.10cm] {2};
\node [emptyball] (L4) [below of=L3, yshift=-0.10cm] {3};
\node [ball] (L5) [below of=L4, yshift=-0.10cm] {4};
\node [ball] (L6) [below of=L5, yshift=-0.10cm] {5};
\node [greenball] (L7) [below of=L6, yshift=-0.10cm] {6};	

\node [bin] (R1) [right of=L1, xshift=1.5cm] {1};
\node [bin] (R2) [right of=L2, xshift=1.5cm] {2};
\node [bin] (R3) [right of=L3, xshift=1.5cm] {3};
\node [bin] (R4) [right of=L4, xshift=1.5cm] {4};
\node [bin] (R5) [right of=L5, xshift=1.5cm] {5};
\node [bin] (R6) [right of=L6, xshift=1.5cm] {6};
\node [bin] (R7) [right of=L7, xshift=1.5cm] {7};

\path (L1) edge [left] (R1);
\path (L1) edge [left] (R2);
\path (L1) edge [left] (R3);
	
%\path (L2) edge [left, dashed] (R2);
%\path (L2) edge [left, dashed] (R4);
%\path (L2) edge [left, dashed] (R5);
	
\path (L3) edge [left] (R2);
\path (L3) edge [left] (R3);
\path (L3) edge [left] (R5);
	
%\path (L4) edge [left, dashed] (R3);
%\path (L4) edge [left, dashed] (R4);
%\path (L4) edge [left, dashed] (R6);
	
\path (L5) edge [left] (R2);
\path (L5) edge [left] (R5);
\path (L5) edge [left] (R6);
	
\path (L6) edge [left] (R2);
\path (L6) edge [left] (R4);
\path (L6) edge [left] (R6);
	
\path (L7) edge [left] (R2);
\path (L7) edge [left] (R4);
\path (L7) edge [left] (R7);
\end{scope}
\end{tikzpicture}
\label{subfig:peeling1}\caption{          }
\end{subfigure}
\quad
\begin{subfigure}[b]{0.2\textwidth}
\begin{tikzpicture}[node distance=0.7cm,>=stealth',bend angle=45,auto]
\tikzstyle{ball}=[circle,thick,draw=blue!75,fill=blue!20,minimum size=2.5mm]
\tikzstyle{greenball}=[circle,thick,draw=green!60,fill=green!50,minimum size=2.5mm]
\tikzstyle{emptyball}=[circle,thick,draw=black,fill=white,minimum size=2.5mm]
\tikzstyle{bin}=[rectangle,thick,draw=black,fill=white,minimum size=3mm]
\begin{scope}
\node [greenball] (L1) {0};
\node [emptyball] (L2) [below of=L1, yshift=-0.10cm] {1};
\node [greenball] (L3)  [below of=L2, yshift=-0.10cm] {2};
\node [emptyball] (L4) [below of=L3, yshift=-0.10cm] {3};
\node [ball] (L5) [below of=L4, yshift=-0.10cm] {4};
\node [greenball] (L6) [below of=L5, yshift=-0.10cm] {5};
\node [greenball] (L7) [below of=L6, yshift=-0.10cm] {6};	

\node [bin] (R1) [right of=L1, xshift=1.5cm] {1};
\node [bin] (R2) [right of=L2, xshift=1.5cm] {2};
\node [bin] (R3) [right of=L3, xshift=1.5cm] {3};
\node [bin] (R4) [right of=L4, xshift=1.5cm] {4};
\node [bin] (R5) [right of=L5, xshift=1.5cm] {5};
\node [bin] (R6) [right of=L6, xshift=1.5cm] {6};
\node [bin] (R7) [right of=L7, xshift=1.5cm] {7};
	
\path (L3) edge [left] (R2);
\path (L3) edge [left] (R3);
\path (L3) edge [left] (R5);	
\path (L5) edge [left] (R2);
\path (L5) edge [left] (R5);
\path (L5) edge [left] (R6);	
\path (L6) edge [left] (R2);
\path (L6) edge [left] (R4);
\path (L6) edge [left] (R6);
\end{scope}
\end{tikzpicture}
\label{subfig:peeling2}\caption{          }
\end{subfigure}
\quad
\begin{subfigure}[b]{0.2\textwidth}
\begin{tikzpicture}[node distance=0.7cm,>=stealth',bend angle=45,auto]
\tikzstyle{ball}=[circle,thick,draw=blue!75,fill=blue!20,minimum size=2.5mm]
\tikzstyle{greenball}=[circle,thick,draw=green!60,fill=green!50,minimum size=2.5mm]
\tikzstyle{emptyball}=[circle,thick,draw=black,fill=white,minimum size=2.5mm]
\tikzstyle{bin}=[rectangle,thick,draw=black,fill=white,minimum size=3mm]
\begin{scope}
\node [greenball] (L1) {0};
\node [emptyball] (L2) [below of=L1, yshift=-0.10cm] {1};
\node [greenball] (L3)  [below of=L2, yshift=-0.10cm] {2};
\node [emptyball] (L4) [below of=L3, yshift=-0.10cm] {3};
\node [greenball] (L5) [below of=L4, yshift=-0.10cm] {4};
\node [greenball] (L6) [below of=L5, yshift=-0.10cm] {5};
\node [greenball] (L7) [below of=L6, yshift=-0.10cm] {6};	

\node [bin] (R1) [right of=L1, xshift=1.5cm] {1};
\node [bin] (R2) [right of=L2, xshift=1.5cm] {2};
\node [bin] (R3) [right of=L3, xshift=1.5cm] {3};
\node [bin] (R4) [right of=L4, xshift=1.5cm] {4};
\node [bin] (R5) [right of=L5, xshift=1.5cm] {5};
\node [bin] (R6) [right of=L6, xshift=1.5cm] {6};
\node [bin] (R7) [right of=L7, xshift=1.5cm] {7};
	
\path (L5) edge [left] (R2);
\path (L5) edge [left] (R5);
\path (L5) edge [left] (R6);	
\path (L3) edge [left] (R2);
\path (L6) edge [left] (R2);
\end{scope}
\end{tikzpicture}
\caption{          }
\end{subfigure}\label{subfig:peeling3}
\caption{Peeling with truncation. The signal length is $7$ and we design $7$ measurement bins. In the bipartite graph, the left nodes and the right nodes correspond to the sparse coefficients and measurement bins, respectively. The sparse coefficients are shown with color (if a sparse coefficient is recovered, the left node is shown in blue, otherwise it is shown in green). (a) The bipartite graph. The support of the signal is $\{0, 2, 4, 5, 6\}$. The bipartite graph is $3$-left regular, and the connections between zero elements and the measurement bins are shown in dashed lines. (b) Bin $1$ and bin $7$ are single-ton bins, and the corresponding signal components $x[0]$ and $x[6]$ are recovered. (c) Peel $x[0]$ and $x[6]$ from the measurement bins. Since two sparse coefficients are peeled from bin $2$, in the following iterations, we stop using bin $2$. Bin $3$ and bin $4$ become single-ton bins, and the corresponding sparse coefficients are $x[2]$ and $x[5]$. (d) Peel $x[2]$ and $x[5]$ from the measurement bins, and bin $5$ and bin $6$ become single-ton bins. Then, $x[4]$ is recovered.}
\label{fig:example2}
\end{figure*}

We first assume that by the location measurements and verification measurements, we can perfectly identify whether a bin is a single-ton and find the exact location of the sparse coefficient. As we can see, in Figure \ref{fig:example2}, the bins $1$ and $7$ are single-ton bins and the corresponding sparse coefficients are $x[0]$ and $x[6]$, respectively. In the first iteration, the two sparse coefficients are found and we let $\widehat{x}[0]$ and $\widehat{x}[6]$ be the estimated values. Then, we do peeling, meaning that we subtract the measurements contributed by the two sparse coefficients from the measurements in other bins. We get the \emph{remaining} measurements of bins $2$, $3$, $4$, $5$, and $6$ after the first iteration:
\begin{align*}
\vect{y}_2^{(1)}&=\vect{y}_2-\widehat{x}[0]\vect{s}_0 - \widehat{x}[6]\vect{s}_6 \\
\vect{y}_3^{(1)}&=\vect{y}_3- \widehat{x}[0]\vect{s}_0 \\
\vect{y}_4^{(1)}&=\vect{y}_4- \widehat{x}[6]\vect{s}_6 \\
\vect{y}_5^{(1)}&=\vect{y}_5\\
\vect{y}_6^{(1)}&=\vect{y}_6.
\end{align*}
Here, we use $\vect{y}_i$ to denote the measurement in the $i$-th bin, and $\vect{y}_i^{(t)}$ to denote the remaining measurement in the $i$-th bin after the $t$-th iteration. Then we can see that bins $3$ and $4$ become single-ton bins, and the corresponding sparse coefficients are $x[2]$ and $x[5]$, respectively. We should also notice that since two sparse coefficients have been peeled from bin $2$, according to the truncated peeling strategy, we should stop using bin $2$ in the following iterations. Let $\widehat{x}[2]$ and $\widehat{x}[5]$ be the estimated values of the sparse coefficients. Then, the remaining measurements of bins $5$ and $6$ after the second iteration are:
\begin{align*}
\vect{y}_5^{(2)}&=\vect{y}_5^{(1)}-\widehat{x}[2]\vect{s}_2\\
\vect{y}_6^{(2)}&=\vect{y}_6^{(1)}-\widehat{x}[5]\vect{s}_5.
\end{align*}
Then, bins $5$ and $6$ become single-ton bins and the corresponding sparse coefficient is $x[4]$. We can estimate the value of $x[4]$ and get $\widehat{x}[4]$. So far, all the balls have been found, meaning that the all the sparse coefficients are found. We summarize the detailed procedure of peeling decoding algorithm with truncation strategy in Algorithm~\ref{alg:peeling}.

\begin{algorithm}[!h]
\caption{Peeling decoding with truncation strategy}\label{alg:peeling}
\begin{algorithmic} 
\STATE ${\tt Input:}$ Observation $\vect{y}_i$, $i\in[R]$, bin detection matrix $\mat{S}$, coding matrix $\mat{H}$, and truncation threshold $D$
\STATE ${\tt Output:}$ Estimated signal $\widehat{\vect{x}}$
\STATE $\widehat{\vect{x}}\leftarrow\vect{0}$, 
\STATE number of peeled sparse coefficients in each bin: $B_i\leftarrow 0$, $i\in[R]$, 
\STATE Indicator of utilizability of bins: $U_i\leftarrow\TRUE$, $i\in[R]$,
\STATE $\vect{y}_i^{(0)} \leftarrow \vect{y}_i$, $i\in[R]$, $\text{stop}\leftarrow\FALSE$, $t \leftarrow 1$
\WHILE {$\text{stop}=\FALSE$}
\STATE Find sparse coefficients in single-ton bins. 
\STATE $\mathcal{I}_t\leftarrow\{\text{indices of all single-ton bins found in the iteration $t$}\}$.
\STATE $U_i\leftarrow\FALSE$, for all $i\in\mathcal{I}_t$.
\STATE $\mathcal{J}_t\leftarrow\{\text{locations of sparse coefficient in single-tons found in iteration $t$}\}$. 
\STATE $\vect{y}_i^{(t)}\leftarrow\vect{y}_i^{(t-1)}$, $i\in[R]$.
%\STATE $\hat{s}_j\leftarrow\text{estimated value}$ for all $j\in\mathcal{J}_t$.
\IF{$\mathcal{J}_t\neq\emptyset$}
\FORALL {$j\in\mathcal{J}_t$}
\STATE Estimate $\widehat{x}[j]$.
\FORALL {$i\in[R]$ such that $U_i=\TRUE$ and $h_{i,j}=1$}
\STATE $\vect{y}_i^{(t)}\leftarrow \vect{y}_i^{(t)}-\widehat{x}[j]\vect{s}_j$.
\STATE $B_i\leftarrow B_i+1$.
\IF {$B_i=D$}
\STATE $U_i\leftarrow\FALSE$
\ENDIF
\ENDFOR
\ENDFOR
\ELSE
\STATE $\text{stop}\leftarrow\TRUE$
\ENDIF
\STATE $t\leftarrow t+1$
\ENDWHILE
\RETURN $\widehat{\vect{x}}$
\end{algorithmic}
\end{algorithm}

The following result of the peeling procedure guarantees that when the peeling process stops, an arbitrarily large fraction of sparse coefficients are found. Similar to the results in the noiseless setting and quantized alphabet setting, the proof of Lemma~\ref{lem:peeling} is based on density evolution, and the only difference is in the truncation strategy.
\begin{lem}\label{lem:peeling} 
Assume that we can always find the correct location of the sparse coefficients in single-ton bins. For any $p>0$, when $K$ is large enough, there exist proper parameters $d=O(1)$ and $R=O(\log(1/p)K)$, such that using a random left regular graph $\mathcal{G}_{\rm reg}^N(R,{d})$, after $n_p$ iterations of truncated peeling, with probability $1-O(\exp\{-c_1(p)K^{c_2(p)}\})$, the fraction of non-zero signal elements that are not detected is less than $p$. Here, $c_1(p),c_2(p)>0$ are two quantities determined by $p$.
\end{lem} 
\begin{proof}
See Appendix~\ref{prf:peeling}.
\end{proof}

\subsection{Single-ton Detection and Signal Estimation}\label{sec:detect}

In Section \ref{sec:peel}, we have shown that if the single-ton bins are always perfectly detected, and the exact location of the sparse coefficients can always be found, an arbitrarily large fraction of non-zero signal elements can be recovered. Then, the remaining issue is to guarantee correct single-ton detection and accurate value estimation. 

Recall that in the first iteration, if a bin is indeed a single-ton bin, from the location measurements, one can decode the modified codeword corresponding to the location index of the sparse coefficient. Due to the sign ambiguity, there may be two possible locations and the true location is guaranteed to be one of them with high probability. We still need to find the correct location and estimate the values of the sparse coefficient. On the other hand, if the bin is not a single-ton, the decoding algorithm of the concatenated code still returns at most two possible locations and we have to make sure that these bins are not considered as single-ton bins. These problems are addressed by energy tests using the verification measurements, based on the same idea as in \cite{yin2015fast}.

\subsubsection{Signal Value Estimation}
Consider a particular bin at a particular iteration. For simplicity, in this part, we resume the notation in Section~\ref{sec:bin_detection_matrix}; more specifically, we omit the index of the bin and the iteration counter, and use $\vect{u}_1$ to denote the \emph{remaining} verification measurement at a particular iteration (this means that the contribution of the recovered sparse coefficients are already subtracted). Let $j$ be a possible location of the sparse coefficient that the decoding algorithm of the concatenated code suggests. We assume that the bin is indeed a single-ton with the single-ton ball located at $j$, and estimate $x[j]$ by the remaining verification measurements, i.e.,
\begin{equation}\label{eq:estimate}
\widehat{x}[j]=\frac{1}{P_1}\sum_{k=1}^{P_1} s_{1,k,j} u_{1,k}.
\end{equation}
Here, $s_{1,k,j}$ is the element at the $k$-th row and the $j$-th column of the verification matrix $\mat{S}_1$, and $u_{1,k}$ is the $k$-th element in $\vect{u}_1$. Intuitively, this estimation method is simply averaging over the measurements with corrected sign, meaning that we flip the sign if the corresponding entry in the verification matrix is $-1$. The theoretical guarantee of single-ton detection and estimation is presented in Lemma \ref{lem:estimation}.

\begin{lem}\label{lem:estimation}
For any $\epsilon>0$, with $P_1=O(\frac{\sigma^2}{\epsilon^2}\log(N))$ verification measurements in each bin, when $\beta> c\epsilon$ for some constant $c>0$, we can accurately detect any single-ton bin within a constant number of iterations. More specifically, we have:
\begin{itemize}
\item[(i)] the location measurements can find the correct location of the sparse coefficient in the single-ton bin with probability $1-O(1/\poly(N))$,
\item[(ii)] the estimated value of sparse coefficient $\widehat{x}[j]$ satisfies $| \widehat{x}[j] - x[j] | \le C_j\epsilon$ for some constant $C_j>0$ with probability $1-O( 1/\poly(N) )$. 
\end{itemize}
\end{lem}

\begin{proof}
See Appendix~\ref{prf:estimate}.
\end{proof}

We note that result (i) is a simple extension of the conclusion that we get in Section \ref{sec:bin_detection_matrix}, where we focused on the first iteration, and result (ii) shows that for any target accuracy level $\epsilon>0$, if the number of verification measurements is $P_1=O(\frac{\sigma^2}{\epsilon^2}\log(N))$, we can estimate the signal value within constant factor of $\epsilon$ with high probability. 

\subsubsection{Energy Test}

So far, we have seen that if a bin is indeed a single-ton, the location measurements can find the correct location of the sparse coefficient and the verification measurements can give accurate estimation of the value. However, there are still several things left. As we have mentioned, we need to clarify sign ambiguity, and rule out measurement bins that are not single-tons. These operations can be done by energy tests. 

Consider the $i$-th bin in the $t$-th iteration. Let $\vect{u}_1$ be the remaining verification measurements, and $\mathcal{B}$ be the set of location indices of sparse coefficients in the $i$-th bin that have been found before this iteration. Before using the location measurements to find the location of new sparse coefficients, we use an energy test to check if this bin is a zero-ton bin, i.e., check if $\supp{\vect{z}} = \mathcal{B}$. If it is, there is no need to run the decoding algorithm of the concatenated code. More specifically, we construct $\widehat{\vect{u}}_1=\sum_{g\in\mathcal{B}}\widehat{x}[g] \vect{s}_{1,g}$ and conduct the zero-ton energy test with threshold $\tau>0$:
\begin{align*}
&\text{if } \frac{1}{P_1}\| \vect{u}_1 - \widehat{\vect{u}}_1 \|_2^2<\tau, \text{ bin $i$ is a zero-ton bin;}\\
&\text{else bin $i$ is not a zero-ton bin.}
\end{align*}

If the bin is not a zero-ton bin, we use the location measurements to find a possible single-ton location $j$ and get the estimated the value $\widehat{x}[j]$. We need to verify if there is indeed $\supp{\vect{z}}=\mathcal{B}\cup\{j\}$. Similar to the zero-ton test, we construct $\widehat{\vect{u}}_1=\sum_{g\in\mathcal{B}\cup\{j\}}\widehat{x}[g]\vect{s}_{1,g}$ and conduct the single-ton energy test with threshold $\tau>0$:
\begin{align*}
&\text{if } \frac{1}{P_1}\| \vect{u}_1 - \widehat{\vect{u}}_1 \|_2^2<\tau, \\
&\text{\quad bin $i$ is a single-ton bin with sparse coefficient located at $j$;}\\
&\text{else} \\
&\text{\quad bin $i$ is not a single-ton bin with sparse coefficient located at $j$.}
\end{align*}

The intuition behind both energy tests is simple. We actually make a hypothesis that the true signal of a bin is $\widehat{\vect{z}}$ and construct the corresponding verification measurements $\widehat{\vect{u}}_1=\mat{S}_1\widehat{\vect{z}}$. If the support of $\widehat{\vect{z}}$ and $\vect{z}$ are the same and the values are accurately estimated, i.e., $\|\widehat{\vect{z}} - \vect{z} \|_\infty<C_0\epsilon$, for some constant $C_0>0$, then the energy of the difference between the actual measurements and the constructed measurements should be small; otherwise, the energy should be large. The theoretical guarantees of both energy tests are provided in Lemma~\ref{lem:energy}.

\begin{lem}\label{lem:energy}
When $\beta = \Omega((\sigma+\epsilon)^2)$, there exists a proper threshold $\tau>0$ such that any energy test succeeds with probability $1-O(1/\poly(N))$, when $P_1=O(\max\{\sigma^2/\epsilon^2, 1\}\log(N))$. 
\end{lem}
\begin{proof}
See Appendix~\ref{prf:energy}.
\end{proof}
With all these ingredients above, we are now ready to prove Theorem~\ref{thm_continuous_recovery}, which is our main result in the continuous alphabet setting with noise. The proof is a simple application of the total law of probability, similar to the ideas that we use in Appendix~\ref{sec:noisy_recovery_perf_analysis}. We relegate the brief proof to Appendix~\ref{prf:continuous_recovery}. We also mention that, since we use random left regular bipartite graph, the recovered $1-p$ fraction of the support is uniformly distributed over the full support of the unknown signal. Therefore, by running the algorithm $\log(K)$ times independently, each sparse coefficient can be recovered with high probability, and thus we can get the full support recovery guarantee.

\section{Numerical Experiments}\label{sec:numerical}
In this section, we provide the empirical performance of our design in the noiseless and noisy settings. Each data point in the simulation is generated by averaging over 200 experiments, where the signals $\mathbf{x}$ are generated once and kept fixed for all the subsequent experiments. In particular, the support of $\mathbf{x}$ are generated uniformly random from $[N]$. In the presence of noise, the signal-to-noise ratio (SNR) is defined as
\begin{align}
	\mathsf{SNR} 
	= \frac{\mathbb{E}\left[\left\|\mathbf{A}\mathbf{x}\right\|^2\right]}{\mathbb{E}\left[\left\|\mathbf{w}\right\|^2\right]}
	=\frac{\left\|\mathbf{x}\right\|^2}{\sigma^2}\frac{\bar{d}}{R}
\end{align}
where $\bar{d}$ is the average left node degree of the bipartite graph, $R$ is the number of right nodes in the graph, and the expectation is taken with respect to the noise, {\it random bipartite graph} and {\it bin detection matrix}. Then in noisy settings, we generate i.i.d. Gaussian noise with variance $\sigma^2$ according to the specified SNR. 

\subsection{Scalability of Measurement and Computational Costs for Noiseless Recovery}
In this case, we examine the measurement cost and run-time of our noiseless recovery algorithm. The measurement matrix $\mathbf{A}$ is constructed using the coding matrix $\mathbf{H}$ from the irregular graph ensemble $\mathcal{G}_{\rm irreg}^N(R, D)$ by fixing $R=1.1K$ and $D=100$. We show experiments with different sparsities where $K=200$, $400$ and $600$ and for each of the sparsity settings, we simulate our noiseless recovery algorithm for recovering sparse signals of dimension $N = 10^4$ to $N=7\times 10^4$. It can be seen in Fig. \ref{fig:noiseless_cost} that the measurement and computational costs remain constant irrespective of the growth in $N$.

\begin{figure}[h]
\begin{center}
\includegraphics[width=0.4\linewidth]{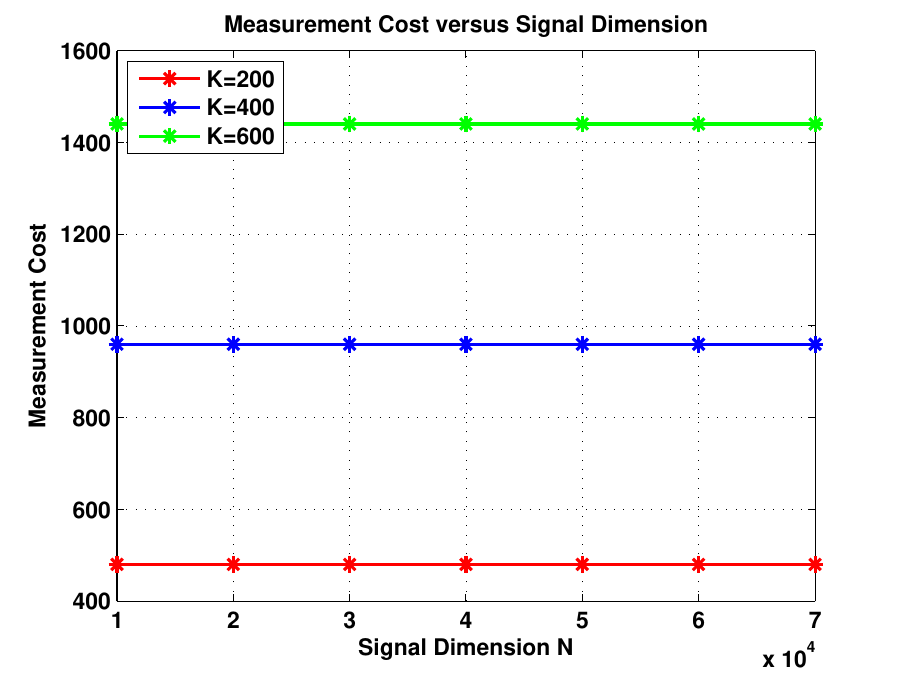}
\includegraphics[width=0.4\linewidth]{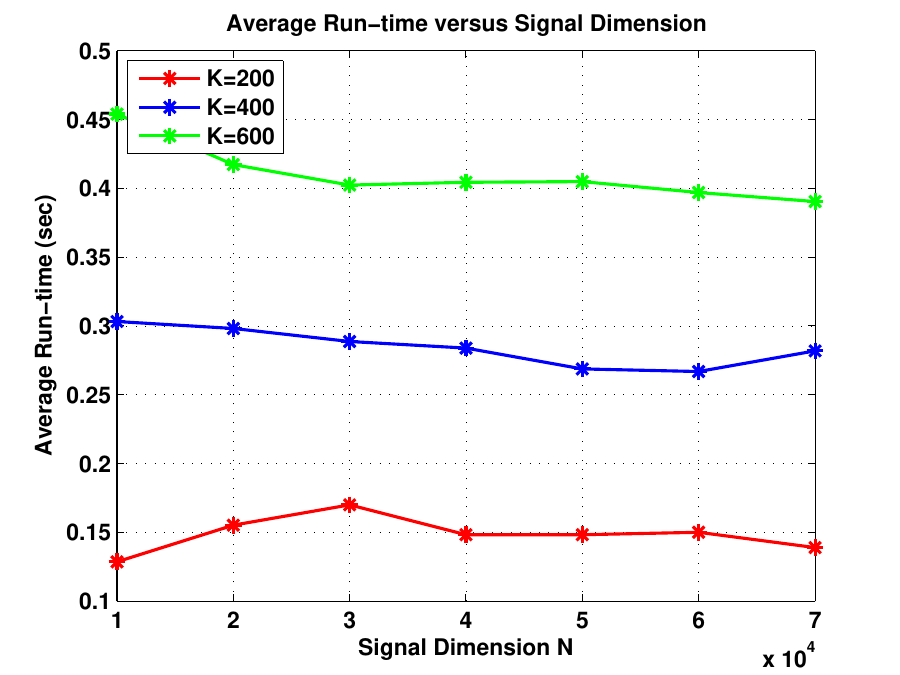}
\caption{Measurement and computational costs as functions of the signal dimension $N$ for noiseless recovery.  It can be seen that the measurement and computational costs remain constant irrespective of the growth in $N$.}
\label{fig:noiseless_cost}
\end{center}
\end{figure}	

\subsection{Noise Robustness and Scalabilty in the Quantized Alphabet Setting}\label{sec:noise_robustness}

In this subsection, we showcase the robustness and scalability of the noisy design in the quantized alphabet setting. The sparse coefficient are chosen from $\{-1,1\}$ uniformly at random. The measurement matrix $\mathbf{A}$ is constructed as follows:
\begin{itemize}
	\item the coding matrix $\mathbf{H}$ is constructed using the regular graph ensemble $\mathcal{G}_{\rm reg}^N(R, {d})$ with a regular degree ${d}=3$ and a redundancy $R=2K$;
	\item we choose a $P_1\times N$ random Rademacher matrix for the zero-ton and single-ton verifications with $P_1=\log N$, and a $P_2\times N$ coded binary matrix for the single-ton search with $P_2=2\log_2 N$. In particular, the coded binary matrix $\mathbf{C}=\mathbf{G}\mathbf{B}$ is chosen based on the $P_2\times \log_2N$ generator matrix $\mathbf{G}$ associated with a $(3,6)$-regular LDPC code, and the single-ton search utilizes the Gallager's bit flipping algorithm for decoding. 
\end{itemize}	

To demonstrate the noise robustness, the probability of success is plotted against a range of SNR from $0$dB to $16$dB for both designs. In each experiment, $50$-sparse signals $\mathbf{x}$ (i.e., $K=50$) with $N=10^5$ are generated. It is seen in Fig. \ref{fig:noise_robustness} that for a given measurement cost, there exists a threshold of SNR, above which our noisy recovery schemes succeed with probability $1$. It is also observed that the thresholds increase gracefully when the measurement cost is reduced.

\begin{figure}[h]
\begin{center}
\includegraphics[width=0.4\linewidth]{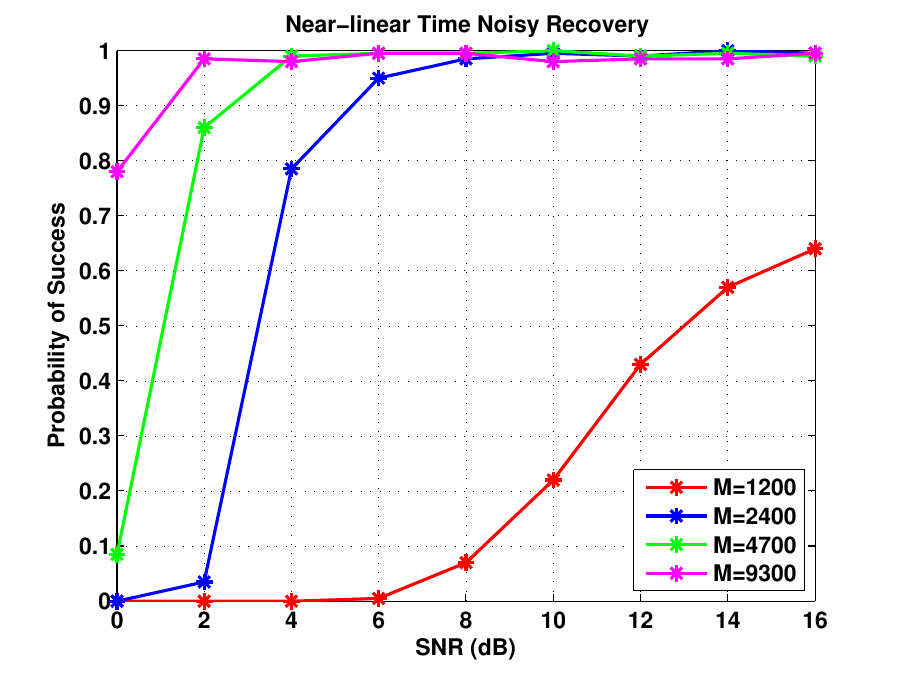}
\includegraphics[width=0.4\linewidth]{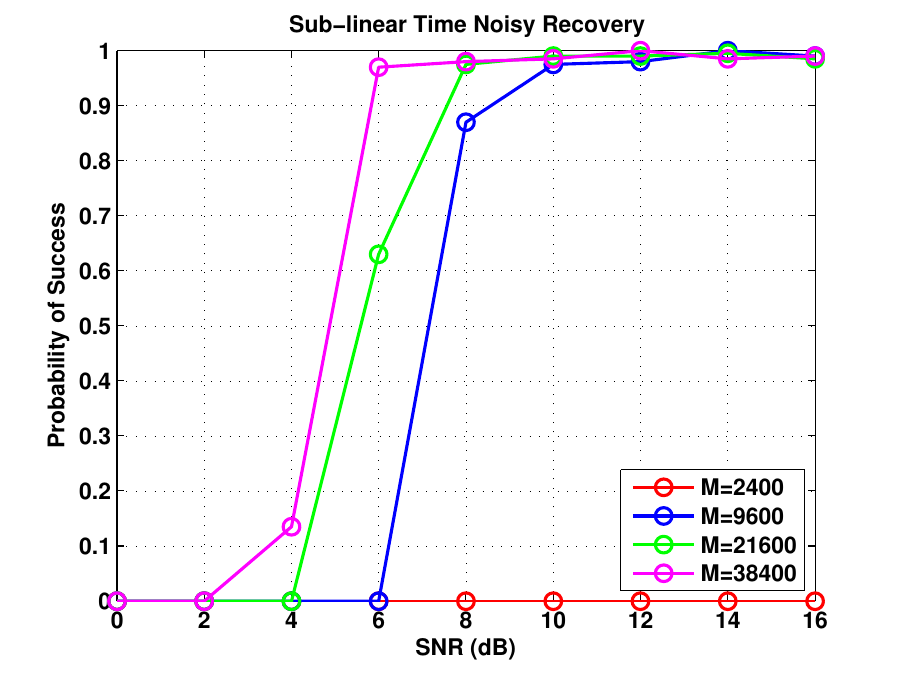}
\caption{Probability of success of near-linear time noisy recovery and sub-linear time noisy recovery against SNR for $N=0.1~\textrm{million}$ and $K=50$. We can see that for a given measurement cost, there exists a threshold of SNR, above which our noisy recovery schemes succeed with probability $1$. It is also observed that the thresholds increase gracefully when the measurement cost is reduced.}
\label{fig:noise_robustness}
\end{center}
\end{figure}

To showcase the scalability, we trace the average measurement cost and run-time for both designs. In each experiment, the sparsity of the $K$-sparse signals $\mathbf{x}$ is chosen as $K=N^\delta$ under different sparsity regimes $\delta = 1/6,1/3$ and $1/2$, while the ambient dimensions of the signals for each sparsity regime ranges from $N = 10^2$ to $N=10^7\approx 10~\textrm{million}$. The measurements are obtained under SNR = $20$dB. As we can see, for both the near-linear time recovery algorithm and the sub-linear time recovery algorithm, the measurement costs scale sub-linearly in the signal dimension $N$. As for the time complexity, the sub-linear algorithm scales as $O(N^{\delta})$, i.e., sub-linear in $N$.

\begin{figure}[!h]
\begin{subfigure}{0.49\textwidth}
\centering
\includegraphics[width=0.8\linewidth]{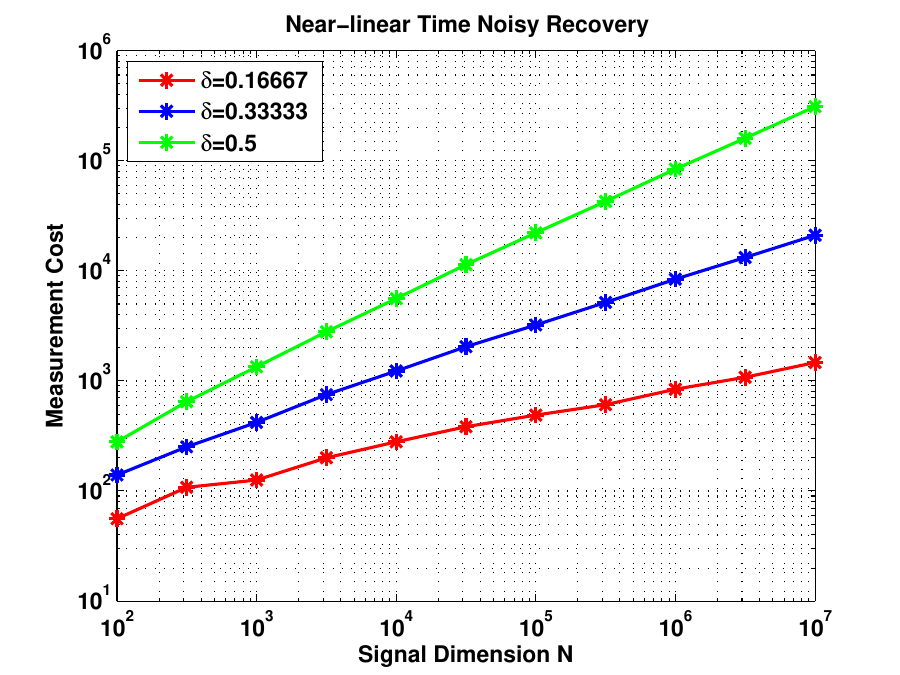}
\caption{Measurement cost for {\it near-linear time} noisy recovery}
\end{subfigure}
\begin{subfigure}{0.49\textwidth}
\centering
\includegraphics[width=0.8\linewidth]{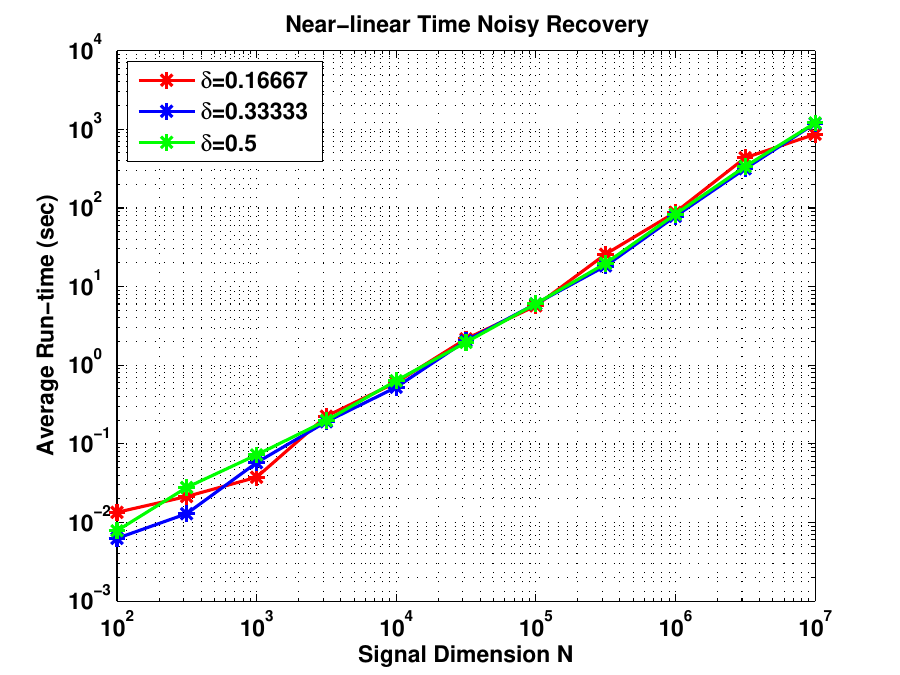}
\caption{Average run-time for {\it near-linear time} noisy recovery}
\end{subfigure}
\begin{subfigure}{0.49\textwidth}
\centering
\includegraphics[width=0.8\linewidth]{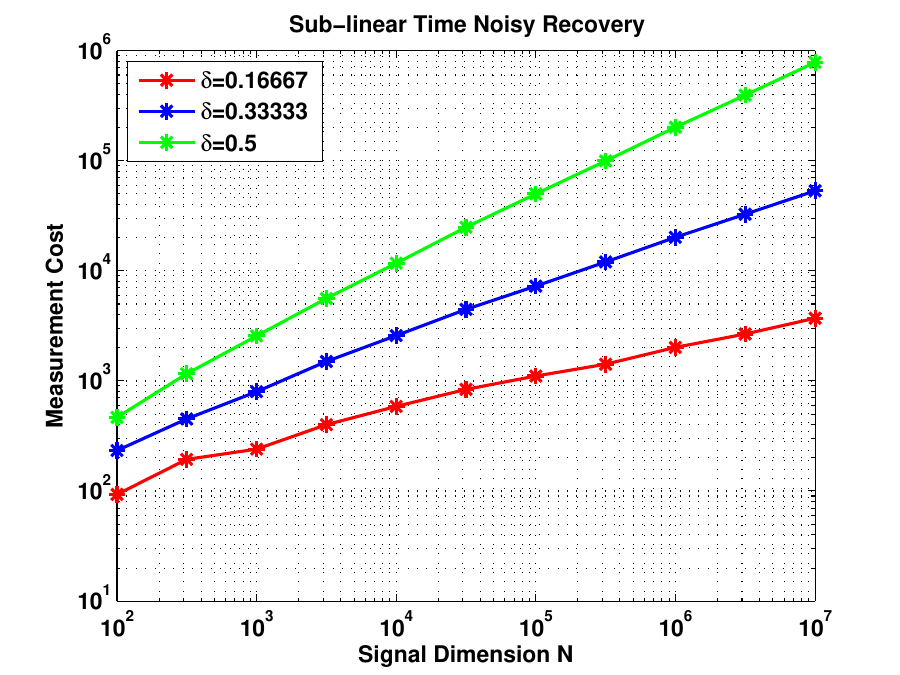}
\caption{Measurement cost for {\it sub-linear time} noisy recovery}
\end{subfigure}
\begin{subfigure}{0.49\textwidth}
\centering
\includegraphics[width=0.8\linewidth]{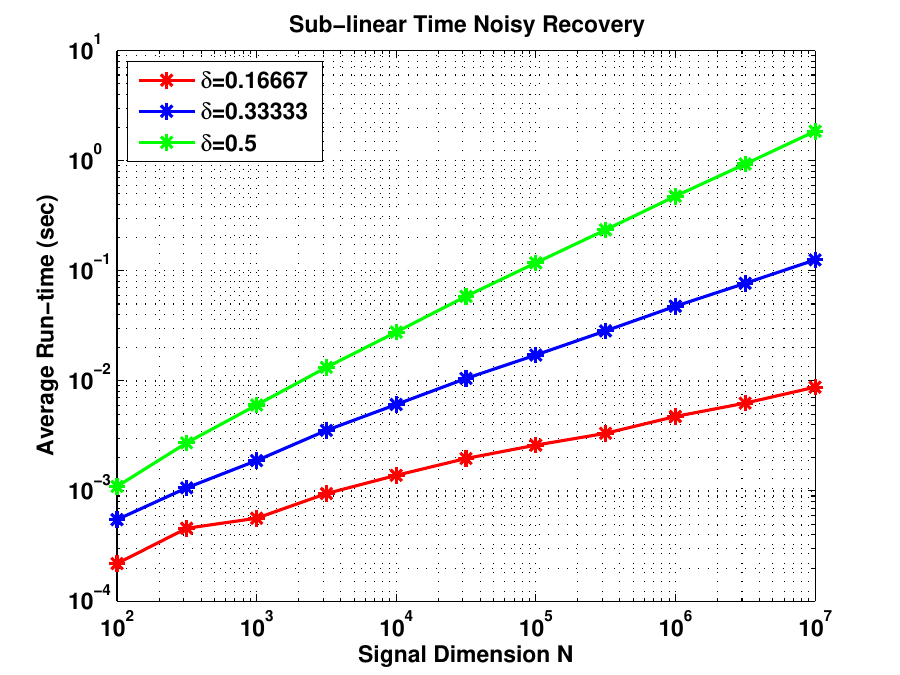}
\caption{Average run-time for {\it sub-linear time} noisy recovery}
\end{subfigure}
\caption{Measurement and computational costs as functions of the signal dimension $N$ for noisy recovery in the quantized alphabet setting. It can be seen that the measurement cost of the near-linear time and sub-linear time designs scale sub-linearly with respect to $N$. For instance, when $N=10$ million and $K=\sqrt{N}$ (the green curves on both plots), the measurement costs for both schemes are approximately $10^6$. We can also see that the run-time for the sub-linear time recovery algorithm indeed scales sub-linearly in $N$. For example, when choosing $K=N^{1/6}$, the red curve in (d) scales as $O(N^{1/6})$.}
\label{fig:noisy_meas_cost}
\end{figure}
 
\subsection{Noise Robustness and Scalabilty in the Continuous Alphabet Setting}\label{sec:noise_robustness_continuous}
For the noisy recovery algorithm in the continuous alphabet setting, we conduct two experiments to test the measurement cost and time complexity.

\begin{figure}[!h]
\begin{subfigure}{0.49\textwidth}
\centering
\includegraphics[width=0.8\linewidth]{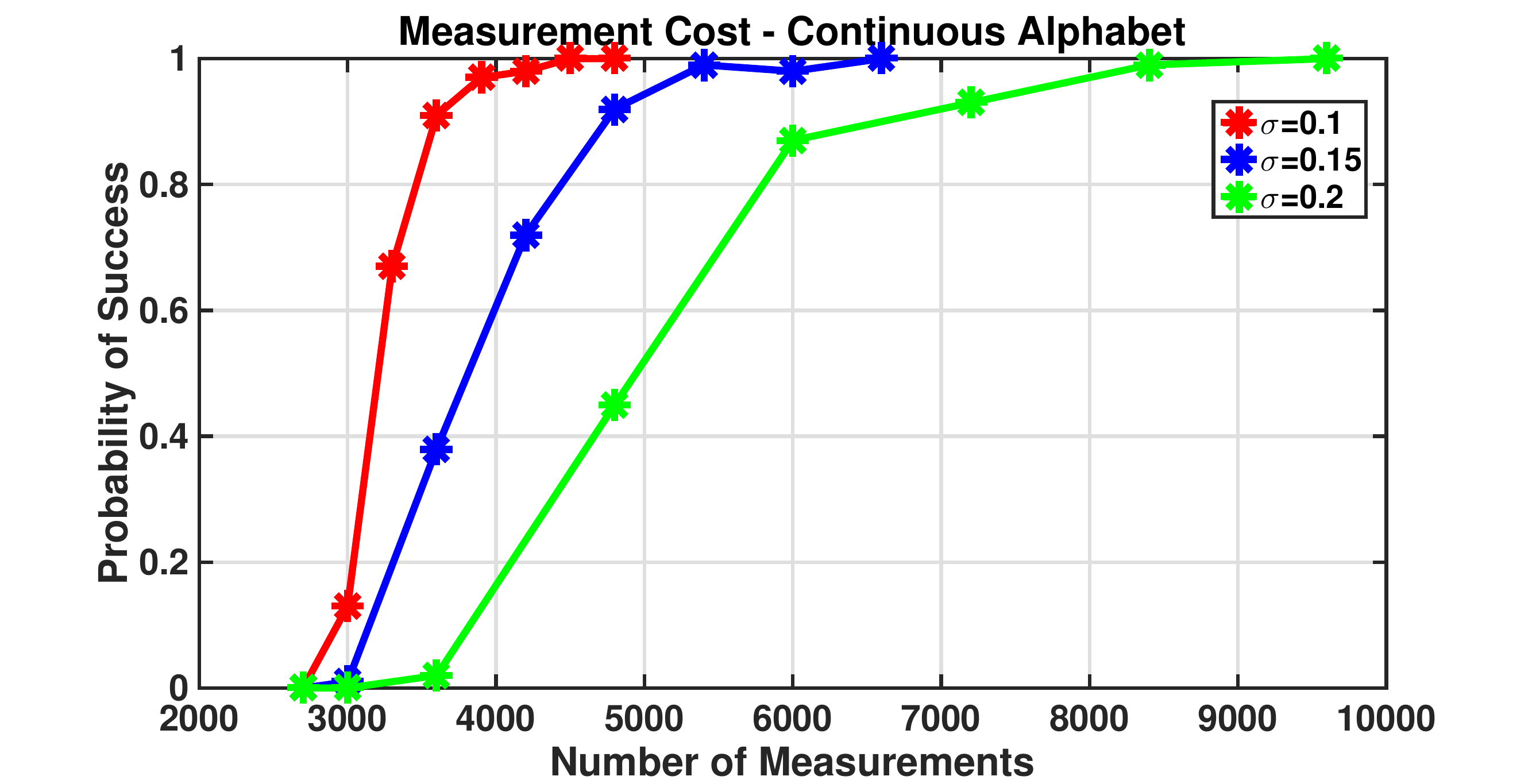}
\caption{Measurement cost in the continuous alphabet setting}
\end{subfigure}
\begin{subfigure}{0.49\textwidth} 
\centering
\includegraphics[width=0.8\linewidth]{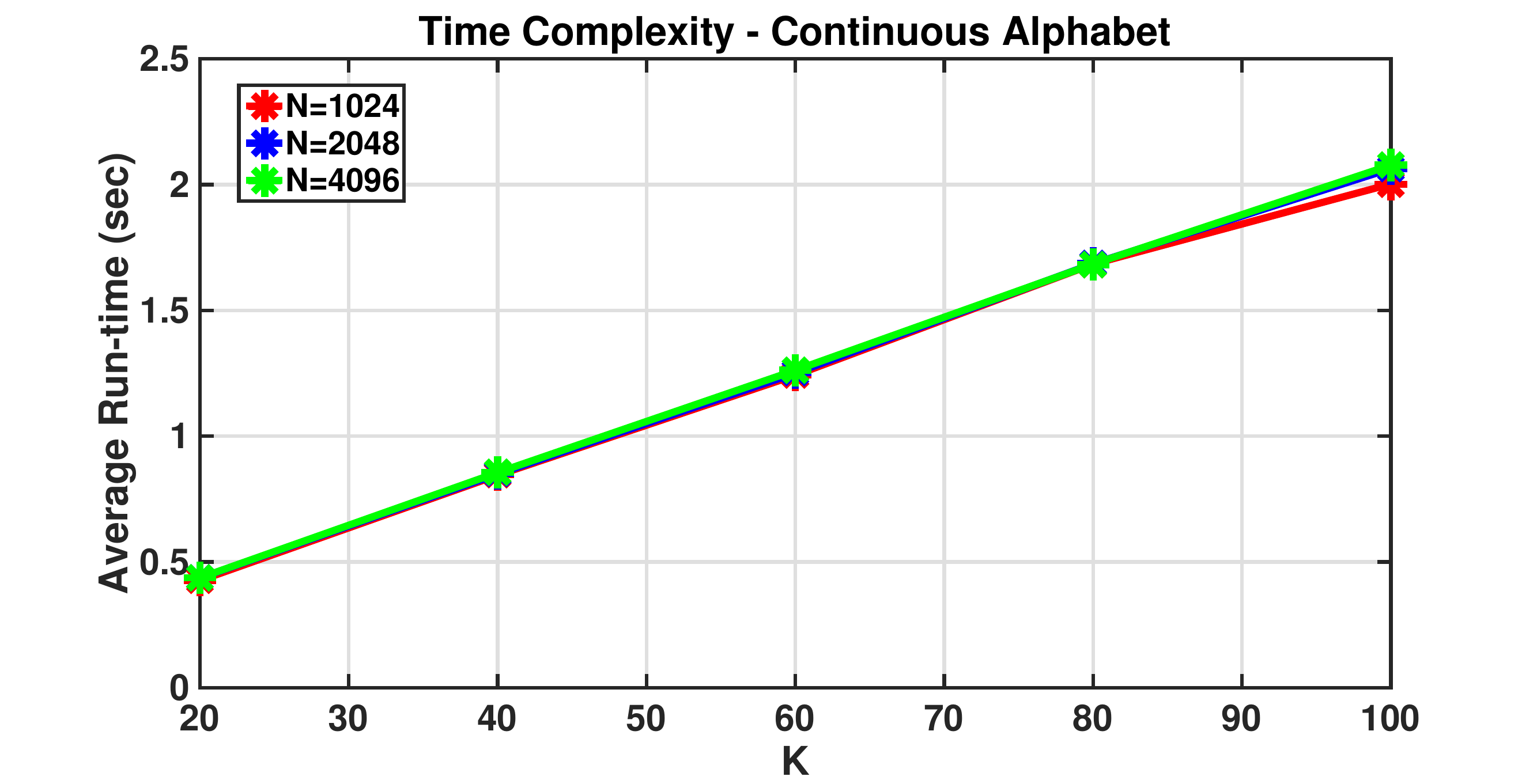}
\caption{Time complexity in the continuous alphabet setting}
\end{subfigure}
\caption{Measurement cost and time complexity of our noisy recovery algorithm in the continuous alphabet setting. It can be seen that when we have enough measurements, we can successfully recover the unknown signal with the $\ell_\infty$ norm guarantee. We can also see that the time complexity of the algorithm increases linear in $K$ but does not have significant dependence on $N$.}
\label{fig:continuous}
\end{figure}
 
In the both experiments, we set the left degree of the random bipartite graph to be $d=10$, and the number of bins to be $R=10K$. The maximum number of sparse coefficients that can be peeled from a bin is set to be $D=5$. The sparse coefficients of the signal are generated from a uniform distribution in $[-10,-3]\cup[3,10]$, and the locations of the sparse coefficients are uniformly chosen from the $N$ coordinates. The additive noise is i.i.d. Gaussian distributed with zero mean. The inner code that we use for the single-ton detection is a $(3, 6)$ regular LDPC code with rate $0.5$.

In the first experiment, we choose $N=4096, K=10$, and test the measurement cost of our algorithm. More specifically, we test how the empirical probability of successful recovery changes when we increase the number of verification measurements in each bin. We define the event of successful recovery as the cases when the supports of $\vect{x}$ and $\widehat{\vecx}$ are the same and $\| \widehat{\vecx} - \vect{x} \|_\infty \le 0.1$. The phase transition behavior under different noise power is shown in Fig.~\ref{fig:continuous} (a).

In the second experiment, we fix the variance of the noise to be $0.1$, and the number of verification measurements in each bin to be $2\log_2 N$. We test the average running time with different $(N, K)$ pairs. As we can see in Fig.~\ref{fig:continuous} (b), the time cost of our algorithm is linear in $K$ and do not have significant dependence on $N$, and this behavior justifies our theory.

\section{Conclusions}\label{sec:conclusion}
In this paper, we have addressed the support recovery problem for compressed sensing using sparse-graph codes. We have proposed a compressed sensing design framework for sub-linear time support recovery, by introducing a new family of measurement matrices and fast recovery algorithms. In the noiseless setting, our framework can recover any arbitrary $K$-sparse signal in $O(K)$ time using $2K$ measurements asymptotically with a vanishing error probability. In the noisy setting, when the sparse coefficients take values in a finite and quantized alphabet and the sparsity $K$ is sub-linear in the signal dimension $K=O(N^\delta)$ for some $0<\delta<1$, our framework can achieve the same goal in time $O(K\log(N/K))$ using $O(K\log(N/K))$ measurements obtained from measurement matrix with elements $\{-1,0,1\}$. In this setting, our results are order-optimal in terms of measurement costs and run-time. For continuous-valued sparse coefficients, our algorithm can recover an arbitrarily large fraction of the support of the sparse signal using $O(K\log(N/K)\log\log(N/K))$ measurements, and $O(K\log^{1+r}(N/K))$ run-time, where $r$ is an arbitrarily small constant. We also obtain recovery guarantees in the $\ell_\infty$ and $\ell_1$ norms. We note that our algorithm is the first algorithm that can achieve both sub-linear measurement cost and time complexity for compressed sensing problems. We also provide simulation results to corroborate our theoretical findings. Our theoretical and experimental results justify that our framework can potentially enable real-time or near-real-time processing for massive datasets featuring sparsity, which are relevant to a multitude of practical applications.

\bibliographystyle{ieeetr}
\bibliography{ref}

\appendices

\section{Peeling Decoder in the Presence of Noise}\label{sec:noisy_recovery_perf_analysis}

Let ${E}_{\rm bin}$ be the event where our ``guess-and-check'' bin detection scheme makes a mistake. From the law of total probability, we have
\begin{align*}
	\Pf &= \Prob{\supp{\widehat{\mathbf{x}}}\neq\supp{\mathbf{x}}\big|{E}_{\rm bin}^c}\Prob{{E}_{\rm bin}^c}
	+\Prob{\supp{\widehat{\mathbf{x}}}\neq\supp{\mathbf{x}}\big|{E}_{\rm bin}}\Prob{{E}_{\rm bin}}\\
	&\leq \Prob{\supp{\widehat{\mathbf{x}}}\neq\supp{\mathbf{x}}\big|{E}_{\rm bin}^c}+\Prob{{E}_{\rm bin}}.
\end{align*}
Since it is known from Theorem \ref{thm_peeling_decoder} that $\Prob{\supp{\widehat{\mathbf{x}}}\neq\supp{\mathbf{x}}\big|{E}_{\rm bin}^c}={O}(1/K)$, then if further we have
\begin{align}\label{P_RBI_upperbound}
	\Prob{{E}_{\rm bin}} = {O}\left(\frac{1}{K}\right),
\end{align}
the overall failure probability can be upper bounded as
\begin{align*}	
	\Pf = {O}\left(\frac{1}{K}\right).
\end{align*}
Now it remains to show that \eqref{P_RBI_upperbound} holds. The main idea is to analyze the error probability of making at least an error on any bin measurement, followed by a union bound on all the $R={O}(K)$ bins. Denote the error event in any bin $j$ as $E_j$, then we have the following union bound across $R=\eta K$ measurement bins
\begin{align}
	\Prob{{E}_{\rm bin}} \leq \bigcup_{j=1}^{\eta K} \Prob{E_j},
\end{align}
where $\bar{d}$ is the average left degree of the bipartite graph. Without loss of generality, we drop the bin index such that
\begin{align}
	\Prob{{E}_{\rm bin}} \leq \eta K \Prob{E},
\end{align}
where $\Prob{E}$ is the error probability for an arbitrary bin. According to Lemma \ref{lem_random_sensing} and \ref{lem_binary_sensing}, the error probability per bin is at most $\Prob{E} = {O}(1/K^2)$, and therefore the overall probability of error is $\Prob{{E}_{\rm bin}} = {O}(1/K)$.

\section{Oracle-based Peeling Decoder using the Regular Ensemble $\mathcal{G}_{\rm reg}^N(R,d)$}\label{proof_lem_peeling_decoder}

\subsection{\bf Concentration Analysis}\label{sec:concentration_analysis}

\subsubsection{Proof of Mean Analysis on General Graphs}
From \eqref{sum_edges}, we have
\begin{align}\label{mean_analysis}
	\mathbb{E}\left[Z_{i}\right] = \sum_{e=1}^{K{d}} \mathbb{E}\left[Z_{i}^{(\textrm{e})}\right]  = K d\mathbb{E}\left[Z_{i}^{(\textrm{e})}\right].
\end{align}
From basic probability laws on conditional expectations
\begin{align*}
	\mathbb{E}\left[Z_{i}^{(\textrm{e})}\right]
	&=
	\mathbb{E}\left[Z_{i}^{(\textrm{e})}|\mathcal{T}_{i}\right]\Prob{\mathcal{T}_{i}}
	+\mathbb{E}\left[Z_{i}^{(\textrm{e})}|\mathcal{T}_{i}^c\right]\Prob{\mathcal{T}_{i}^c}.
\end{align*}
Recall from the density evolution analysis that $\mathbb{E}\left[Z_{i}^{(\textrm{e})}|\mathcal{T}_{i}\right] = p_i$, we have
\begin{align}
	\Prob{\mathcal{T}_{i}} \leq 1,\quad
	\mathbb{E}\left[Z_{\textrm{e}}|\mathcal{T}_{i}^c\right] \leq 1
\end{align}
and therefore the following holds:
\begin{align}
	{p}_{i} - \Prob{\mathcal{T}_{i}^c}
	\leq
	\mathbb{E}\left[Z_{i}^{(\textrm{e})}\right] 
	\leq
	{p}_{i} + \Prob{\mathcal{T}_{i}^c}.
\end{align}
If the probability of a general graph not behaving like a tree can be made arbitrarily small for any $\varepsilon>0$,
\begin{align}\label{prob_kot_tree}
	 \Prob{\mathcal{T}_{i}^c} < \frac{\varepsilon}{4},
\end{align}
then we can obtain the result in \eqref{mean_on_general_graph} by letting $p_i=\varepsilon/4$ in the density evolution analysis. Next, we show that \eqref{prob_kot_tree} holds for sufficiently large $K$.
\begin{lem}
For any given constant $\varepsilon>0$ and iteration $i>0$, there exists some absolute constant $K_0>0$ such that 
\begin{align}
	\Prob{\mathcal{T}_{i}^c}
	<
	c_0 \frac{\log^i K}{K}
\end{align}
for some constant $c_0>0$ as long as $K>K_0$.
\end{lem}
From this lemma, we can see that for an arbitrary $\varepsilon>0$, the result follows as long as $K>K_0$ where $K_0$ is the smallest constant that satisfies $K_0/\log^i K_0 > 4c_0/\varepsilon$ given $\varepsilon$ and $i$. In the following we give the proof of the lemma.

\begin{proof}
Let $C_j$ be the number of check nodes and $V_j$ be the number of variable nodes in the neighborhood $\mathcal{N}_{\textrm{e}}^{2j}$. In \cite{richardson2001capacity}, it has been shown that the directed neighborhood $\mathcal{N}_{\textrm{e}}^{2i}$ at depth $i$ is not a tree with probability at most ${O}(1/K)$. However, the proof therein largely rests on the regular degrees for the left and right nodes in the graph. Now, because the graph ensemble $\mathcal{G}_{\rm reg}^N(R, {d})$ follows Poisson distributions on the right, the results in \cite{richardson2001capacity} are not immediately applicable here. In this setting, the key idea is to prove that the size of the directed neighborhood $\mathcal{N}_{\textrm{e}}^{2i}$ unfolded up to depth $2i$ is bounded by ${O}(\log^i K)$ with high probability, and this neighborhood is not a tree with probability at most ${O}({\log^i K}/{K})$.

To show this, we unfold the neighborhood of an edge $e$ up to level $i$. Fix some constant $\kappa_1$, then at each level $j\leq i$ we upper bound the probability of a tree having more than ${O}(\log^j K)$ left nodes $V_j > \kappa_1 \log^j K$ and right nodes $C_j > \kappa_1 \log^j K$. Specifically, from the law of total probability, we upper bound the probability for some constant $\kappa_1>0$
\begin{align}\label{temp_eq4}
	\Prob{\mathcal{T}_{i}^c}
	&\leq 
	\Prob{V_j > \kappa_1 \log^j K}
	+ \Prob{C_j > \kappa_1 \log^j K}\\
	&~~~ + \Prob{\mathcal{T}_{i}^c | V_j < \kappa_1 \log^j K,C_j < \kappa_1 \log^j K}.
\end{align}
Denoting the first term in \eqref{temp_eq4} as $a_j = \Prob{V_j > \kappa_1 \log^j K}$, we bound $a_j$ using the total law of probability as follows
\begin{align}\label{temp_recursion}
		a_j
		&\leq
		a_{j-1} +
		\Prob{V_j > \kappa_1 \log^j K | V_{j-1} < \kappa_1  \log^{j-1} K}.
\end{align}
Since the left degree from the regular and irregular ensembles is upper bounded by constants ${d}$ and $(D+1)$ respectively, thus given $V_{j-1} < \kappa_1  \log^{j-1} K$ at depth $(j-1)$, the number of right neighbors is bounded by $C_{j-1} <\kappa_2 \log^{j-1} K$ for some $\kappa_2>0$. Therefore, the second term in \eqref{temp_recursion} can be bounded as
\begin{align}
	\Prob{V_j > \kappa_1 \log^j K | V_{j-1} < \kappa_1  \log^{j-1} K}
	&\leq
	\Prob{V_j > \kappa_1 \log^j K | C_{j-1} < \kappa_2 \log^{j-1} K}.
\end{align}
Now let the number of check nodes at exactly depth $(j-1)$ be $C_{j-1}'$ such that $C_{j-1} = C_{j-1}'+C_{j-2}$, and further let $d_\ell$ be the degree of each check node at this depth $\ell=1,\cdots,C_{j-1}'$, then the right hand side can be evaluated as
\begin{align}\label{temp_eq3}
	&\Prob{V_j > \kappa_1 \log^j K | C_{j-1} < \kappa_2 \log^{j-1} K}
	\leq
	\Prob{\sum_{\ell=1}^{C_{j-1}'} d_\ell \geq \kappa_3 \log^j K}
\end{align}
for some $\kappa_3>0$. Since each check node degree $d_\ell$ is an independent Poisson variable with rate $1/\eta$, the sum of $d_\ell$ over $\ell=1,\cdots,C_{j-1}'$ remains a Poisson variable with rate $C_{j-1}'/\eta$. Since obviously $C_{j-1}'<C_{j-1} < \kappa_1 \log^{j-1} K$ such that the sum rate is $C_{j-1}'/\eta = {O}(\log^{j-1} K)$. With this sum rate, the probability in \eqref{temp_eq3} can be upper bounded with the tail bound of a Possible variable $X$ with rate $\lambda$ as $\Prob{X\geq x}\leq (\lambda e/ x)^x$:
\begin{align*}
	\Prob{\sum_{\ell=1}^{C_{j-1}'} d_\ell \geq \kappa_3 \log^j K}
	&\leq \left(\frac{e C_{j-1}'/\eta}{\kappa_3 \log^j K}\right)^{\kappa_3\log^j K}
	= \left(\frac{e \times {O}(\log^{j-1} K)}{\kappa_3 \log^j K}\right)^{\kappa_3\log^j K}
	\leq \left(\frac{\kappa_4}{\log K}\right)^{\kappa_3\log^j K} \leq \frac{\kappa_5}{K}	
\end{align*}
for some sufficiently large constants $\kappa_4>0$ and $\kappa_5>0$. Therefore we have
\begin{align}
	\alpha_j \leq \alpha_{j-1} + \frac{\kappa_5}{K}
\end{align}
and thus the number of variable nodes exposed until the $i$-th iteration can be bounded by $\log^j K$ with high probability 
\begin{align}
	\Prob{V_j > \kappa_1 \log^j K} = {O}\left(\frac{1}{K}\right). 
\end{align}	
Similar technique can be used to show that the tail bound for the check nodes is 
\begin{align}
	\Prob{C_j > \kappa_1 \log^j K} = {O}\left(\frac{1}{K}\right).
\end{align}

Now that it has been shown that the number of nodes is well bounded by ${O}(\log^j K)$, we can proceed to bound the second term in \eqref{temp_eq4} by induction. Assuming that  the neighborhood $\mathcal{N}_{\textrm{e}}^{2j}$ at the $j$-th iteration $(j<i)$ is tree-like, we prove that $\mathcal{N}_{\textrm{e}}^{2(j+1)}$ is tree-like with high probability. First of all, we examine the neighborhood $\mathcal{N}_{\textrm{e}}^{2j+1}$. %Assume that $t$ additional edges have been revealed at this level without forming a cycle. 
The probability that a certain edge from a variable node does not create a cycle in $\mathcal{N}_{\textrm{e}}^{2j+1}$ is the probability that it is connected to one of the check nodes that are not already included in the tree in $\mathcal{N}_{\textrm{e}}^{2j}$, which is lower bound by $1-C_{j}/(\eta K)$. Therefore, given that $\mathcal{N}_{\textrm{e}}^{2j}$ is tree-like, the probability that $\mathcal{N}_{\textrm{e}}^{2j+1}$ is tree-like is lower bounded by 
\begin{align}
	\left(1-\frac{C_{j}}{\eta K}\right)^{C_{j+1}-C_j} > \left(1-\frac{C_{i}}{\eta K}\right)^{C_{j+1}-C_j}.
\end{align}
Similarly, given that $\mathcal{N}_{\textrm{e}}^{2j+1}$ is tree-like, the probability that $\mathcal{N}_{\textrm{e}}^{2(j+1)}$ is tree-like is lower bounded by 
\begin{align}
	\left(1-\frac{V_{j}}{K}\right)^{V_{j+1}-V_j} > \left(1-\frac{V_{i}}{K}\right)^{V_{j+1}-V_j}.
\end{align}
Therefore, the probability that $\mathcal{N}_{\textrm{e}}^{2(j+1)}$ is tree-like is lower bounded by 
\begin{align*}
	\left(1-\frac{C_{i}}{\eta K}\right)^{C_{j}}\left(1-\frac{V_{i}}{K}\right)^{V_{j}}
	\geq 
	\left(1-\frac{C_{i}}{\eta K}\right)^{C_{i}}\left(1-\frac{V_{i}}{K}\right)^{V_{i}}
	\geq 1- \left(\frac{V_{i}^2}{K}+\frac{C_{i}^2}{\eta K}\right)
	\geq 1- {O}\left(\frac{\log^i K}{K}\right).
\end{align*}
Therefore the probability of not being tree-like is upper bounded by 
\begin{align}
	\Prob{\mathcal{T}_{i}^c}< c_0 \frac{\log^i K}{K}
\end{align}
for some absolute constant $c_0>0$.

\end{proof}

\subsubsection{\bf Proof of Concentration to Mean by Large Deviation Analysis}
Now it remains to show the concentration of $Z_{i}$ around its mean $\mathbb{E}[Z_{i}]$. According to \eqref{sum_edges}, the number of remaining edges is a sum of random variables $Z_{i}=\sum_{e=1}^{K{d}} Z_{i}^{\textrm{e}}$ while summands $Z_{i}^{\textrm{e}}$ are not independent with each other. Therefore, to show the concentration, we use a standard martingale argument and Azuma's inequality provided in \cite{richardson2001capacity} with some modifications to account for the irregular degrees of the right nodes. 

Suppose that we expose the whole set of $E=K{d}$ edges of the graph one at a time. We let
\begin{align}
	Y_\ell=\mathbb{E}\left[Z_{i}|Z_{i}^{1},\cdots,Z_{i}^{\ell}\right],\quad \ell=1,\cdots, K{d}.
\end{align}	 
By definition, $Y_0,Y_1,\cdots,Y_{K{d}}$ are a Doob's martingale process, where $Y_0=\mathbb{E}[Z_{i}]$ and $Y_{K{d}}=Z_{i}$. To use Azuma's inequality, it is required that $|Y_{\ell+1}-Y_\ell|\leq \Delta_\ell$ for some $\Delta_\ell>0$.  If the variable node has a regular degree ${d}$ and the check node has a regular degree $d_C$, then \cite{richardson2001capacity} shows that $\Delta_\ell = 8({d} d_C)^i$ with $i$ being the number of peeling iterations. However, the check node degree is not regular with degree $d_C$ and therefore requires further analysis.

\subsubsection*{Proof of Finite Difference $\Delta_\ell$}
To prove that the difference $\Delta_\ell$ is finite for check node degrees with Poisson distributions, we first prove that the degree of all the check nodes can be upper bounded by $d_C\leq {O}(K^{\frac{2}{4i+1}})$ with probability\footnote{Let $X$ be a Poisson variable with parameter $\lambda$, then the following holds
\begin{align*}
	\Prob{X>c K^{\frac{2}{4i+1}} }
	&\leq
	\left(\frac{e\lambda}{c K^{\frac{2}{4i+1}}} \right)^{c K^{\frac{2}{4i+1}}}
	\leq c_1 \exp\left(-c_2 K^{\frac{2}{4i+1}} \right)
\end{align*}
for some $c_1$ and $c_2$.} at least $$c_1K\exp\left(-c_2 K^{\frac{2}{4i+1}} \right)$$ for some constants $c_1$ and $c_2$. Let $\mathcal{B}$ be the event that at least one check node has more than ${O}\left(K^{\frac{2}{4i +1}}\right)$ edges, then for some $c_3>0$ we have
\begin{align}
	\Prob{\mathcal{B}} < c_3K\exp\left(-c_2 K^{\frac{2}{4i+1}} \right).
\end{align}
by applying a union bound on all the $R=\eta K$ check nodes of the graphs from $\mathcal{G}_{\rm reg}^N(R, {d})$. As a result, under the complement event $\mathcal{B}^c$, we have
\begin{align}
	\Delta_\ell^2 = {O}\left(K^{\frac{4i}{4i+1}}\right).
\end{align}

\subsubsection*{Large Deviation by Azuma's Inequality}
For any given $\varepsilon>0$, the tail probability of the event $Z_{i}>K{d}\varepsilon$ can be computed as
\begin{align*}
	\Prob{\left|Z_{i}-\mathbb{E}[Z_{i}]\right| > \frac{K{d} \varepsilon}{2}}
	&\leq
	\Prob{\left|Z_{i}-\mathbb{E}[Z_{i}]\right| > \frac{K{d} \varepsilon}{2}\Big |\mathcal{B}^c} + \Prob{\mathcal{B}}\\
	&\leq 2\exp\left(-\frac{K^2 \bar{d}^2 \varepsilon^2/4}{2\sum_{\ell=1}^{K{d}} \Delta_\ell^2 }\right) + c_3K\exp\left(-c_2 K^{\frac{2}{4i+1}} \right)\\
	&\leq 2\exp\left(-c_4 \varepsilon^2 K^{\frac{1}{4i+1}}\right),
\end{align*}
where $c_4$ is some constant depending on ${d}$, $\eta$ and all the other constants $c_1,c_2,c_3$. This concludes our proof for \eqref{concentration_DE}.

%\end{proof}

\subsection{\bf Proof of Graph Expansion Properties in Lemma \ref{lem_graph_expander}}\label{sec:expander_graph}
Let $\mathcal{S}_v$ denote the event that a variable node subset of size $v$ with at most $\bar{d}|\mathcal{S}_v|/2$ neighbors, whose probability can be obtained readily for any size $|\mathcal{S}_v|=v$ as 
\begin{align}
	\Prob{\mathcal{S}_v} \leq {K \choose v} { \eta K \choose \bar{d} v/2} \left(\frac{v \bar{d}}{2 \eta K}\right)^{\bar{d}v},
\end{align}
where we have used the fact that the number of check nodes is $\eta K$. Using the inequality ${a \choose b}\leq (a e/b)^b$, we have
\begin{align}
	\Prob{\mathcal{S}_v} \leq \left(\frac{v}{K}\right)^{(\bar{d}/2-1) v} c^v\leq \left(\frac{v c^2}{K}\right)^{v/2},
\end{align}
where $c=e(\bar{d}/2\eta)^{\bar{d}/2}$ is some constant. Then a union bound is applied over all possible values $v$ up to the remaining variable nodes $\varepsilon_\star K$. Choosing $\varepsilon_\star<1/(2c^2)$ yields
\begin{align}
	\sum_{v=2}^{\varepsilon_\star K}\Prob{\mathcal{S}_v} 
	\leq
	\sum_{v=2}^{\varepsilon_\star K}\left(\frac{v c^2}{K}\right)^{v/2}
	= O\left(\frac{1}{K}\right).
\end{align}
Therefore, asymptotically in $K$, the random graphs from both the regular and irregular ensembles are good expanders on small sets of variable nodes.

%%%%%%%%%%%%%%%
\section{Oracle-based Peeling Decoder using the Irregular Ensemble $\mathcal{G}_{\rm irreg}^N(R,D)$}\label{sec:proof_thm_peeling_irregular}

Based on the peeling decoder analysis in Section \ref{sec:peeling_decoder_analysis}, it can be easily shown that the concentration analysis and graph expansion property carry over to the irregular graph ensemble. Hence, we focus on the density evolution for the oracle-based peeling decoder over irregular ensemble. 

To study the probability $p_i$ of an edge being present in the pruned graph from the irregular ensemble after $i$ iterations, we need to first understand the right edge degree distributions $\rho_{j}$ of the graph. Using the degree sequence $\lambda_{j}$ of the irregular graph ensemble $\mathcal{G}_{\rm irreg}^N(R,D)$ in Definition \ref{lem_graph_ensemble_irregular}, it can be shown that the right degree sequence $\rho_{j}$ follows a Poisson distribution similar to \eqref{rho_d}
\begin{align*}
	\rho_{j} \approx \frac{\left(\bar{d}/(1+\epsilon) \right)^{{j}-1}e^{-\bar{d}/(1+\epsilon)}}{(j-1)!},
\end{align*}
where we have used $R=(1+\epsilon)K$ and $\bar{d}$ is the average degree of a left node in the irregular graph ensemble
\begin{align}\label{avg_degree}
		\bar{d} = \frac{1}{\sum_{{j}=2}^{D+1} {\lambda_{j}}/{j}} %= \frac{1}{\sum_{d=2}^{D} \frac{1}{d(d-1)}}.
		= H(D)\left(1+\frac{1}{D}\right).
	\end{align}
Using the left and right degree sequence $(\lambda_{j},\rho_{j})$, we can readily obtain the left and right degree generating polynomials $\lambda(x)=\sum_{{d}=1}^\infty \lambda_{j} x^{{j}-1}$ and $\rho(x)=\sum_{j=1}^\infty \rho_{j} x^{{j}-1}$ 
\begin{align*}
	\lambda(x)&= \frac{1}{H(D)}\sum_{{j}=2}^{D+1} \frac{1}{({j}-1)}x^{{j}-1},\quad \rho(x)=e^{-\frac{\bar{d}}{1+\epsilon}(1-x)}.
\end{align*}		

As a result, the associated density evolution equation can be written using the degree generating polynomials similar to that in \eqref{density_evolution} 
		\begin{align}\label{density_evolution_irregular}
			p_i 
			&= f(p_{i-1}) = \lambda(1-\rho(1-{p}_{i-1})),\quad i = 1,2,3,\cdots.			
		\end{align}
The density evolution analysis suggests that if the fraction $p_i$ in \eqref{density_evolution_irregular} can be made arbitrarily small if the density evolution recursion is contracting 
\begin{align}\label{contraction_mapping}
	\lambda(1-\rho(1-x)) < x,\quad \forall~x\in[0,1]. 
\end{align}
Examples of this density evolution using different values of $D$ and $\epsilon$ are given in \figref{fig:DE_irregular}.
\begin{figure}[p]
\centering
\begin{subfigure}{0.4\textwidth}
\includegraphics[width=1\linewidth]{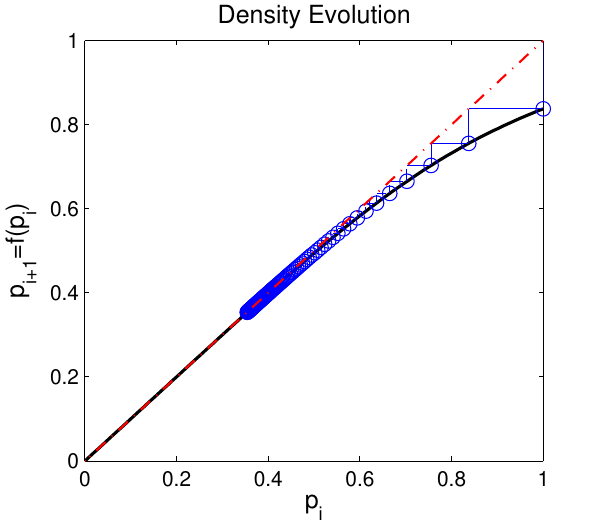}
\vspace{-0.6cm}
\caption{$\epsilon=0.1$ and $D=10$}
\end{subfigure}
\begin{subfigure}{0.4\textwidth}
\includegraphics[width=1\linewidth]{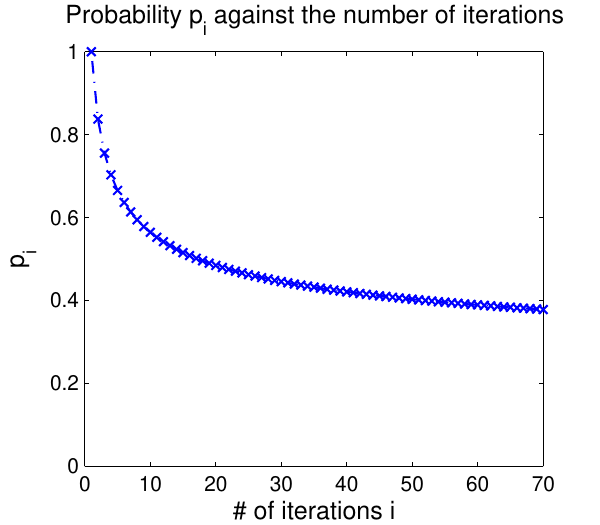}
\vspace{-0.6cm}
\caption{$\epsilon=0.1$ and $D=10$}
\end{subfigure}
\begin{subfigure}{0.4\textwidth}
\includegraphics[width=1\linewidth]{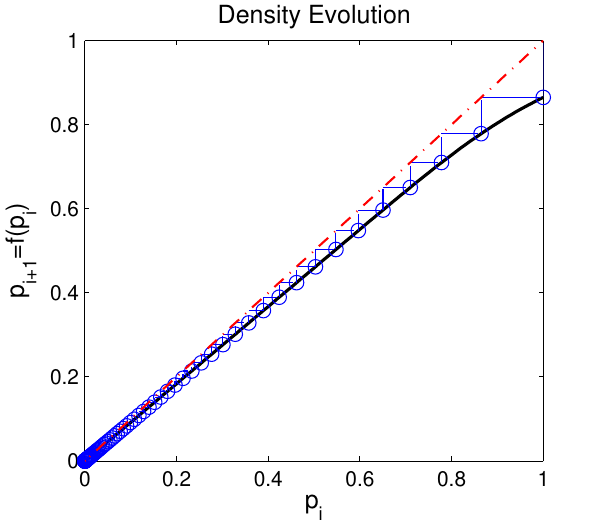}
\vspace{-0.6cm}
\caption{$\epsilon=0.1$ and $D=100$}
\end{subfigure}
\begin{subfigure}{0.4\textwidth}
\includegraphics[width=1\linewidth]{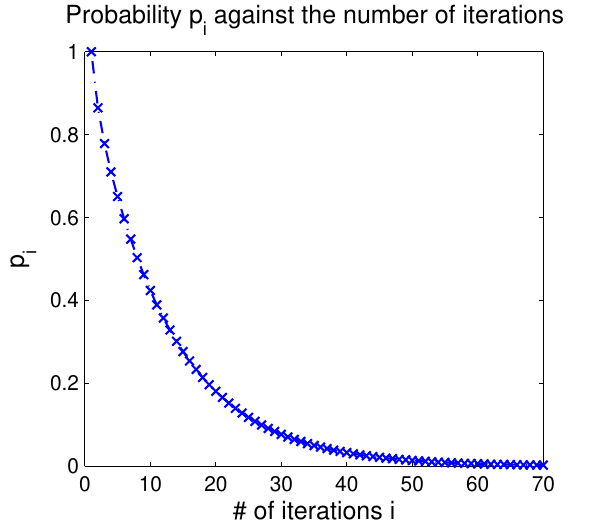}
\vspace{-0.6cm}
\caption{$\epsilon=0.1$ and $D=100$}
\end{subfigure}
\begin{subfigure}{0.4\textwidth}
\includegraphics[width=1\linewidth]{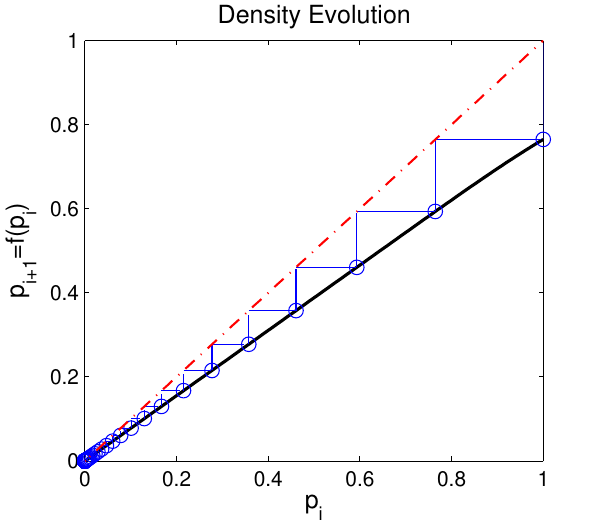}
\vspace{-0.6cm}
\caption{$\epsilon=0.3$ and $D=100$}
\end{subfigure}
\begin{subfigure}{0.4\textwidth}
\includegraphics[width=1\linewidth]{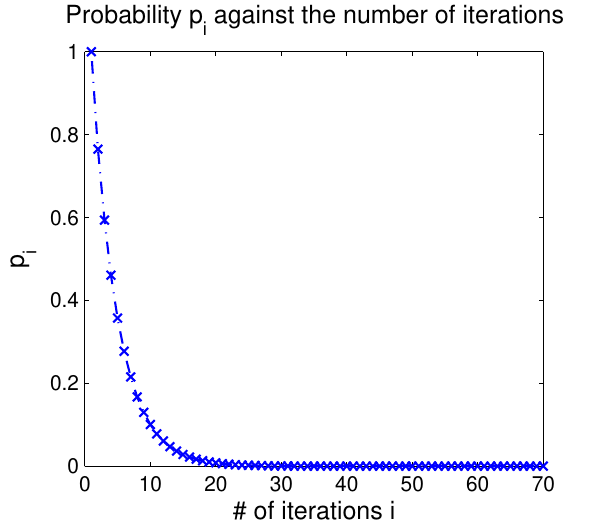}
\vspace{-0.6cm}
\caption{$\epsilon=0.3$ and $D=100$}
\end{subfigure}
\caption{The density evolution $f(p_i)$ and the probability $p_i$ at each iteration $i$, where we have shown cases with $\epsilon=0.1$ and $D=10$ and $D=100$, as well as the case with $\epsilon=0.3$ and $D=100$. In the density evolution figures (a)-(c)-(e), the red line is the line $p_{i+1}=p_i$ while the black line is the density evolution $f(p_i)$ against $p_i$. The blue circles that ``zig-zag'' between the red line and the black line are the specific $p_i$'s at each peeling iteration. It can be seen from (a) and (c) that when $\epsilon$ is small (i.e. $\epsilon =0.1$), the density evolution requires a large maximum left degree $D$ to reach density $0$. On the other hand, when $\epsilon$ is large (i.e. $\epsilon=0.3$), the density $p_i$ reaches $0$ very quickly in (e) with the same maximum left degree $D=100$. The values of $p_i$ marked by the blue circles in (a)-(c)-(e) are further plotted against the peeling iterations $i$ in (b)-(d)-(f), where in the case with $\epsilon=0.3$ and $D=100$ the density $p_i$ approaches $0$ after less than $20$ iterations.}\label{fig:DE_irregular}
\end{figure}
Clearly, when $\epsilon=0.1$, the density evolution equation becomes a contraction mapping when $D=100$ but not when $D=10$. Now we study how to choose $D$ for any given $\epsilon>0$. Since $\lambda(x)$ is a non-decreasing function, we can apply $x = \lambda^{-1}({p}_{i-1})$ on both sides of \eqref{contraction_mapping}, then the contraction condition is equivalent to
	\begin{align}
		\rho(1-\lambda(x)) > 1-x,\quad \forall x\in[0,1].
	\end{align}
	By substituting the right generating polynomial $\rho(x)$ into the above recursion, we have
	\begin{align}\label{temp1}
		\rho(1-\lambda(x)) = e^{-\frac{\bar{d}}{(1+\epsilon)}\lambda(x)}.
	\end{align}
	To simplify our expressions, we further bound $\lambda(x)$ for the irregular graph ensemble $\mathcal{G}_{\rm irreg}^N(R,D)$ as $\lambda(x) > -\frac{1}{H(D)}\log(1-x)$. This is because $\lambda(x)$ is a $D$-term approximation of the Taylor expansion for $\log(1-x)$, scaled by the normalization constant $H(D)$. By substituting this bound into \eqref{temp1}, we have
		\begin{align*}
			\rho(1-\lambda(x)) > e^{\frac{\bar{d}}{(1+\epsilon)}\frac{1}{H(D)}\log(1-x)} = (1-x)^{\frac{\bar{d}}{(1+\epsilon) H(D)}}.
		\end{align*}
		It can be seen that the right hand side is no less than $1-x$ as long as $H(D)\geq \frac{\bar{d}}{(1+\epsilon)}$. Substituting the average degree $\bar{d}$ from \eqref{avg_degree} back to this condition, then for any $\epsilon>0$, we can choose $D>1/\epsilon$ as in Definition \ref{lem_graph_ensemble_irregular} to render the recursion a contracting mapping. 
		
Finally, together with the concentration analysis and graph expansion properties of the irregular graphs, the oracle-based peeling decoder successfully decodes all the edges in the graph with probability at least $1-{O}(1/K)$.

%%%%%%%%%%%%%%%%%%
\section{Proof of Lemma \ref{lem_random_sensing}}\label{proof_lem_random_sensing}
%%%%%%%%%%%%%%%%%%

\begin{defi}\label{def_error_category}
Denoting by $\Prob{E}$ the error probability of the robust bin detection algorithm for an arbitrary bin, we can bound $\Prob{E}$ as
\begin{align}
	\Prob{E} 
	&\leq \Prob{\mathcal{H}_{\textrm{S}}(k,x[k])}+ \sum_{\mathcal{F}\in\{\mathcal{H}_{\textrm{Z}},\mathcal{H}_{\textrm{M}}\}}\Prob{\mathcal{H}_{\textrm{S}}(k,x[k])\leftarrow\mathcal{F}}
\end{align}
where $\mathcal{F}$ is either a zero-ton $\mathcal{H}_{\textrm{Z}}$  or a multi-ton $\mathcal{H}_{\textrm{M}}$ and
\begin{enumerate}
	\item $\Prob{\mathcal{H}_{\textrm{S}}(k,x[k])}$ is called the {\bf missed verification} rate in which %either the single-ton search provides a wrong estimate such that $\widehat{k}\neq k$ and/or $\widehat{x}[\widehat{k}]\neq x[k]$, or 
	the single-ton verification fails even when the underlying bin is a single-ton $\bdsb{y}\sim\mathcal{H}_{\textrm{S}}(k,x[k])$ for some $k\in[N]$ and $x[k]$.  
	\item $\Prob{\mathcal{H}_{\textrm{S}}(k,x[k])\leftarrow\mathcal{F}}$ is called the {\bf false verification} rate in which the single-ton verification is passed for some single-ton $\mathcal{H}_{\textrm{S}}(\widehat{k},\widehat{x}[\widehat{k}])$ with an index-value pair $(\widehat{k},\widehat{x}[\widehat{k}])$ when the ground truth is $\mathcal{F}\in\{\mathcal{H}_{\textrm{Z}},\mathcal{H}_{\textrm{M}}\}$.
\end{enumerate}
\end{defi}

Now we compute the probability mentioned above in the following propositions. 

\begin{prop}[False Verification Rate]\label{prop_false_detection}
For some constant $\gamma\in(0,1)$, the false verification rate can be upper bounded as
\begin{align*}%\label{rate_false_detection}
		&\Prob{\mathcal{H}_{\textrm{S}}(\widehat{k},\widehat{x}[\widehat{k}])\leftarrow\mathcal{H}_{\textrm{Z}}}< e^{-\frac{P}{4}(1-\gamma)^2\left(\frac{\SNRmin}{1+\SNRmin}\right)^2}\\
		&\Prob{\mathcal{H}_{\textrm{S}}(\widehat{k},\widehat{x}[\widehat{k}])\leftarrow\mathcal{H}_{\textrm{M}}}<
		e^{-\frac{P}{4}(3-\gamma)^2 \left(\frac{\SNRmin}{1+3\SNRmin}\right)^2}.
\end{align*}
\end{prop}
\begin{proof}
	See Appendix \ref{proof_prop_false_detection}.
\end{proof}

%%%%%

\begin{prop}[Missed Verification Rate]\label{prop_miss_detection}
For some constant $\gamma\in(0,1)$, the missed verification rate can be upper bounded as
\begin{align*}%\label{rate_miss_detection}
		&\Prob{\mathcal{H}_{\textrm{S}}(k,x[k])}< 
		e^{-\frac{P}{4}\left(\sqrt{1+2\gamma\SNRmin} - 1\right)^2} + 2(N-1)\left(e^{-\frac{P}{4}\SNRmin} + e^{-\frac{P}{16}}\right).
\end{align*}
\end{prop}
\begin{proof}
	See Appendix \ref{proof_prop_miss_detection}.
\end{proof}
Without loss of generality, let us choose $\gamma=1/2$ and thus all the error probabilities vanish at a rate $O(1/N^c)$ as long as $P \geq \alpha \log N$, where $\alpha$ satisfies:
\begin{align}\label{req_alpha}
	\begin{cases}
	\alpha \geq 16c\left(1+\frac{1}{\SNRmin}\right)^2\\
	\alpha \geq \frac{16c}{9} \left(1+\frac{3}{\SNRmin}\right)^2\\
	\alpha \geq \frac{4c}{\left(\sqrt{1+\SNRmin} - 1\right)^2}\\
	\alpha \geq \frac{16(c+1)}{\SNRmin}\\
	\alpha \geq 16(c+1)
	\end{cases}.
\end{align}
Therefore, it is sufficient to have $\alpha \geq 16(c+1)(1+\SNRmin^{-1})$ at high SNR regime (i.e. $\SNRmin \gg 1$) and $\alpha \geq 16(c+1)/\SNRmin^2$ at low SNR regime (i.e. $\SNRmin \ll 1$). Letting $c=2\delta$ such that $O(1/N^c) = O(1/K^2)$, we have the claimed result.

\subsection{Proof of False Verification Rates in Proposition \ref{prop_false_detection}}\label{proof_prop_false_detection}
The false verification events occur if the zero-ton or single-ton verifications fail when the ground truth is either a zero-ton or a multi-ton 
\begin{align}
	\bdsb{y} = \mathbf{S}\mathbf{z} +\mathbf{w}
\end{align}
with $\mathbf{z}$ being a zero-ton $\mathbf{z}=\mathbf{0}$ or a multi-ton $\left|\supp{\mathbf{z}}\right| > 1$.

\subsubsection{Detecting a Zero-ton as a Single-ton}
This event happens when a zero-ton $\bdsb{y}=\mathbf{w}$ passes the single-ton verification:
\begin{align}
	\Prob{\mathcal{H}_{\textrm{S}}(\widehat{k},\widehat{x}[\widehat{k}])\leftarrow\mathcal{H}_{\textrm{Z}}}
	=
	\Prob{\frac{1}{P}\left\| \bdsb{y} - \widehat{x}[k]\mathbf{s}_{\widehat{k}}\right\|^2 \leq (1+\gamma \SNRmin) \sigma^2}
\end{align}
Substituting $\bdsb{y}=\mathbf{w}\sim\mathcal{N}(\mathbf{0},\sigma^2 \mathbf{I})$, clearly we have
\begin{align}\label{eq:special_case}
	\bdsb{y}-\widehat{x}[\widehat{k}]\mathbf{s}_{\widehat{k}}
	\sim
	\mathcal{N}(\mathbf{0}, (\rho^2 + \sigma^2) \mathbf{I}).
\end{align}
Therefore, the probability can be bounded by a chi-squared tail:	
\begin{align}
	\Prob{\mathcal{H}_{\textrm{S}}(\widehat{k},\widehat{x}[\widehat{k}])\leftarrow\mathcal{H}_{\textrm{Z}}}
	=\Prob{\frac{1}{P} \chi_P^2 \leq  (1+\gamma\SNRmin)\frac{\sigma^2}{\rho^2 + \sigma^2}}
	\leq 
	e^{-\frac{(1-\gamma)^2}{4}\left(\frac{\SNRmin}{1+\SNRmin}\right)^2\times P}.
\end{align}

\subsubsection{Detecting a Multi-ton as a Single-ton}
By definition, the error probability can be evaluated under the multi-ton model 
\begin{align}
	\bdsb{y}  
	%&= \sum_{\bdsb{k} \in \mathcal{K}} x[k] \mathbf{s}_k + \mathbf{w}
	= \mathbf{S}\mathbf{z} + \mathbf{w}
\end{align}
when it passes the single-ton verification step for some index-value pair $(\widehat{k},\widehat{x}[\widehat{k}])$
\begin{align*}
	&\Prob{\mathcal{H}_{\textrm{S}}(\widehat{k},\widehat{x}[\widehat{k}])\leftarrow\mathcal{H}_{\textrm{M}}}
	=\Prob{\frac{1}{P}\left\|\bdsb{y}-\widehat{x}[\widehat{k}]\mathbf{s}_{\widehat{k}}\right\|^2 \leq  (1+\gamma \SNRmin)\sigma^2}
\end{align*}
for some $\widehat{k}$ and $\widehat{x}[\widehat{k}]$. Clearly, according to the bin detection matrix given in Definition \ref{def_random_sensing}, we have
\begin{align}\label{eq:general_case}
	\bdsb{y}-\widehat{x}[\widehat{k}]\mathbf{s}_{\widehat{k}}
	\sim
	\mathcal{N}(\mathbf{0}_{P\times 1}, \sigma_u^2 \mathbf{I}_{P\times P}),\quad \sigma_u^2= \left\|\mathbf{z} -\widehat{x}[\widehat{k}]\mathbf{e}_{\widehat{k}}\right\|^2 + \sigma^2.
\end{align}
Therefore, the probability can be bounded by a chi-squared tail:	
\begin{align}
	\Prob{\mathcal{H}_{\textrm{S}}(\widehat{k},\widehat{x}[\widehat{k}])\leftarrow\mathcal{H}_{\textrm{M}}}
	=\Prob{\frac{1}{P} \chi_P^2 \leq  (1+\gamma\SNRmin)\frac{\sigma^2}{\sigma_u^2}}.
\end{align}
As long as $\gamma \SNRmin < \sigma_u^2/\sigma^2$, this tail can be obtained from Lemma \ref{general_tail} as:
\begin{align}
	\Prob{\mathcal{H}_{\textrm{S}}(\widehat{k},\widehat{x}[\widehat{k}])\leftarrow\mathcal{H}_{\textrm{M}}}
	\leq
	\exp\left(-\frac{P}{4} \left(1-(1+\gamma\SNRmin)\frac{\sigma^2}{\sigma_u^2}\right)^2\right).
\end{align}
We further bound this quantity with the worst case where the underlying multi-ton consists of two coefficients. Thus, we have $\left\|\mathbf{z} -\widehat{x}[\widehat{k}]\mathbf{e}_{\widehat{k}}\right\|^2 = 3\rho^2$ and $\sigma_u^2 = 3\rho^2+\sigma^2$. As a result, we have
\begin{align}
	\Prob{\mathcal{H}_{\textrm{S}}(\widehat{k},\widehat{x}[\widehat{k}])\leftarrow\mathcal{H}_{\textrm{M}}}
	\leq
	e^{-\frac{(3-\gamma)^2}{4} \left(\frac{\SNRmin}{1+3\SNRmin}\right)^2\times P}.
\end{align}

%%%%%%%%%%%%%

\subsection{Proof of Missed Verification Rates in Proposition \ref{prop_miss_detection}}\label{proof_prop_miss_detection}

The missed verification events occur if the zero-ton or single-ton verifications pass when the ground truth is a single-ton $\mathcal{H}_{\textrm{S}}(k,x[k])$ for some $k\in[N]$:
\begin{align}
	\bdsb{y} = \mathbf{S}\mathbf{z} +\mathbf{w} = \mathbf{s}_{k}x[k] + \mathbf{w}.
\end{align}

This event occurs when the ground truth is a single-ton $\mathcal{H}_{\textrm{S}}(k,x[k])$ with an index-value pair $(k,x[k])$, but the single-ton verification fails for some index-value pair $(\widehat{k},\widehat{x}[\widehat{k}])$ obtained from the single-ton search:
\begin{align*}
	\Prob{\mathcal{H}_{\textrm{S}}(k,x[k])}
	&=\Prob{\frac{1}{P}\left\|\bdsb{y}-\widehat{x}[\widehat{k}]\mathbf{s}_{\widehat{k}}\right\|^2\geq (1+\gamma\SNRmin)\sigma^2}.
\end{align*}
Since the single-ton search may or may not return the correct index-value pair, this probability is obtained by the total law of probability as
\begin{align*}
	&\Prob{\frac{1}{P}\left\|\bdsb{y}-\widehat{x}[\widehat{k}]\mathbf{s}_{\widehat{k}}\right\|^2\geq (1+\gamma\SNRmin)\sigma^2}\\
	&=\Prob{\frac{1}{P}\left\|\bdsb{y}-\widehat{x}[\widehat{k}]\mathbf{s}_{\widehat{k}}\right\|^2\geq (1+\gamma\SNRmin)\sigma^2\Big|\widehat{k} \neq k \textrm{~or~} \widehat{x}[\widehat{k}]\neq x[k]}
	\times\Prob{\widehat{k} \neq k \textrm{~or~} \widehat{x}[\widehat{k}]\neq x[k] }\\
	&~~~+\Prob{\frac{1}{P}\left\|\bdsb{y}-\widehat{x}[\widehat{k}]\mathbf{s}_{\widehat{k}}\right\|^2\geq (1+\gamma\SNRmin)\sigma^2\Big|\widehat{k} = k\textrm{~and~}\widehat{x}[\widehat{k}] = x[k]}
	\times\Prob{\widehat{k} = k\textrm{~and~}\widehat{x}[\widehat{k}] = x[k] }\\
	&\leq \Prob{\widehat{k} \neq k \textrm{~or~} \widehat{x}[\widehat{k}]\neq x[k] }
	+ \Prob{\frac{1}{P}\left\|\bdsb{y}-\widehat{x}[\widehat{k}]\mathbf{s}_{\widehat{k}}\right\|^2\geq (1+\gamma\SNRmin)\sigma^2\Big|\widehat{k} = k\textrm{~and~}\widehat{x}[\widehat{k}] = x[k]}.
\end{align*}
Note that the second term is the probability of some noise samples $\mathbf{w}$ exceeding the single-ton verification threshold $(1+\gamma\SNRmin)\sigma^2$, which can be easily bounded by a chi-squared tail:
\begin{align}\label{eq:term1}
	\Prob{\frac{1}{P}\left\|\bdsb{y}-\widehat{x}[\widehat{k}]\mathbf{s}_{\widehat{k}}\right\|^2\geq (1+\gamma\SNRmin)\sigma^2\Big|\widehat{k} = k\textrm{~and~}\widehat{x}[\widehat{k}] = x[k]}
	\leq
	e^{-\frac{P}{4}\left(\sqrt{1+2\gamma\SNRmin}-1\right)^2}
\end{align}
Now we focus on obtaining a tail bound for the single-ton search error $\Prob{\widehat{k} \neq k \textrm{~or~} \widehat{x}[\widehat{k}]\neq x[k]}$. Since in the randomized design, we exploit a maximum likelihood estimator, the error probability can be obtained as:
\begin{align}\label{temp_eq6}
	\Prob{\widehat{k} \neq k \textrm{~or~} \widehat{x}[k]\neq x[k]}
	&\leq
	(N-1)\Prob{\left\|\bdsb{y} - \widehat{x}[k]\mathbf{s}_{\widehat{k}} \right\|^2<\left\|\bdsb{y} - x[k]\mathbf{s}_{k} \right\|^2},
\end{align}
where a union bound over all the $N-1$ codewords. Next, we bound the pair-wise error probability: 
\begin{align*}
	\Prob{\left\|\bdsb{y} - \widehat{x}[k]\mathbf{s}_{\widehat{k}} \right\|^2<\left\|\bdsb{y} - x[k]\mathbf{s}_{k} \right\|^2}
	&=
	\Prob{(x[k]\mathbf{s}_{k}^T-\widehat{x}[\widehat{k}]\mathbf{s}_{\widehat{k}}^T)\mathbf{w} < - \frac{\left\|x[k]\mathbf{s}_{k} - \widehat{x}[\widehat{k}]\mathbf{s}_{\widehat{k}}\right\|^2}{2}}\\
	&=\Prob{\mathcal{N}(0,1) > \frac{\left\|x[k]\mathbf{s}_{k} - \widehat{x}[\widehat{k}]\mathbf{s}_{\widehat{k}}\right\|}{2\sigma}}.
\end{align*}
Since $\mathbf{s}_k$ and $\mathbf{s}_{\widehat{k}}$ are also random, we calculate the above probability as follows:
\begin{align*}
	\Prob{\mathcal{N}(0,1) > \frac{\left\|x[k]\mathbf{s}_{k} - \widehat{x}[\widehat{k}]\mathbf{s}_{\widehat{k}}\right\|}{2\sigma}}
	&\leq
	\Prob{\mathcal{N}(0,1) > \frac{\left\|x[k]\mathbf{s}_{k} - \widehat{x}[\widehat{k}]\mathbf{s}_{\widehat{k}}\right\|}{2\sigma} \Big| \left\|x[k]\mathbf{s}_{k} - \widehat{x}[\widehat{k}]\mathbf{s}_{\widehat{k}}\right\|^2 \geq P\rho^2} \\
	&~~~~ + \Prob{\left\|x[k]\mathbf{s}_{k} - \widehat{x}[\widehat{k}]\mathbf{s}_{\widehat{k}}\right\|^2 < P\rho^2}\\
	&\leq
	2e^{-\frac{P\rho^2}{4\sigma^2}} + e^{-\frac{P}{16}} = 2e^{-\frac{\SNRmin}{4}\times P} + + e^{-\frac{1}{16}\times P},
\end{align*}
where we have used the fact that $x[k]\mathbf{s}_{k} - \widehat{x}[\widehat{k}]\mathbf{s}_{\widehat{k}}\sim\mathcal{N}(\mathbf{0},2\rho^2\mathbf{I})$. Together with \eqref{eq:term1}, the result follows.

%%%%%%%%%%%%%%%%%%
\section{Proof of Lemma \ref{lem_binary_sensing}}\label{proof_lem_binary_sensing}
%%%%%%%%%%%%%%%%%%
The analysis of the noisy design in Definition \ref{def_binary_sensing} is structurally similar to that of Lemma \ref{lem_random_sensing}, except that the bounding techniques are slightly different. In the following, we provide the false verification and missed verification rate for this design.

\begin{prop}[False Verification Rate]\label{prop_false_detection_binary}
For some constant $\gamma\in(0,1)$, the false verification rate can be upper bounded as
\begin{align*}%\label{rate_false_detection}
		&\Prob{\mathcal{H}_{\textrm{S}}(\widehat{k},\widehat{x}[\widehat{k}])\leftarrow\mathcal{H}_{\textrm{Z}}}< e^{-\frac{P}{4}\frac{(1-\gamma)^2\SNRmin^2}{1+2\SNRmin}}\\
		&\Prob{\mathcal{H}_{\textrm{S}}(\widehat{k},\widehat{x}[\widehat{k}])\leftarrow\mathcal{H}_{\textrm{M}}}<e^{-\frac{P}{4}\frac{(\frac{1}{2}-\gamma)^2\SNRmin^2}{1+\SNRmin}}+2e^{-c P}
\end{align*}
with some constant $c$.
\end{prop}
\begin{proof}
	See Appendix \ref{proof_prop_false_detection_binary}.
\end{proof}

%%%%%

\begin{prop}[Missed Verification Rate]\label{prop_miss_detection_binary}
For some constant $\gamma\in(0,1)$, the missed verification rate can be upper bounded as
\begin{align*}%\label{rate_miss_detection}
		&\Prob{\mathcal{H}_{\textrm{S}}(k,x[k])}< 
		e^{-\frac{P}{4}\left(\sqrt{1+2\gamma\SNRmin}-1\right)^2}+e^{-\zeta P}+ 2e^{-2\SNRmin P}
\end{align*}
for some constant $\zeta>0$ associated with the error exponent of the channel code $\mathbf{C}$ used in the single-ton search.
\end{prop}
\begin{proof}
	See Appendix \ref{proof_prop_miss_detection_binary}.
\end{proof}
Without loss of generality, let us choose $\gamma=1/4$ and thus all the error probabilities vanish at a rate $O(1/N^q)$ as long as $P \geq \alpha \log N$, where $\alpha$ satisfies:
\begin{align}\label{req_alpha}
	\begin{cases}
	\alpha \geq  \frac{64q}{9} \times \frac{1+2\SNRmin}{\SNRmin^2}\\
	\alpha \geq  \max\left\{64q\times \frac{1+\SNRmin}{\SNRmin^2}, \frac{q}{c} \right\}\\
	\alpha \geq \max\left\{\frac{4q}{\left(\sqrt{1+\SNRmin/2} - 1\right)^2}, \frac{q}{\zeta}, \frac{q}{2\SNRmin}\right\}
	\end{cases}.
\end{align}
Therefore, we have some sufficiently large constant $\alpha$ that satisfies all the above requirements. Since $K=O(N^\delta)$ for some $\delta\in(0,1)$, we have $P=(\alpha/(1-\delta))\log (N/K)$. Finally, letting $q=2\delta$ such that $O(1/N^q) = O(1/K^2)$, we have the claimed result. It can be seen that as $\SNRmin\rightarrow\infty$, the bottleneck in determining the error probability is the error exponent of the channel code $\zeta>0$, which approaches zero when the code rate approaches the channel capacity.

\subsection{Proof of False Verification Rates in Proposition \ref{prop_false_detection_binary}}\label{proof_prop_false_detection_binary}
The false verification events occur if the zero-ton or single-ton verifications fail when the ground truth is either a zero-ton or a multi-ton 
\begin{align}
	\bdsb{y} = \mathbf{S}\mathbf{z} +\mathbf{w}
\end{align}
with $\mathbf{z}$ being a zero-ton $\mathbf{z}=\mathbf{0}$ or a multi-ton $\left|\supp{\mathbf{z}}\right| > 1$.

\subsubsection{Detecting a Zero-ton as a Single-ton}
This event happens when a zero-ton $\bdsb{y}=\mathbf{w}$ passes the single-ton verification:
\begin{align}
	\Prob{\mathcal{H}_{\textrm{S}}(\widehat{k},\widehat{x}[\widehat{k}])\leftarrow\mathcal{H}_{\textrm{Z}}}
	=
	\Prob{\frac{1}{P}\left\| \bdsb{y} - \widehat{x}[k]\mathbf{s}_{\widehat{k}}\right\|^2 \leq (1+\gamma \SNRmin) \sigma^2}.
\end{align}
Since $\bdsb{y}=\mathbf{w}$ and $\|\widehat{x}[\widehat{k}]\mathbf{s}_{\widehat{k}}\|^2=P\rho^2$, it can be easily bounded by Lemma \ref{general_tail} as:
\begin{align}
	 \Prob{\mathcal{H}_{\textrm{S}}(\widehat{k},\widehat{x}[\widehat{k}])\leftarrow\mathcal{H}_{\textrm{Z}}}
	 \leq e^{-\frac{P}{4}\frac{(1-\gamma)^2\SNRmin^2}{1+2\SNRmin}}
\end{align}

\subsubsection{Detecting a Multi-ton as a Single-ton}
By definition, the error probability can be evaluated under the multi-ton model for some $L$-sparse vector $\mathbf{z}$:
\begin{align}
	\bdsb{y}  
	= \mathbf{S}\mathbf{z} + \mathbf{w}
\end{align}
when it passes the single-ton verification step for some index-value pair $(\widehat{k},\widehat{x}[\widehat{k}])$
\begin{align*}
	&\Prob{\mathcal{H}_{\textrm{S}}(\widehat{k},\widehat{x}[\widehat{k}])\leftarrow\mathcal{H}_{\textrm{M}}}
	=\Prob{\frac{1}{P}\left\|\bdsb{y}-\widehat{x}[\widehat{k}]\mathbf{s}_{\widehat{k}}\right\|^2 \leq  (1+\gamma \SNRmin)\sigma^2}
\end{align*}
for some $\widehat{k}$ and $\widehat{x}[\widehat{k}]$. Since $\mathbf{s}_{\widehat{k}}$ is not Gaussian, and thus we bound this probability with respect to $\mathbf{s}_{\widehat{k}}$ and $\mathbf{w}$ separately. Substituting $\bdsb{y}= \mathbf{S}\mathbf{z} + \mathbf{w}$ and replacing $\mathbf{u}= \mathbf{S}\mathbf{z} -\widehat{x}[\widehat{k}]\mathbf{s}_{\widehat{k}}$, we have:
\begin{align*}
	&\Prob{\frac{1}{P}\left\|\mathbf{u} + \mathbf{w}\right\|^2\leq (1+\gamma\SNRmin)\sigma^2}\\
	&=\Prob{\frac{1}{P}\left\|\mathbf{u} + \mathbf{w}\right\|^2\leq (1+\gamma\SNRmin)\sigma^2\Big| \frac{1}{P}\left\|\mathbf{u}\right\|^2\geq \frac{\SNRmin}{2} \sigma^2}
	\times\Prob{\frac{1}{P}\left\|\mathbf{u}\right\|^2\geq \frac{\SNRmin}{2} \sigma^2}\\
	&~~~+\Prob{\frac{1}{P}\left\|\bdsb{y}-\widehat{x}[\widehat{k}]\mathbf{s}_{\widehat{k}}\right\|^2\leq (1+\gamma\SNRmin)\sigma^2\Big|\frac{1}{P}\left\|\mathbf{u}\right\|^2\leq \frac{\SNRmin}{2} \sigma^2}
	\times\Prob{\frac{1}{P}\left\|\mathbf{u}\right\|^2\leq \frac{\SNRmin}{2} \sigma^2}\\
	&\leq \Prob{\frac{1}{P}\left\|\mathbf{u} + \mathbf{w}\right\|^2\leq (1+\gamma\SNRmin)\sigma^2\Big| \frac{1}{P}\left\|\mathbf{u}\right\|^2\geq \frac{\SNRmin}{2} \sigma^2} + \Prob{\frac{1}{P}\left\|\mathbf{u}\right\|^2\leq \frac{\SNRmin}{2} \sigma^2}.
\end{align*}
The first term can be bounded easily by Lemma \ref{general_tail} as:
\begin{align}
	 \Prob{\frac{1}{P}\left\|\mathbf{u} + \mathbf{w}\right\|^2\leq (1+\gamma\SNRmin)\sigma^2\Big| \frac{1}{P}\left\|\mathbf{u}\right\|^2\geq \frac{\SNRmin}{2} \sigma^2}
	 \leq e^{-\frac{P}{4}\frac{(\frac{1}{2}-\gamma)^2\SNRmin^2}{1+\SNRmin}}.
\end{align} 
Now it remains to bound $\Prob{\frac{1}{P}\left\|\mathbf{u}\right\|^2\leq \frac{\SNRmin}{2} \sigma^2}$, where $\mathbf{u}=\mathbf{S}\tilde{\mathbf{z}}$ and $\tilde{\mathbf{z}}=\mathbf{z}-\widehat{x}[\widehat{k}]\mathbf{e}_{\widehat{k}}$.

\begin{lem}\label{lem_boundedness}
Given $\bdsb{\phi}_p\defn\mathbf{S}_{(p,:)}^T$ and $\tilde{\mathbf{z}}$, the variable $\xi_p = |U[p]|^2 = |\bdsb{\phi}_p^T\tilde{\mathbf{z}}|^2$
is sub-exponential with mean $\bar{\xi}=\|\tilde{\mathbf{z}}\|^2$ and an Orlicz-norm (i.e. the $\psi_1$-norm of sub-exponential variables) for some absolute constant $c_5>0$
\begin{align}
	\xi_{\psi_1} = c_5 \bar{\xi}.
\end{align}
\end{lem}
\begin{proof}
Note that one can re-write the variable as $\xi_p=\bdsb{\phi}_p^H\mathbf{Q}\bdsb{\phi}_p$ with $\mathbf{Q}=\tilde{\mathbf{z}}^\ast\tilde{\mathbf{z}}^T$.
It is clear that $\xi_p$ is bounded and hence it is sub-exponential with mean
\begin{align}
	\bar{\xi} = \mathbb{E}\left[\bdsb{\phi}_p^H\mathbf{Q}\bdsb{\phi}_p\right] = \mathrm{Tr}(\mathbf{Q}) = \left\|\tilde{\mathbf{z}}\right\|^2.
\end{align}
To compute its Orlicz-norm, we only need to find the constant $\xi_{\psi_1}$ such that the following holds:
\begin{align*}
	\Prob{|\xi_p - \bar{\xi}|>t}< 2\exp\left(-\frac{t}{\xi_{\psi_1}}\right).
\end{align*}
Since $\left\|\mathbf{Q}\right\|_F=\left\|\tilde{\mathbf{z}}\right\|^2$, we can readily obtain the Orlicz-norm of the variable $\xi_{\psi_1} = c_5 \bar{\xi}$. Since $\bdsb{\phi}_p$ contains i.i.d. sub-gaussian variables, we can apply the Hanson-Wright inequality to obtain
\begin{align*}
	\Prob{|\xi_p - \bar{\xi}|>t}
	&=
	\Prob{\left|\bdsb{\phi}_p^H\mathbf{Q}\bdsb{\phi}_p-\mathbb{E}\left[\bdsb{\phi}_p^H\mathbf{Q}\bdsb{\phi}_p\right]\right|>t}
	\leq 2\exp\left(-\frac{t}{c_5\left\|\mathbf{Q}\right\|_F}\right)
\end{align*}
for some $c_5>0$. 
\end{proof}

By Lemma \ref{lem_boundedness}, the variable $\xi_p=|U[p]|^2$ is sub-exponential with mean $\bar{\xi}=\left\|\tilde{\mathbf{z}}\right\|^2$ and an Orlicz-norm $\xi_{\psi_1} = c_5\bar{\xi}$. Using the Bernstein-type inequality, then for any $t>0$ we have
\begin{align*}
	\Prob{ \left|\frac{1}{P} \sum_{p\in[P]}(\xi_p - \bar{\xi}) \right| \geq t}
	&\leq 2\exp\left(-c_6\frac{Pt}{\bar{\xi}}\right)
\end{align*}
for some constant $c_6$. By taking $t=\bar{\xi}-\SNRmin\sigma^2/2$, we have
\begin{align}
	\Prob{\frac{1}{P} \sum_{p\in[P]}(\xi_p - \bar{\xi}) \leq  - (\bar{\xi}-\SNRmin\sigma^2/2)}
	%&\leq 2\exp\left(-\frac{P}{2\bar{\xi}^2}(\bar{\xi}-2\gamma\sigma^2)^2\right)\\
	&\leq 2\exp\left(-c_6P\frac{(\bar{\xi}-\SNRmin\sigma^2/2)}{\bar{\xi}}\right)\\
	&=2\exp\left[-c_6P\left(1-\frac{\SNRmin\sigma^2}{2\bar{\xi}}\right)\right].
\end{align}
%Without loss of generality, let the multi-ton $\mathbf{z}$ be an $L$-ton for some $L\geq 2$ such that $1\leq \left|\supp{\tilde{\mathbf{z}}}\right|\leq L+1$. 
Since the probability is monotonically decreasing with respect to $\bar{\xi}$, we can substitute the minimum 
$\bar{\xi}= \left\|\tilde{\mathbf{z}}\right\|^2\geq 3\rho^2$ for any multi-ton into the above tail bound and obtain
\begin{align*}
	\Prob{\frac{1}{P} \sum_{p\in[P]}\xi_p \leq  \frac{\SNRmin}{2}\sigma^2}
	\leq 2e^{-c P}
\end{align*}
for some $c$.

%%%%%%%%%%%%%

\subsection{Proof of Missed Verification Rates in Proposition \ref{prop_miss_detection_binary}}\label{proof_prop_miss_detection_binary}

The missed verification events occur if the zero-ton or single-ton verifications pass when the ground truth is a single-ton $\mathcal{H}_{\textrm{S}}(k,x[k])$ for some $k\in[N]$:
\begin{align}
	\mathbf{u}_2 = \mathbf{S}\mathbf{z} +\mathbf{w} = x[k]\mathbf{s}_{k} + \mathbf{w}.
\end{align}

This event occurs when the ground truth is a single-ton $\mathcal{H}_{\textrm{S}}(k,x[k])$ with an index-value pair $(k,x[k])$, but the single-ton verification fails for some index-value pair $(\widehat{k},\widehat{x}[\widehat{k}])$ obtained from the single-ton search:
\begin{align*}
	\Prob{\mathcal{H}_{\textrm{S}}(k,x[k])}
	&=\Prob{\frac{1}{P}\left\|\mathbf{u}_2-\widehat{x}[\widehat{k}]\mathbf{s}_{\widehat{k}}\right\|^2\geq (1+\gamma\SNRmin)\sigma^2}.
\end{align*}
Since the single-ton search may or may not return the correct index-value pair, this probability is obtained by the total law of probability as
\begin{align*}
	&\Prob{\frac{1}{P}\left\|\mathbf{u}_2-\widehat{x}[\widehat{k}]\mathbf{s}_{\widehat{k}}\right\|^2\geq (1+\gamma\SNRmin)\sigma^2}\\
	&=\Prob{\frac{1}{P}\left\|\mathbf{u}_2-\widehat{x}[\widehat{k}]\mathbf{s}_{\widehat{k}}\right\|^2\geq (1+\gamma\SNRmin)\sigma^2\Big|\widehat{k} \neq k \textrm{~or~} \widehat{x}[\widehat{k}]\neq x[k]}
	\times\Prob{\widehat{k} \neq k \textrm{~or~} \widehat{x}[\widehat{k}]\neq x[k] }\\
	&~~~+\Prob{\frac{1}{P}\left\|\mathbf{u}_2-\widehat{x}[\widehat{k}]\mathbf{s}_{\widehat{k}}\right\|^2\geq (1+\gamma\SNRmin)\sigma^2\Big|\widehat{k} = k\textrm{~and~}\widehat{x}[\widehat{k}] = x[k]}
	\times\Prob{\widehat{k} = k\textrm{~and~}\widehat{x}[\widehat{k}] = x[k] }\\
	&\leq \Prob{\widehat{k} \neq k \textrm{~or~} \widehat{x}[\widehat{k}]\neq x[k] }
	+ \Prob{\frac{1}{P}\left\|\mathbf{u}_2-\widehat{x}[\widehat{k}]\mathbf{s}_{\widehat{k}}\right\|^2\geq (1+\gamma\SNRmin)\sigma^2\Big|\widehat{k} = k\textrm{~and~}\widehat{x}[\widehat{k}] = x[k]}.
\end{align*}
Note that the second term is the probability of some noise samples $\mathbf{w}$ exceeding the single-ton verification threshold $(1+\gamma\SNRmin)\sigma^2$, which can be easily bounded by a chi-squared tail:
\begin{align}\label{eq:term1}
	\Prob{\frac{1}{P}\left\|\mathbf{u}_2-\widehat{x}[\widehat{k}]\mathbf{s}_{\widehat{k}}\right\|^2\geq (1+\gamma\SNRmin)\sigma^2\Big|\widehat{k} = k\textrm{~and~}\widehat{x}[\widehat{k}] = x[k]}
	\leq
	e^{-\frac{P}{4}\left(\sqrt{1+2\gamma\SNRmin}-1\right)^2}
\end{align}

Now we focus on obtaining a tail bound for the single-ton search error $\Prob{\widehat{k} \neq k \textrm{~or~} \widehat{x}[\widehat{k}]\neq x[k]}$. Since we obtain the coefficient and index using $\mathbf{u}_0$ and $\mathbf{u}_1$ separately, the error probability can be obtained as:
\begin{align}\label{temp_eq6}
	\Prob{\widehat{k} \neq k \textrm{~or~} \widehat{x}[k]\neq x[k]}
	&\leq
	\Prob{\widehat{k} \neq k} + \Prob{\widehat{x}[k]\neq x[k]}.
\end{align}
Clearly, the first term is equivalent to the decoding error probability of the channel code $\mathbf{C}$, which is $\Prob{\widehat{k} \neq k} = e^{-\zeta P}$. On the other hand, given $\mathbf{u}_0=x[k] \mathbf{1}_P+\mathbf{w}_0$ and $x[k] \in \{\pm \rho\}$, the probability of wrongly estimating the coefficient when the bin is a single-ton can be upper bounded easily as
\begin{align}
	 \Prob{\widehat{x}[k]\neq x[k]}
	 &=
	 \Prob{\left\|\mathbf{u}_0 + x[k] \mathbf{1}_P \right\|^2 \leq \left\|\mathbf{u}_0 - x[k] \mathbf{1}_P \right\|^2}\\
	 &\leq
	 \Prob{\left\|2x[k] \mathbf{1}_P + \mathbf{w}_0 \right\|^2 \leq \left\|\mathbf{w}_0 \right\|^2}\\
	 %&=
	 %\Prob{2\rho^2P + x[k]\mathbf{1}^T\mathbf{w}_0 \leq 0}\\
	 &=
	 \Prob{\mathcal{N}(0, 1) \geq \frac{2\rho \sqrt{P}}{\sigma}}
	 \leq
	 2e^{-2\SNRmin P}.
\end{align}

%%%%%%%%%%%%
\section{Proof of Lemma~\ref{lem:code}}\label{sec:prf_lemma_code}
In this section, we prove Lemma \ref{lem:code}. The construction of the concatenated code in Lemma \ref{lem:code} is based on Justesen's concatenation scheme \cite{justesen1972class} and similar method is also analyzed in \cite{cheraghchi2009capacity}. 
The concatenated code consists of an outer code $f_{\text{out}}$ and an ensemble of inner codes $\mathcal{I}$. For the outer codes, we use an expander-based code proposed in \cite{spielman1995linear}. The outer code maps the message to a codeword with length $p$ on an alphabet with size $2^k$, i.e., $f_{\text{out}}:[N]\rightarrow [2^k]^p$. Recall that by definition, the rate of the outer code is $R_{\text{out}}=\lceil \log(N) \rceil/p$. We make essential use of the Theorem in \cite{spielman1995linear}.
\begin{thm}\label{thm:expander}
For every integer $k > 0$ and every absolute constant $R^\prime< 1$, there is an explicit family of expander-based linear codes with alphabet $[2^k]$ and rate $R_{\text{out}}=R^\prime$ that is error-correcting for a $O(1)$ fraction of errors. The running time of the encoder and the decoder is linear in the block length of the codewords.
\end{thm}
Note that here, the $O(1)$ fraction of error can be adversarially chosen, and that the decoding algorithm of the outer code does not rely on the knowledge of the channel. Now let $(c_1,c_2,\ldots,c_p)\in[2^k]^p$ be the codeword that we obtained from the outer code, and we call it the outer codeword. As we have mentioned, we use an ensemble of inner codes $\mathcal{I}$, which means that $\mathcal{I}=\{g_1,\ldots, g_p\}$ is a collection of $p$ codes which encode the symbols in the outer codeword as a new $q$-bit codeword with alphabet $\{1,-1\}$. Specifically, each code $g_i$ in $\mathcal{I}$ is a map $g_i: [2^k] \rightarrow \{1, -1\}^q$, and we encode the $i$-th symbol in the outer codeword by the $i$-th code in $\mathcal{I}$. This gives us the final codeword $(g_1(c_1),g_2(c_2),\ldots, g_p(c_p))\in\{1,-1\}^{qp}$, which also implies that the block length of the concatenated code is $P_0=qp$.

Then we show the details of the inner code ensemble. We choose the inner code ensemble to be the Wozencraft's ensemble \cite{massey1963threshold}. The  
Wozencraft's ensemble satisfies the property that all but a $o(1)$ fraction of the codes in the ensemble are capacity achieving, where the asymptotic is with respect to the block length $q$. Specifically, for the capacity achieving codes in the ensemble, the probability of decoding error is exponentially small in the block length $q$, i.e., $e^{-\alpha q}$ for some constant $\alpha>0$, as long as the rate of the codes $R_{\text{in}}=k/q$ is below the capacity of the BCS. Here, we should notice that we do need an upper bound of the bit flip probability in the design of the inner code since we need to get a lower bound of the capacity of the BSC, however, we do not need the exact value of the bit flip probability. Then, it is shown in \cite{cheraghchi2009capacity} that using brute force maximum likelihood decoder for the inner code and the decoding algorithm of the expander-based outer code, the error probability is exponentially small in the block length of the concatenated code, i.e., $e^{-\alpha^\prime P_0}$ for some constant $\alpha^\prime>0$. 

Now we analyze the block length and decoding complexity of the concatenated code. The number of codes in the Wozencraft's ensemble is $2^q$, meaning that $p=2^q$. Since rate of the outer code is a constant $R_{\text{out}}=\lceil \log(N) \rceil/p$ which can be arbitrarily close to 1, we know that $p=O(\log(N))$. Then $q=O(\log\log(N))$ and the block length of the concatenated code is $P_0=qp=O(\log(N)\log\log(N))$. Consequently the error probability is $e^{-\alpha^\prime P_0}=O(\frac{1}{\poly(N)})$, where $\poly(N)$ is a polynomial of $N$ which can have arbitrarily large degree. Consider the decoding complexity. For the inner code, the complexity of testing each possible message is $O(q)$ and there are $2^k=2^{qR_{\text{in}}}$ messages. Therefore, for each inner code, the computational complexity of the brute force maximum likelihood decoding is $O(2^{qR_{\text{in}}}q)$. Since there are $p$ inner codes, the complexity of decoding all the inner codes is $O(2^{qR_{\text{in}}}qp)=O(p^{1+R_{\text{in}}} q) = O(\log^{1+R_{\text{in}}}(N)\log\log(N))$. Since we do not require the inner code to be capacity achieving, $R_{\text{in}}$ can be arbitrarily close to $0$, we can conclude that complexity of decoding all the inner codes is $O(\log^{1+r}(N))$, where $r>0$ can be arbitrarily small. Since the complexity of decoding the outer code is linear in its block length, which is $O(p)=O(\log(N))$, we know that the decoding complexity of the concatenated code is $O(\log^{1+r}(N))$.

%%%%%%%%%%%
\section{Proof of Lemma \ref{lem:peeling}}\label{prf:peeling}
The proof of Lemma \ref{lem:peeling} is based on density evolution, and the basic idea is to get a recursive equation to analyze the fraction of sparse coefficients that are not recovered in a particular iteration. We provide a brief proof here and focus on the truncation peeling strategy, which is main difference from the previous results.

We do not consider the connection between the zero elements and the measurement bins, meaning that we only focus on the $d$-left regular random bipartite graph with $K$ left nodes and $R$ right nodes. We let $R=\eta K$ for some constant $\eta>0$. Using Poisson approximation, we get the expected fraction of edges which are connected to right nodes with degree $i$ is 
$$
\rho_i \approx \frac{(d/\eta)^{i-1}e^{-d/\eta}}{(i-1)!}.
$$
We then consider the peeling process as a message passing process on the bipartite graph. According to our peeling decoding algorithm, a single-ton can send a ``peeling'' message to a left node connected to it, and a peeled left node sends ``peeled'' message to all the bins that are connected to it. In a particular iteration, a bin sends a ``peeling'' message to a left node through an edge if other edges connected to this bin all send ``peeled'' messages in the previous iteration and a left node sends a ``peeled'' message to a bin through an edge if if at least one of the bins that is connected to it sends a ``peeling'' message to it. We should also notice that the bins with degree greater than $D$ never send ``peeling'' message to the left nodes due to the truncation strategy. 

As in previous proofs, we still need to first assume that the neighborhood of each edge with a constant depth is a tree (tree-like assumption). Let $p_j$ be the probability that in the $j$-th iteration, a randomly chosen edge is \emph{not} peeled, i.e., sending a ``not peeled'' message. Then, under the tree-like assumption, we have the density evolution equation:
$$
p_{j+1} =F(p_j)= \left(1-\sum_{i=1}^D \rho_i(1-p_j)^{i-1} \right)^{d-1}.
$$
Similar to the analysis in~\cite{pedarsani2017phasecode}, we need to consider the fix point of $F(t)$, i.e., the point such that $F(t) = t$, and show that the fix point can be arbitrarily small by choosing proper parameters. We have
\begin{align}
F(t) &=  \left( 1-\sum_{i=1}^D \frac{(d/\eta(1-t))^{i-1}e^{-d/\eta}}{(i-1)!} \right)^{d-1}  \nonumber \\
& = \left( 1-\sum_{i=0}^{D-1} \frac{(d/\eta (1-t))^i e^{-d/\eta}}{i!} \right)^{d-1}  \nonumber \\
& = \left( 1- e^{-d/\eta}(e^{d(1-t)/\eta} - \frac{e^\xi(d(1-t)/\eta)^D}{D!}) \right)^{d-1}, \nonumber
\end{align}
where $0<\xi<d(1-t)/\eta$. We can choose $D$ to be large enough such that $\frac{(d(1-t)/\eta)^D}{D!}<\frac{1}{2}$. Then we have
$$
F(t) < \left( 1-\frac{1}{2}e^{-dt/\eta} \right)^{d-1} := G(t).
$$
Then we know that the fix point of $F(t)$ should be upper bounded by that of $G(t)$. Further, if we keep $d/\eta$ to be a constant and enlarge $d$, the fix point of $G(t)$ can be arbitrarily small, and consequently, the fix point of $F(t)$ can be arbitrarily small. More specifically, let $p^\star\in(0,1)$ be the fix point of $F(t)$, then for any $p>0$, there exist parameters $d$ and $\eta$ such that $p^\star<p$. Here, we briefly analyze the relationship between $\eta$ and the fix point of $G(t)$, denoted by $t^\star$. Since $t^\star=G(t^\star)=( 1-\frac{1}{2}e^{-dt^\star/\eta} )^{d-1}$, and $t^\star$ is close to 0, we have $e^{-dt^\star/\eta}\approx 1$ and thus $t^\star\approx (\frac{1}{2})^{d-1}$. Therefore, $d=O(\log(1/t^\star))$, and further, since we keep $d/\eta$ as a constant, $\eta=O(\log(1/t^\star))$. Since the fix point of $F(t)$, $p^\star$ is upper bounded by $t^\star$, we have $\eta=O(\log(1/p^\star))$. We can choose parameters such that $p^\star=O(p)$ and then, $\eta=O(\log(1/p))$.
Using the same argument as in \cite{pedarsani2017phasecode}, we can show that for any $p>0$, there exist a constant $n$ and proper parameters $d$ and $\eta$ such that $p_n<p$.

By the same martingale argument as in previous analysis, and taking the event that the tree-like assumption does not hold into consideration, we can show that the fraction of sparse coefficients which are not peeled is highly concentrated around $p_n$. Let $Z$ be the fraction of sparse coefficients which are not peeled after $n$-th iteration, when $K$ is large enough, we have for any $\delta>0$,
$$
\Prob{ | Z-p_n | >\delta}<2\exp\{-C\delta^2 K^{1/(4n+1)}\},
$$
where $C>0$ is a universal constant. The proof of Lemma \ref{lem:peeling} is completed by choosing $n$ such that $p_n<p$.

%%%%%%%%%%%
\section{Proof of Lemma \ref{lem:estimation}}\label{prf:estimate}
Without loss of generality, we omit the bin index, but we still keep an iteration counter in the notation. More specifically, we use $\vect{u}_0^{(t)}$ and $\vect{u}_1^{(t)}$ to denote the \emph{remaining} location and verification measurements in a particular bin (say bin $i$) at the $t$-th iteration, respectively. We also use $\vect{z}$ to denote the signal that has actual contribution to the measurements, and $\vect{w}_0$ and $\vecw_1$ to denote the noise in the location and verification measurements, respectively.

Consider $t=1$. In the first iteration, we know that $\vect{u}_0^{(1)}$ and $\vect{u}_1^{(1)}$ are exactly the original measurements, i.e., 
\begin{align*}
\vect{u}_0^{(1)} &= \mat{S}_0\vect{z} + \vecw_0 \\
\vect{u}_1^{(1)} &= \mat{S}_1\vect{z} + \vecw_1.
\end{align*}
In Lemma~\ref{lem:code}, we have shown that if a bin is indeed a single-ton and the sparse coefficient is located at $j$, the concatenated code that we designed in the location matrix can find the location\footnote{As we have mentioned, due to sign ambiguity, the decoding algorithm can return up to two locations, but one of them is guaranteed to be $j$ with high probability.}  $j$ with probability $1-O(1/\poly(N))$. According to the estimation method in (\ref{eq:estimate}), we have
$$
\widehat{x}[j]=\frac{1}{P_1} \sum_{k=1}^{P_1} s_{1,k,j} u_{1,k}^{(1)} = \frac{1}{P_1} \sum_{k=1}^{P_1} x[j]+\tilde{w}_{1,k},
$$
where $\tilde{w}_{1,k} = s_{1,k,j}w_{1,k}$. Then we have $\tilde{\vect{w}}_1=\{ s_{1,k,j}w_{1,k} \}_{k=1}^{P_1}\sim\mathcal{N}(\vect{0}, \sigma^2\mat{I})$. By Chernoff bound, we have for any $\epsilon>0$,
$$
\Prob{\ABSL{\widehat{x}[j] - x[j]} \ge \epsilon} \le \exp \{ -\frac{P_1\epsilon^2}{2\sigma^2} \},
$$
therefore, by choosing $P_1=O(\frac{\sigma^2}{\epsilon^2}\log(N))$, we can get $\widehat{x}[j]$ such that $\ABSL{\widehat{x}[j] - x[j]} < \epsilon$ with probability $1-O(1/\poly(N))$. 

Consider the $t$-th iteration, $t>1$. Since $t>1$, bin $i$ is not a single-ton bin in the first iteration, and thus, $\ABSL{\supp{\vect{z}}}>1$. Let $\mathcal{B}=\supp{\vect{z}}\setminus\{j\}$, i.e., $\mathcal{B}$ is the set of location indices of the sparse coefficients which are peeled off from bin $i$ before the $t$-th iteration. According to the truncation peeling strategy, we have $\ABSL{\mathcal{B}}\le D-1$. Assume that for any $g\in\mathcal{B}$, we have $\ABSL{\widehat{x}[g]-x[g]}<C_g\epsilon$ with probability $1-O(1/\poly(N))$ for some constant $C_g>0$. We let $C_B=\sum_{g\in\mathcal{B}}C_g$.

Now we show that if there exists appropriate constant $C$ such that $\beta\ge C\epsilon$, then the decoding algorithm can find the location of the sparse coefficient at $j$ with probability $1-O(1/\poly(N))$. Recall that to conduct the decoding algorithm of the concatenated code, we need to take the sign of the remaining location measurements, i.e., getting $\mathsf{sgn}[\vect{u}_0^{(t)}]$. According to the peeling algorithm, we have 
$$
\vect{u}_0^{(t)} = \vect{u}_0^{(1)} - \sum_{g\in\mathcal{B}} \widehat{x}[g]\vect{s}_{0,g},
$$
which yields
$$
\vect{u}_0^{(t)} =\sum_{g\in\mathcal{B}}(x[g]-\widehat{x}[g])\vect{s}_{0,g}+x[j]\vect{s}_{0,j}+\vect{w}_0.
$$
Let $\tilde{\vect{s}} = \sum_{g\in\mathcal{B}}(x[g]-\widehat{x}[g])\vect{s}_{0,g}+x[j]\vect{s}_{0,j}$. We assume that $\beta$ is large enough such that in any constant iteration $\beta\ge 2C_B\epsilon$. Then, for each entry in $\tilde{\vect{s}}$,  we have $\SGN{\tilde{s}_k}=\SGN{x[j]s_{0,k,j}}$, and we can think of $\SGN{u_{0,k}^{(t)}}$ as a received symbol by transmitting $\SGN{x[j]s_{0,k,j}}$ through a BSC with bit flip probability upper bounded by $\Phi(-\frac{\beta}{2\sigma})$. Then, the decoding algorithm of the concatenated code still works since we have a constant upper bound of the bit flip probability. 

Then, we show that the value estimation method still works in the $t$-th iteration.
Since
$$
s_{1, k, j}u_{1,k}^{(t)} = x[j]+\sum_{g\in\mathcal{B}}(x[g]-\widehat{x}[g])s_{1, k,j}s_{1,k,g}+s_{1, k,j}w_{1,k}^{(t)},
$$
and $ \widehat{x}[j]=\frac{1}{P_1}\sum_{k=1}^{P_1} s_{1, k, j}u_{1,k}^{(t)} $, we know that conditioned on $\mat{S}_{1}$, $\widehat{x}[g]$ and the event that $\ABSL{\widehat{x}[g]-x[g]}<C_g \epsilon$ for all $g\in\mathcal{B}$, $\widehat{x}[j]\sim\mathcal{N}(x[j]+\bar{x}, \frac{\sigma^2}{P_1})$, where $\bar{x}=\frac{1}{P_1}\sum_{k=1}^{P_1}\sum_{g\in\mathcal{B}}(x[g]-\widehat{x}[g])s_{1,k,j}s_{1,k,g}$. We can see that $\ABSL{\bar{x}}<C_B\epsilon$, and by Chernoff bound,
\begin{equation}\label{eq:bound}
\Prob{\ABSL{\widehat{x}[j]-x[j]-\bar{x}}\ge \epsilon \bigm| \mat{S}_1, \ABSL{\widehat{x}[g]-x[g]}<C_g\epsilon }\le \exp\{-\frac{P_1\epsilon^2}{2\sigma^2}\}.
\end{equation}
Since (\ref{eq:bound}) is true for all $\mat{S}_1$, we can remove the condition on $\mat{S}_1$. Considering the fact that $\ABSL{\bar{x}}<C_B\epsilon$, we get
$$
\Prob{ | \widehat{x}[j]-x[j] | \ge (C_B+1)\epsilon \bigm| \ABSL{\widehat{x}[g]-x[g]}<C_g\epsilon} \le \exp\{-\frac{P_1\epsilon^2}{2\sigma^2}\}.
$$
Then, by law of total probability and union bound, we get
$$
\Prob{\ABSL{\widehat{x}[j]-x[j]}\ge (C_B+1)\epsilon} \le \exp\{-\frac{P_1\epsilon^2}{2\sigma^2}\} + \sum_{g\in\mathcal{B}} \Prob{\ABSL{\widehat{x}[g]-x[g]}\ge C_g\epsilon}\le O(\frac{1}{\poly(N)}),
$$
when $P_1=O(\frac{\sigma^2}{\epsilon^2}\log(N))$, which completes the proof.

%%%%%%%%%%%
\section{Proof of Lemma \ref{lem:energy}}\label{prf:energy}
We make essential use of the Johnson-Lindenstrauss Lemma~\cite{johnson1984extensions}; more specifically, we use the form stated in~\cite{baraniuk2008simple}.
\begin{lem}\label{lem:concentration}
~\cite{baraniuk2008simple} Let $\mat{S}_1\in\{-1,1\}^{P_1\times N}$ be a matrix with i.i.d. Rademacher entries. For any $\theta\in(0,1)$ and any $\vect{v}\in\R^{N}$, we have
$$
\Prob{\ABS{\frac{1}{P_1}\TWONL{\mat{S}_1\vect{v}}^2 - \TWONL{\vect{v}}^2} \ge \theta\TWONL{\vect{v}}^2 } \le 2\exp\{-P_1(\frac{\theta^2}{4} - \frac{\theta^3}{6})\}.
$$
\end{lem}

In the following, we omit the bin index and iteration counter, and let $\vect{u}_1$ be the actual verification measurements of bin $i$ and $\vect{w}_1$ be the corresponding noise. Let $\vect{z}$ be the signal that has actual contribution to the measurements in this bin, i.e.,
$$
\vect{u}_1=\mat{S}_1\vect{z} + \vect{w}_1,
$$
and $\widehat{\vect{z}}$ be the hypothesis signal. Then, we have $\vect{u}_1-\mat{S}_1\widehat{\vect{z}} = \mat{S}_1\tilde{\vect{z}}+\vect{w}_1$, where $\tilde{\vect{z}}=\vect{z}-\widehat{\vect{z}}$. By Lemma~\ref{lem:concentration}, we have
$$
\Prob{\sqrt{1-\theta} \TWONL{\tilde{\vect{z}}}\le \frac{1}{\sqrt{P_1}}\TWONL{\mat{S}_1\tilde{\vect{z}}} \le \sqrt{1+\theta}\TWONL{\tilde{\vect{z}}}} \ge 1-2\exp\{-P_1(\frac{\theta^2}{4} - \frac{\theta^3}{6})\}.
$$
By triangle inequality, $\TWONL{\mat{S}_1 \tilde{\vect{z}} } - \TWONL{\vect{w}_1} \le \TWONL{\vect{u}_1 - \mat{S}_1\tilde{\vect{z}} } \le \TWONL{\mat{S}_1 \tilde{\vect{z}} } + \TWONL{\vect{w}_1}$.

Then, on the one hand, we have
$$
\Prob{\frac{1}{\sqrt{P_1}}\TWONL{\vect{u}_1 - \mat{S}_1\tilde{\vect{z}}} \ge \sqrt{1-\theta}\TWONL{\tilde{\vect{z}}} - \frac{1}{\sqrt{P_1}} \TWONL{\vect{w}_1} }  \ge 1-2\exp\{-P_1(\frac{\theta^2}{4} - \frac{\theta^3}{6})\}.
$$
By the concentration inequality of $\chi^2$ distribution, for any $\phi\in(0,3)$, we have
$$
\Prob{\frac{1}{P_1}\TWONL{\vect{w}_1}^2 \ge \sigma^2(1+\phi)} \\
\le \exp\{-P_1\frac{\phi^2}{18}\}.
$$
By union bound, we get
$$
\Prob{\frac{1}{\sqrt{P_1}}\TWONL{\vect{u}_1 - \mat{S}_1\tilde{\vect{z}}} \ge \sqrt{1-\theta}\TWONL{\tilde{\vect{z}}} - \sigma\sqrt{1+\phi}} \ge 1-2\exp\{-P_1(\frac{\theta^2}{4} - \frac{\theta^3}{6})\}-\exp\{-P_1\frac{\phi^2}{18}\}.
$$
Suppose that the supports of the hypothesis signal and the true signal are different, i.e., $\supp{\widehat{\vect{z}}}\neq \supp{\vect{z}}$, then by our assumption of the signal, $\|\tilde{\vect{z}}\|_2 \ge \sqrt{\beta}$. If $\sqrt{1-\theta}\sqrt{\beta}- \sigma\sqrt{1+\phi}>0$, we can get a valid threshold, which means that if $\beta>\sigma^2(\frac{1+\phi}{1-\theta})$, when $P_1=O(\log(N))$, 
\begin{equation}\label{eq:test1}
\Prob{\frac{1}{P_1}\TWONL{\vect{u}_1 - \mat{S}_1\tilde{\vect{z}}}^2 \ge \tau} \ge 1-O(\frac{1}{\poly(N)}),
\end{equation}
for any $\tau\in(0, (\sqrt{1-\theta}\sqrt{\beta}- \sigma\sqrt{1+\phi})^2)$.

On the other hand, we also have
$$
\Prob{\frac{1}{\sqrt{P_1}}\TWONL{\vect{u}_1 - \mat{S}_1\tilde{\vect{z}}}\le \sqrt{1+\theta}\TWONL{\tilde{\vect{z}}} + \frac{1}{\sqrt{P_1}}\TWONL{\vect{w}_1}} \ge 1-2\exp\{-P_1(\frac{\theta^2}{4} - \frac{\theta^3}{6})\}.
$$
Consider the case when $\supp{\widehat{\vect{z}}}=\supp{\vect{z}}$. In this case, we have found the correct support, or equivalently, all the locations of the singleton balls are found. By Lemma \ref{lem:estimation}, we know that $\|\tilde{\vect{z}}\|_\infty<\tilde{C}\epsilon$ for some constant $\tilde{C}$ with probability $1-O(1/\poly(N))$, when $P_1=O(\frac{\sigma^2}{\epsilon^2}\log(N))$. According to the truncation strategy, we also have $\ABSL{ \supp{\tilde{\vect{z}}} } \le D$, and thus $\TWONL{\tilde{\vect{z}}} \le \sqrt{D}\tilde{C}\epsilon:=C^\prime\epsilon$. Using this fact and union bound, we get
$$
\Prob{\frac{1}{\sqrt{P_1}}\TWONL{\vect{u}_1 - \mat{S}_1\tilde{\vect{z}}} \le \sqrt{1+\theta}C^\prime\epsilon+\sigma\sqrt{1+\phi}} 
\ge 1-O(\frac{1}{\poly(N)}),
$$
and thus, for any $\tau>(\sqrt{1+\theta}C^\prime\epsilon+\sigma\sqrt{1+\phi})^2$,
\begin{equation}\label{eq:test2}
\Prob{\frac{1}{P_1}\TWONL{\vect{u}_1 - \mat{S}_1\tilde{\vect{z}}} \le \tau} \ge 1-O(\frac{1}{\text{poly}(N)}).
\end{equation}
We can see that to get a valid threshold for both tests (\ref{eq:test1}) and (\ref{eq:test2}), we need 
$$
\sqrt{1-\theta}\sqrt{\beta}- \sigma\sqrt{1+\phi}>\sqrt{1+\theta}C^\prime\epsilon+\sigma\sqrt{1+\phi},
$$
and since $\theta$ and $\phi$ are constants, the proof is completed.

%%%%%%%%%%%
\section{Proof of Theorem~\ref{thm_continuous_recovery}}\label{prf:continuous_recovery}
We provide the brief final proof of Theorem~\ref{thm_continuous_recovery}. First, we analyze the error probability. There are three possible error events, 
\begin{itemize}
\item[(i)] $E_1$: the peeling algorithm does not find at least $1-p$ fraction of sparse coefficients.
\item[(ii)] $E_2$: error in decoding algorithm of concatenated code (location decoding).
\item[(iii)] $E_3$: error in value estimation or energy test.
\end{itemize}
Here, by error in value estimation, we mean there exists a sparse coefficient $x[j]$ and its estimate $\widehat{x}[j]$ such that $|x[j] - \widehat{x}[j] | \ge O(\epsilon)$.
We have shown that $\Prob{E_1|E_2^c,E_3^c}=O(\exp\{-c_1(p)K^{-c_2(p)}\})$. Since we need to conduct $O(K)$ times of location decoding and energy tests, using union bound, we know that $\Prob{E_2}=O(1/{\poly(N)})$ and $\Prob{E_3}=O(1/\poly(N))$. Then by union bound and law of total probability, we get the error probability 
\begin{align*}
\Prob{E_1 \cup E_2 \cup E_3} &\le \Prob{E_1} + \Prob{E_2} + \Prob{E_3} \\
& = \Prob{E_1|E_2^c,E_3^c} \Prob{E_2^c,E_3^c} + \Prob{E_1|E_2\cup E_3} \Prob{E_2\cup E_3} + \Prob{E_2} + \Prob{E_3} \\
& \le \Prob{E_1|E_2^c,E_3^c} + 2(\Prob{E_2} + \Prob{E_3} ) \\
& \le O(\exp\{-c_1(p)K^{-c_2(p)}\}) + O(1/{\poly(N)}) \\
& = O(1/{\poly(N)}),
\end{align*}
where the last inequality is due to the fact that $K=O(N^\delta)$ for some constant $\delta\in(0,1)$. The time complexity of the algorithm can be analyzed by the same method as in the quantized alphabet setting, and we omit the analysis here. 

Then, we turn to the $\ell_1$ norm recovery guarantee. Let $|x_{(1)}|, |x_{(2)}|, \ldots, |x_{(K)}|$ be the magnitudes of the $K$ sparse coefficients, ordered increasingly. Recall that we assume $| x_{(K)} | \le O(K^c)$ for some $c\in(0,1)$. Partition the $K$ sparse coefficients to $g = K^{(1+c)/2} $ subgroups as follows\footnote{Here, we simply assume that $K$ is an integer multiple of $g$.}:
$$
(|x_{(1)}|,\ldots, |x_{(K/g)}|), (|x_{(K/g+1)}|, \ldots, x_{(2K/g)}), \ldots, (|x_{(K-K/g+1)}|, \ldots, |x_{(K)}|).
$$
Let $b_i$ be the largest number in subgroup $i$. By Hoeffding's inequality, the probability that more than $(p+t)K/g$ elements are missed in a subgroup is upper bounded by $2e^{-2t^2K/g}$. Taking $t = 1/\log(K)$ and using union bound, we have
\begin{equation}\label{eq:ell1eq1}
\|\widehat{\vect{x}} - \vect{x} \|_1 \le \sum_{i=1}^g [ b_i (p+1/\log(K)) K/g + O(K\epsilon /g) ] = O(K\epsilon) + \sum_{i=1}^g b_i (p+1/\log(K)) K/g,
\end{equation}
with probability $1-O(ge^{-\frac{2K}{g\log^2(K)}} + \frac{1}{\poly(N)})$. Further,
\begin{equation}\label{eq:ell1eq2}
\begin{aligned}
\sum_{i=1}^g b_i K/g & \le (|x_{(1)}| + \sum_{i=1}^g b_i) K/g \\
&\le  \|\vecx\|_1 + b_g K/g \\
&\le  \|\vecx\|_1(1+O(\frac{K^c}{g})) \\
&= \|\vecx\|_1 (1+O(K^{-\frac{1-c}{2}})).
\end{aligned}
\end{equation}
Then, combining~\eqref{eq:ell1eq1} and~\eqref{eq:ell1eq2}, we can see that with probability at least $1-O(K^{\frac{1+c}{2}}e^{-\frac{2K^{(1-\gamma)/2}}{\log^2(K)}} + \frac{1}{\poly(N)})$,
$$
\|\widehat{\vect{x}} - \vect{x} \|_1 \le \|\vecx\|_1 (p+1/\log(K)) (1+O(K^{-\frac{1-c}{2}})) + O(K\epsilon) = p\|\vecx\|_1(1+o(1)) + O(K\epsilon).
$$
Since $K = O(N^\delta)$, $\frac{1}{\poly(N)}$ is the dominant term in the error probability.
In addition, since $\|\vecx\|_1 \ge K\beta$, we obtain
$$
\| \widehat{\vect{x}} - \vect{x} \|_1 \le p\|\vecx\|_1(1+o(1)) + O(\frac{\epsilon}{\beta} \| \vecx \|_1) := \kappa \| \vecx \|_1.
$$
Here, $\kappa$ can be arbitrarily small since $p$ and $\epsilon$ can be arbitrarily small. Thus, we conclude that with probability at least $1-O(\frac{1}{\poly(N)})$, we have $\|\widehat{\vect{x}} - \vect{x} \|_1 \le \kappa \|\vecx\|_1$.

%%%%%%%%%%%
\section{Tail Bounds}\label{sec:general_tail_chi}
Here we derive some tail bounds that are useful in our analysis.

\begin{lem}[Non-central Chi-Square Tail Bounds in \cite{birge2001alternative}]\label{lem_general_tail_chi}
Let $Z\sim\chi_D^2$ be a non-central chi square variable with $D$ degrees of freedom and non-centrality parameter $\nu\geq 0$. Then for all $z\geq 0$, the following tail bounds hold:
\begin{align*}
	&\Prob{Z \geq (D+\nu)+2\sqrt{(D+2\nu) z} + 2z} \leq \exp(-z)\\
	&\Prob{Z \leq (D+\nu)-2\sqrt{(D+2\nu) z}} \leq \exp(-z)	
\end{align*}
\end{lem}

\begin{lem}\label{general_tail}
Given $\mathbf{u}=[u[0],\cdots,u[P-1]]^T$ and a vector $\mathbf{w}=[w[0],\cdots,w[P-1]]^T$ with i.i.d. Gaussian variables $w[p]\sim\mathcal{N}(0,\theta^2)$ for all $p\in[P]$, the following tail bound holds:
\begin{align}\label{eq:general_tail}
	&\Prob{\frac{1}{P}\left\|\mathbf{u}+\mathbf{w}\right\|^2 \geq \tau_1} \leq e^{-\frac{P}{4}\left(\sqrt{2\tau_1/\theta^2-1} - \sqrt{1+2\nu_0}\right)^2}\\
	&\Prob{\frac{1}{P}\left\|\mathbf{u}+\mathbf{w}\right\|^2 \leq \tau_2} \leq e^{-\frac{P}{4}\frac{\left(1+\nu_0 - \tau_2/\theta^2\right)^2}{1+2\nu_0}}
\end{align}
for any $\tau_1$ and $\tau_2$ that satisfy
\begin{align}\label{tau_requirement}
	\tau_1 &\geq \theta^2(1+\nu_0),\quad
	\tau_2 \leq \theta^2(1+\nu_0),
\end{align}
where $\nu_0$ is the {\it normalized non-centrality parameter} given by
\begin{align}\label{nu0_def}
	\nu_0 \defn \frac{\left\|\mathbf{u}\right\|^2}{P\theta^2}.
\end{align}	
\end{lem}

\begin{proof}
The quantity $\left\|\mathbf{u}+\mathbf{w}\right\|^2$ can be written element-wise as
\begin{align}
	\left\|\mathbf{u}+\mathbf{w}\right\|^2 
	=
	\sum_{p=0}^{P-1} \left(u[p]+w[p]\right)^2 
\end{align}
where each summand is a normal random variable with mean $u[p]$ and variance $\theta^2$. Therefore, according to the definition of non-central chi-square variables, the quantity
\begin{align}
	\frac{\left\|\mathbf{u}+\mathbf{w}\right\|^2}{\theta^2} \sim \chi_P^2
\end{align}
is a non-central $\chi^2$ random variable of $P$ degrees of freedom with a non-centrality parameter 
\begin{align}\label{nu_def}
	\nu = \sum_{p=0}^{P-1}\frac{|u[p]|^2}{\theta^2} = \frac{\left\|\mathbf{u}\right\|^2}{\theta^2}.
\end{align}
For notational convenience, we use the normalized non-centrality parameter $\nu_0$ in \eqref{nu0_def} such that $\nu = P\nu_0$. Without loss of generality, let the thresholds $\tau_1$ and $\tau_2$ take the following form with respect to $z_1$ and $z_2$:
\begin{align*}
	\tau_1 &= \frac{\theta^2}{P}\left[(P+P\nu_0)+2\sqrt{(P+2P\nu_0) z_1} + 2z_1\right]\\
	\tau_2 &= \frac{\theta^2}{P}\left[(P+P\nu_0)-2\sqrt{(P+2P\nu_0) z_2}\right],
\end{align*}
then the tail bounds in Lemma \ref{lem_general_tail_chi} can be obtained easily with respect to $z_1$ and $z_2$. Using \eqref{nu_def}, the corresponding $z_1$ and $z_2$ can be solved as
\begin{align*}
	z_1 &= \frac{P}{4}\left(\sqrt{2\tau_1/\theta^2-1} - \sqrt{1+2\nu_0}\right)^2\\
	z_2 &= \frac{P}{4}\frac{\left(1+\nu_0 - \tau_2/\theta^2\right)^2}{1+2\nu_0}
\end{align*}
as long as the thresholds $\tau_1$ and $\tau_2$ satisfy \eqref{tau_requirement}. Thus according to Lemma \ref{lem_general_tail_chi}, we have the tail bounds in \eqref{eq:general_tail}.
\end{proof}

\begin{cor}\label{general_tail_worst}
Suppose that the normalized non-centrality parameter $\nu_0$ in Lemma \ref{general_tail} is bounded between
\begin{align}
	0\leq \nu_{\min} \leq \nu_0\leq \nu_{\max},
\end{align}
then the following worst case tail bounds hold:
\begin{align*}%\label{eq:general_tail_worst}
	&\Prob{\frac{1}{P}\left\|\mathbf{u}+\mathbf{w}\right\|^2 \geq \tau_1} \leq e^{-\frac{P}{4}\left(\sqrt{2\tau_1/\theta^2-1} - \sqrt{1+2\nu_{\max}}\right)^2}\\
	&\Prob{\frac{1}{P}\left\|\mathbf{u}+\mathbf{w}\right\|^2 \leq \tau_2} \leq e^{-\frac{P}{4}\frac{\left(1+\nu_{\min} - \tau_2/\theta^2\right)^2}{1+2\nu_{\min}}}
\end{align*}
for any $\tau_1$ and $\tau_2$ that satisfy
\begin{align}\label{tau_requirement2}
	\tau_1 &\geq \theta^2(1+\nu_{\max}),\quad
	\tau_2 \leq \theta^2(1+\nu_{\min}).
\end{align}
\end{cor}
\begin{proof}
The first tail bound can be easily obtained since $\tau_1 \geq \theta^2(1+\nu_{\max})$, the exponent is monotonically decreasing with respect to $\nu_0$, and therefore substituting it with $\nu_{\max}$ leads to an upper bound. 

The second tail bound depends on the monotonicity with respect to $\nu_0$. The tail bound is monotonic with respect to the exponent, so in the following we examine the monotonicity of the exponent with respect to $\nu_0$. The exponent can be re-written as a form of the $x+1/x$ function:
\begin{align}
	\frac{\left(1+\nu_0 - \tau_2/\theta^2\right)^2}{1+2\nu_0}
	&=\left(\nu_0 + \frac{1}{2}\right) + \frac{\left(\frac{1}{2}-\frac{\tau_2}{\theta^2}\right)^2}{\left(\nu_0 + \frac{1}{2}\right)} + 2\left(\frac{1}{2}-\frac{\tau_2}{\theta^2}\right),
\end{align}
which has a minimum at
\begin{align}
	\nu_0^\star = \left|\frac{1}{2}-\frac{\tau_2}{\theta^2}\right| - \frac{1}{2},
\end{align}
and monotonically increasing for any $\nu_0>\nu_0^\star$. Now it remains to see whether $\nu_0^\star$ is within the interval $[\nu_{\min},\nu_{\max}]$, which needs to be discussed separately depending on the choice of $\tau_2$:
\begin{enumerate}
	\item $\theta^2/2\leq \tau_2\leq \theta^2(1+\nu_{\min})$: in this case, we have
	\begin{align}
		\nu_0^\star = \frac{\tau_2}{\theta^2} - 1\leq \nu_{\min}.
	\end{align}	
	\item $0 < \tau_2 < \theta^2/2$: in this case, we have
	\begin{align}
		\nu_0^\star = - \frac{\tau_2}{\theta^2} \leq 0 \leq \nu_{\min}.
	\end{align}		
\end{enumerate}
Therefore, it has been shown that as long as $\tau_2$ satisfies \eqref{tau_requirement2}, the exponent is monotonically increasing with respect to $\nu_0\in[\nu_{\min},\nu_{\max}]$ and therefore the minimum exponent is achieved by substituting $\nu_0$ with $\nu_{\min}$.
\end{proof}

\end{document}